\documentclass{article}
\usepackage{geometry}
\usepackage{amsmath}
\numberwithin{equation}{section}

\usepackage{amssymb}

\usepackage{amsthm}
\usepackage[utf8]{inputenc}

\usepackage[numbers]{natbib}
\usepackage{tikz}
\usetikzlibrary{positioning}
\usepackage{pgfplots}

\usepackage{subcaption}

\usepackage{scrextend}
\usepackage{hyperref}

\newlength\figureheight
\newlength\figurewidth
\newtheorem{thm}{Theorem}
\numberwithin{thm}{section}

\numberwithin{lmm}{section}
\newtheorem{crll}{Corollary}
\numberwithin{crll}{section}
\newtheorem{prop}{Proposition}
\numberwithin{prop}{section}

\theoremstyle{definition}
\newtheorem{defn}{Definition}
\numberwithin{defn}{section}
\newtheorem{expl}{Example}
\numberwithin{expl}{section}

\theoremstyle{remark}

\numberwithin{rmk}{section}

\newcommand{\Sact}{\mathcal{S}}

\newcommand{\R}{\mathbb{R}}

\newcommand{\N}{\mathbb{N}}
\newcommand{\C}{\mathbb{C}}
\newcommand{\Z}{\mathbb{Z}}
\newcommand{\G}{G}

\newcommand{\risefac}[2]{{#1}^{\overline{#2}}}

\newcommand{\bigO}{\mathcal{O}}

\newcommand{\Fop}{\mathcal{F}}
\newcommand{\asyOp}{\mathop{\mathcal{A}}}
\newcommand{\asyOpV}[3]{\sideset{_{}}{^{#1}_{#3}}{\asyOp}}

\newcommand*\fring[3]{\R[[#1]]^{#2}}

\newcommand{\hopffg}{\mathcal{H}^\text{\normalfont{fg}}}
\newcommand{\unit}{\text{\normalfont{u}}}
\newcommand{\counit}{\epsilon}
\newcommand{\sdsubdiags}{\mathcal{P}^\text{s.d.}}

\DeclareMathOperator{\id}{id}
\DeclareMathOperator{\Aut}{Aut}

\DeclareMathOperator{\res}{res}

\author{Michael Borinsky\footnote{borinsky@physik.hu-berlin.de}}

\date{}
\title{Renormalized asymptotic enumeration of Feynman diagrams}

\begin{document}

\ifmmode
\usebox{\fgtadpoletwo}
\else
\newsavebox{\fgtadpoletwo}
\savebox{\fgtadpoletwo}{%
\begin{tikzpicture}[x=1ex,y=1ex,baseline={([yshift=-.5ex]current bounding box.center)}] \coordinate (vm); \coordinate [left=.7 of vm] (v0); \coordinate [right=.7 of vm] (v1); \draw (v0) circle(.7); \draw (v1) circle(.7); \filldraw (vm) circle (1pt); \end{tikzpicture}%
}
\fi
\ifmmode
\usebox{\fgbananathree}
\else
\newsavebox{\fgbananathree}
\savebox{\fgbananathree}{%
\begin{tikzpicture}[x=1ex,y=1ex,baseline={([yshift=-.5ex]current bounding box.center)}] \coordinate (vm); \coordinate [left=1 of vm] (v0); \coordinate [right=1 of vm] (v1); \draw (v0) -- (v1); \draw (vm) circle(1); \filldraw (v0) circle (1pt); \filldraw (v1) circle (1pt); \end{tikzpicture}%
}
\fi
\ifmmode
\usebox{\fgbananathreeandhandle}
\else
\newsavebox{\fgbananathreeandhandle}
\savebox{\fgbananathreeandhandle}{%
\begin{tikzpicture}[x=1ex,y=1ex,baseline={([yshift=-.5ex]current bounding box.center)}] \coordinate (v); \coordinate [above=1.2 of v] (vm1); \coordinate [left=1 of vm1] (v01); \coordinate [right=1 of vm1] (v11); \draw (v01) -- (v11); \draw (vm1) circle(1); \filldraw (v01) circle (1pt); \filldraw (v11) circle (1pt); \coordinate [below=1.2 of v] (vm2); \coordinate [left=.75 of vm2] (v02); \coordinate [right=.75 of vm2] (v12); \coordinate [left=.7 of v02] (i02); \coordinate [right=.7 of v12] (o02); \draw (v02) -- (v12); \filldraw (v02) circle (1pt); \filldraw (v12) circle (1pt); \draw (i02) circle(.7); \draw (o02) circle(.7); \end{tikzpicture}%
}
\fi
\ifmmode
\usebox{\fgcloseddunce}
\else
\newsavebox{\fgcloseddunce}
\savebox{\fgcloseddunce}{%
\begin{tikzpicture}[x=1ex,y=1ex,baseline={([yshift=-.5ex]current bounding box.center)}] \coordinate (v0); \coordinate[right=.5 of v0] (v3); \coordinate[right=1.5 of v3] (v4); \coordinate[above left=1.5 of v0] (v1); \coordinate[below left=1.5 of v0] (v2); \coordinate[right=.7 of v4] (o); \draw (v3) to[bend left=20] (v2); \draw (v3) to[bend right=20] (v1); \draw (v1) to[bend right=80] (v2); \draw (v1) to[bend left=80] (v2); \draw (v3) -- (v4); \filldraw (v1) circle(1pt); \filldraw (v2) circle(1pt); \filldraw (v3) circle(1pt); \filldraw (v4) circle(1pt); \draw (o) circle(.7); \end{tikzpicture}%
}
\fi
\ifmmode
\usebox{\fgdbleye}
\else
\newsavebox{\fgdbleye}
\savebox{\fgdbleye}{%
\begin{tikzpicture}[x=1ex,y=1ex,baseline={([yshift=-.5ex]current bounding box.center)}] \coordinate (v0); \coordinate[above left=1.5 of v0] (v1); \coordinate[below left=1.5 of v0] (v2); \coordinate[above right=1.5 of v0] (v3); \coordinate[below right=1.5 of v0] (v4); \draw (v1) to[bend left=80] (v2); \draw (v1) to[bend right=80] (v2); \draw (v3) to[bend right=80] (v4); \draw (v3) to[bend left=80] (v4); \draw (v1) -- (v3); \draw (v2) -- (v4); \filldraw (v1) circle(1pt); \filldraw (v2) circle(1pt); \filldraw (v3) circle(1pt); \filldraw (v4) circle(1pt); \end{tikzpicture}%
}
\fi
\ifmmode
\usebox{\fgdblhandle}
\else
\newsavebox{\fgdblhandle}
\savebox{\fgdblhandle}{%
\begin{tikzpicture}[x=1ex,y=1ex,baseline={([yshift=-.5ex]current bounding box.center)}] \coordinate (v) ; \def \n {2}; \def \rad {.7}; \def \rud {2}; \def \rid {2.7}; \foreach \s in {1,...,\n} { \def \angle {360/\n*(\s - 1)}; \coordinate (s) at ([shift=({\angle}:\rad)]v); \coordinate (u) at ([shift=({\angle}:\rud)]v); \coordinate (t) at ([shift=({\angle}:\rid)]v); \draw (s) -- (u); \filldraw (u) circle (1pt); \filldraw (s) circle (1pt); \draw (t) circle (.7); } \draw (v) circle(\rad); \end{tikzpicture}%
}
\fi
\ifmmode
\usebox{\fghandle}
\else
\newsavebox{\fghandle}
\savebox{\fghandle}{%
\begin{tikzpicture}[x=1ex,y=1ex,baseline={([yshift=-.5ex]current bounding box.center)}] \coordinate (v0); \coordinate [right=1.5 of v0] (v1); \coordinate [left=.7 of v0] (i0); \coordinate [right=.7 of v1] (o0); \draw (v0) -- (v1); \filldraw (v0) circle (1pt); \filldraw (v1) circle (1pt); \draw (i0) circle(.7); \draw (o0) circle(.7); \end{tikzpicture}%
}
\fi
\ifmmode
\usebox{\fgpropellerthree}
\else
\newsavebox{\fgpropellerthree}
\savebox{\fgpropellerthree}{%
\begin{tikzpicture}[x=1ex,y=1ex,baseline={([yshift=-.5ex]current bounding box.center)}] \coordinate (v) ; \def \n {3}; \def \rad {1.2}; \def \rud {1.9}; \foreach \s in {1,...,\n} { \def \angle {360/\n*(\s - 1)}; \coordinate (s) at ([shift=({\angle}:\rad)]v); \coordinate (u) at ([shift=({\angle}:\rud)]v); \draw (v) -- (s); \filldraw (s) circle (1pt); \draw (u) circle (.7); } \filldraw (v) circle (1pt); \end{tikzpicture}%
}
\fi
\ifmmode
\usebox{\fgtwohandles}
\else
\newsavebox{\fgtwohandles}
\savebox{\fgtwohandles}{%
\begin{tikzpicture}[x=1ex,y=1ex,baseline={([yshift=-.5ex]current bounding box.center)}] \coordinate (v); \coordinate [above=1.2 of v] (v01); \coordinate [right=1.5 of v01] (v11); \coordinate [left=.7 of v01] (i01); \coordinate [right=.7 of v11] (o01); \draw (v01) -- (v11); \filldraw (v01) circle (1pt); \filldraw (v11) circle (1pt); \draw (i01) circle(.7); \draw (o01) circle(.7); \coordinate [below=1.2 of v] (v02); \coordinate [right=1.5 of v02] (v12); \coordinate [left=.7 of v02] (i02); \coordinate [right=.7 of v12] (o02); \draw (v02) -- (v12); \filldraw (v02) circle (1pt); \filldraw (v12) circle (1pt); \draw (i02) circle(.7); \draw (o02) circle(.7); \end{tikzpicture}%
}
\fi
\ifmmode
\usebox{\fgtwobananasthree}
\else
\newsavebox{\fgtwobananasthree}
\savebox{\fgtwobananasthree}{%
\begin{tikzpicture}[x=1ex,y=1ex,baseline={([yshift=-.5ex]current bounding box.center)}] \coordinate (v); \coordinate [above=1.2 of v](v01); \coordinate [right=2 of v01] (v11); \coordinate [right=1 of v01] (vm1); \draw (v01) -- (v11); \draw (vm1) circle(1); \filldraw (v01) circle (1pt); \filldraw (v11) circle (1pt); \coordinate [below=1.2 of v](v02); \coordinate [right=2 of v02] (v12); \coordinate [right=1 of v02] (vm2); \draw (v02) -- (v12); \draw (vm2) circle(1); \filldraw (v02) circle (1pt); \filldraw (v12) circle (1pt); \end{tikzpicture}%
}
\fi
\ifmmode
\usebox{\fgtadpoletwoclosed}
\else
\newsavebox{\fgtadpoletwoclosed}
\savebox{\fgtadpoletwoclosed}{%
\begin{tikzpicture}[x=1ex,y=1ex,baseline={([yshift=-.5ex]current bounding box.center)}] \coordinate (vm); \coordinate [left=.7 of vm] (v0); \coordinate [right=.7 of vm] (v1); \coordinate [above=.7 of v0] (v2); \coordinate [above=.7 of v1] (v3); \draw (v0) circle(.7); \draw (v1) circle(.7); \draw (v3) arc(0:180:.7) (v2); \filldraw (vm) circle (1pt); \filldraw (v2) circle (1pt); \filldraw (v3) circle (1pt); \end{tikzpicture}%
}
\fi
\ifmmode
\usebox{\fgbananafour}
\else
\newsavebox{\fgbananafour}
\savebox{\fgbananafour}{%
\begin{tikzpicture}[x=1ex,y=1ex,baseline={([yshift=-.5ex]current bounding box.center)}] \coordinate (vm); \coordinate [left=1 of vm] (v0); \coordinate [right=1 of vm] (v1); \draw (v0) to[bend left=45] (v1); \draw (v0) to[bend right=45] (v1); \draw (vm) circle(1); \filldraw (v0) circle (1pt); \filldraw (v1) circle (1pt); \end{tikzpicture}%
}
\fi
\ifmmode
\usebox{\fgthreebubble}
\else
\newsavebox{\fgthreebubble}
\savebox{\fgthreebubble}{%
\begin{tikzpicture}[x=1ex,y=1ex,baseline={([yshift=-.5ex]current bounding box.center)}] \coordinate (vm); \coordinate [left=.7 of vm] (v0); \coordinate [right=.7 of vm] (v1); \coordinate [left=.7 of v0] (vc1); \coordinate [right=.7 of v1] (vc2); \draw (vc1) circle(.7); \draw (vc2) circle(.7); \draw (vm) circle(.7); \filldraw (v0) circle (1pt); \filldraw (v1) circle (1pt); \end{tikzpicture}%
}
\fi
\ifmmode
\usebox{\fgcontractedpropellerthree}
\else
\newsavebox{\fgcontractedpropellerthree}
\savebox{\fgcontractedpropellerthree}{%
\begin{tikzpicture}[x=1ex,y=1ex,baseline={([yshift=-.5ex]current bounding box.center)}] \coordinate (v) ; \def \n {3}; \def \rad {1.2}; \def \rud {1.9}; \foreach \s in {2,...,\n} { \def \angle {360/\n*(\s - 1)}; \coordinate (s) at ([shift=({\angle}:\rad)]v); \coordinate (u) at ([shift=({\angle}:\rud)]v); \draw (v) -- (s); \filldraw (s) circle (1pt); \draw (u) circle (.7); } \filldraw (v) circle (1pt); \coordinate (s) at ([shift=({0}:.7)]v); \draw (s) circle (.7); \end{tikzpicture}%
}
\fi
\ifmmode
\usebox{\fgcontracteddblhandle}
\else
\newsavebox{\fgcontracteddblhandle}
\savebox{\fgcontracteddblhandle}{%
\begin{tikzpicture}[x=1ex,y=1ex,baseline={([yshift=-.5ex]current bounding box.center)}] \coordinate (v0); \coordinate [right=1.5 of v0] (v1); \coordinate [left=.7 of v0] (i0); \coordinate [right=.7 of v1] (o0); \coordinate [right=.7 of o0] (v2); \coordinate [right=.7 of v2] (o1); \draw (v0) -- (v1); \filldraw (v0) circle (1pt); \filldraw (v1) circle (1pt); \filldraw (v2) circle (1pt); \draw (i0) circle(.7); \draw (o0) circle(.7); \draw (o1) circle(.7); \end{tikzpicture}%
}
\fi
\ifmmode
\usebox{\fgcontractedcloseddunce}
\else
\newsavebox{\fgcontractedcloseddunce}
\savebox{\fgcontractedcloseddunce}{%
\begin{tikzpicture}[x=1ex,y=1ex,baseline={([yshift=-.5ex]current bounding box.center)}] \coordinate (v0); \coordinate[right=.5 of v0] (v3); \coordinate[above left=1.5 of v0] (v1); \coordinate[below left=1.5 of v0] (v2); \coordinate[right=.7 of v3] (o); \draw (v3) to[bend left=20] (v2); \draw (v3) to[bend right=20] (v1); \draw (v1) to[bend right=80] (v2); \draw (v1) to[bend left=80] (v2); \filldraw (v1) circle(1pt); \filldraw (v2) circle(1pt); \filldraw (v3) circle(1pt); \draw (o) circle(.7); \end{tikzpicture}%
}
\fi
\ifmmode
\usebox{\fgthreerose}
\else
\newsavebox{\fgthreerose}
\savebox{\fgthreerose}{%
\begin{tikzpicture}[x=1ex,y=1ex,baseline={([yshift=-.5ex]current bounding box.center)}] \coordinate (v) ; \def \rad {1.5}; \coordinate (s1) at ([shift=(0:\rad)]v); \coordinate (s2) at ([shift=(120:\rad)]v); \coordinate (s3) at ([shift=(240:\rad)]v); \draw (v) to[out=60,in=90] (s1) to[out=-90,in=0-60] (v); \draw (v) to[out=180,in=210] (s2) to[out=30,in=60] (v); \draw (v) to[out=300,in=330] (s3) to[out=150,in=180] (v); \filldraw (v) circle (1pt); \end{tikzpicture}%
}
\fi
\ifmmode
\usebox{\fgdblcontractedpropellerthree}
\else
\newsavebox{\fgdblcontractedpropellerthree}
\savebox{\fgdblcontractedpropellerthree}{%
\begin{tikzpicture}[x=1ex,y=1ex,baseline={([yshift=-.5ex]current bounding box.center)}] \coordinate (v) ; \def \rad {1.5}; \coordinate (s1) at ([shift=(0:1.2)]v); \coordinate (s2) at ([shift=(120:\rad)]v); \coordinate (s3) at ([shift=(240:\rad)]v); \coordinate [right=.7 of s1] (o); \draw (v) to[out=180,in=210] (s2) to[out=30,in=60] (v); \draw (v) to[out=300,in=330] (s3) to[out=150,in=180] (v); \draw (v) -- (s1); \filldraw (v) circle (1pt); \filldraw (s1) circle (1pt); \draw (o) circle(.7); \end{tikzpicture}%
}
\fi
\ifmmode
\usebox{\fgbananathreewithbubble}
\else
\newsavebox{\fgbananathreewithbubble}
\savebox{\fgbananathreewithbubble}{%
\begin{tikzpicture}[x=1ex,y=1ex,baseline={([yshift=-.5ex]current bounding box.center)}] \coordinate (vm); \coordinate [left=1 of vm] (v0); \coordinate [right=1 of vm] (v1); \coordinate [right=.7 of v1] (o); \draw (v0) -- (v1); \draw (vm) circle(1); \draw (o) circle(.7); \filldraw (v0) circle (1pt); \filldraw (v1) circle (1pt); \end{tikzpicture}%
}
\fi
\ifmmode
\usebox{\fgtwotadpoletwos}
\else
\newsavebox{\fgtwotadpoletwos}
\savebox{\fgtwotadpoletwos}{%
\begin{tikzpicture}[x=1ex,y=1ex,baseline={([yshift=-.5ex]current bounding box.center)}] \coordinate (v); \coordinate [above=1.2 of v] (vm1); \coordinate [left=.7 of vm1] (v01); \coordinate [right=.7 of vm1] (v11); \draw (v01) circle(.7); \draw (v11) circle(.7); \filldraw (vm1) circle (1pt); \coordinate [below=1.2 of v] (vm2); \coordinate [left=.7 of vm2] (v02); \coordinate [right=.7 of vm2] (v12); \draw (v02) circle(.7); \draw (v12) circle(.7); \filldraw (vm2) circle (1pt); \end{tikzpicture}%
}
\fi
\ifmmode
\usebox{\fgbananathreeandtadpoletwo}
\else
\newsavebox{\fgbananathreeandtadpoletwo}
\savebox{\fgbananathreeandtadpoletwo}{%
\begin{tikzpicture}[x=1ex,y=1ex,baseline={([yshift=-.5ex]current bounding box.center)}] \coordinate (v); \coordinate [above=1.2 of v](vm1); \coordinate [left=1 of vm1] (v01); \coordinate [right=1 of vm1] (v11); \draw (v01) -- (v11); \draw (vm1) circle(1); \filldraw (v01) circle (1pt); \filldraw (v11) circle (1pt); \coordinate [below=1.2 of v] (vm2); \coordinate [left=.7 of vm2] (v02); \coordinate [right=.7 of vm2] (v12); \draw (v02) circle(.7); \draw (v12) circle(.7); \filldraw (vm2) circle (1pt); \end{tikzpicture}%
}
\fi
\ifmmode
\usebox{\fghandleandtadpoletwo}
\else
\newsavebox{\fghandleandtadpoletwo}
\savebox{\fghandleandtadpoletwo}{%
\begin{tikzpicture}[x=1ex,y=1ex,baseline={([yshift=-.5ex]current bounding box.center)}] \coordinate (v); \coordinate [above=1.2 of v] (vm1); \coordinate [left=.75 of vm1] (v01); \coordinate [right=.75 of vm1] (v11); \coordinate [left=.7 of v01] (i01); \coordinate [right=.7 of v11] (o01); \draw (v01) -- (v11); \filldraw (v01) circle (1pt); \filldraw (v11) circle (1pt); \draw (i01) circle(.7); \draw (o01) circle(.7); \coordinate [below=1.2 of v] (vm2); \coordinate [left=.7 of vm2] (v02); \coordinate [right=.7 of vm2] (v12); \draw (v02) circle(.7); \draw (v12) circle(.7); \filldraw (vm2) circle (1pt); \end{tikzpicture}%
}
\fi
\ifmmode
\usebox{\fgbananathreehandle}
\else
\newsavebox{\fgbananathreehandle}
\savebox{\fgbananathreehandle}{%
\begin{tikzpicture}[x=1ex,y=1ex,baseline={([yshift=-.5ex]current bounding box.center)}] \coordinate (vm); \coordinate [left=1 of vm] (v0); \coordinate [right=1 of vm] (v1); \coordinate [right=1.5 of v1] (v2); \coordinate [right=.7 of v2] (o); \draw (v0) -- (v1); \draw (v1) -- (v2); \draw (vm) circle(1); \draw (o) circle(.7); \filldraw (v0) circle (1pt); \filldraw (v1) circle (1pt); \filldraw (v2) circle (1pt); \end{tikzpicture}%
}
\fi

\ifmmode
\usebox{\fgsimpleprop}
\else
\newsavebox{\fgsimpleprop}
\savebox{\fgsimpleprop}{%
    \begin{tikzpicture}[x=1ex,y=1ex,baseline={([yshift=-.5ex]current bounding box.center)}] \coordinate (v) ; \coordinate [right=1.2 of v] (u); \draw (v) -- (u); \filldraw (v) circle (1pt); \filldraw (u) circle (1pt); \end{tikzpicture}%
}
\fi
\ifmmode
\usebox{\fgsimplethreevtx}
\else
\newsavebox{\fgsimplethreevtx}
\savebox{\fgsimplethreevtx}{%
\begin{tikzpicture}[x=1ex,y=1ex,baseline={([yshift=-.5ex]current bounding box.center)}] \coordinate (v) ; \def \n {3}; \def \rad {1.2}; \foreach \s in {1,...,5} { \def \angle {360/\n*(\s - 1)}; \coordinate (u) at ([shift=({\angle}:\rad)]v); \draw (v) -- (u); \filldraw (u) circle (1pt); } \filldraw (v) circle (1pt); \end{tikzpicture}%
}
\fi
\ifmmode
\usebox{\fgonetadpolephithree}
\else
\newsavebox{\fgonetadpolephithree}
\savebox{\fgonetadpolephithree}{%
    \begin{tikzpicture}[x=2ex,y=2ex,baseline={([yshift=-.5ex]current bounding box.center)}] \coordinate (v0) ; \coordinate [right=1 of v0] (v1); \coordinate [left=.7 of v0] (i0); \coordinate [left=.5 of v1] (vm); \draw (vm) circle(.5); \draw (i0) -- (v0); \filldraw (v0) circle(1pt); \filldraw (i0) circle (1pt); \end{tikzpicture}%
}
\fi
\ifmmode
\usebox{\fgtwojoneloopbubblephithree}
\else
\newsavebox{\fgtwojoneloopbubblephithree}
\savebox{\fgtwojoneloopbubblephithree}{%
    \begin{tikzpicture}[x=2ex,y=2ex,baseline={([yshift=-.5ex]current bounding box.center)}] \coordinate (v0) ; \coordinate [right=1 of v0] (v1); \coordinate [left=.7 of v0] (i0); \coordinate [right=.7 of v1] (o0); \coordinate [left=.5 of v1] (vm); \draw (vm) circle(.5); \draw (i0) -- (v0); \draw (o0) -- (v1); \filldraw (v0) circle(1pt); \filldraw (v1) circle(1pt); \filldraw (i0) circle (1pt); \filldraw (o0) circle (1pt); \end{tikzpicture}%
}
\fi
\ifmmode
\usebox{\fgthreejoneltrianglephithree}
\else
\newsavebox{\fgthreejoneltrianglephithree}
\savebox{\fgthreejoneltrianglephithree}{%
\begin{tikzpicture}[x=1ex,y=1ex,baseline={([yshift=-.5ex]current bounding box.center)}] \coordinate (v) ; \def \n {3}; \def \rad {1}; \def \rud {2.2}; \foreach \s in {1,...,5} { \def \angle {360/\n*(\s - 1)}; \def \ungle {360/\n*\s}; \coordinate (s) at ([shift=({\angle}:\rad)]v); \coordinate (t) at ([shift=({\ungle}:\rad)]v); \coordinate (u) at ([shift=({\angle}:\rud)]v); \draw (s) -- (u); \filldraw (u) circle (1pt); \filldraw (s) circle (1pt); } \draw (v) circle(\rad); \end{tikzpicture}%
}
\fi
\ifmmode
\usebox{\fgthreejonelpropinsphithree}
\else
\newsavebox{\fgthreejonelpropinsphithree}
\savebox{\fgthreejonelpropinsphithree}{%
    \begin{tikzpicture}[x=2ex,y=2ex,baseline={([yshift=-.5ex]current bounding box.center)}] \coordinate (v0) ; \coordinate [right=1 of v0] (v1); \coordinate [right=.7 of v1] (v2); \coordinate [left=.7 of v0] (i0); \coordinate [above right=.7 of v2] (o0); \coordinate [below right=.7 of v2] (o1); \coordinate [left=.5 of v1] (vm); \draw (vm) circle(.5); \draw (i0) -- (v0); \draw (v1) -- (v2); \draw (o0) -- (v2); \draw (o1) -- (v2); \filldraw (v0) circle(1pt); \filldraw (v1) circle(1pt); \filldraw (v2) circle(1pt); \filldraw (i0) circle (1pt); \filldraw (o0) circle (1pt); \filldraw (o1) circle (1pt); \end{tikzpicture}%
}
\fi

\maketitle
\begin{abstract}
A method to obtain all-order asymptotic results for the coefficients of perturbative expansions in zero-dimensional quantum field is described. The focus is on the enumeration of the number of skeleton or primitive diagrams of a certain QFT and its asymptotics. The procedure heavily applies techniques from singularity analysis. To utilize singularity analysis, a representation of the zero-dimensional path integral as a generalized hyperelliptic curve is deduced. As applications the full asymptotic expansions of the number of disconnected, connected, 1PI and skeleton Feynman diagrams in various theories are given.
\end{abstract}
\tableofcontents

\section{Introduction}

Perturbation expansions in quantum field theory (QFT) turn out to be divergent \cite{dyson1952divergence}, that means perturbative expansions of observables in those theories usually have a vanishing radius of convergence. The large growth of the coefficients, causing the divergence, is believed to be governed by the proliferation of Feynman diagrams with increasing loop number. The analysis of this large order behavior lead to many important results reaching far beyond the scope of quantum field theory \cite{bender1969anharmonic,bender1973anharmonic,le2012large}. Moreover, the divergence of the perturbation expansion in QFT is linked to non-perturbative effects \cite{alvarez2004langer,garoufalidis2012asymptotics,dunne2014generating}.

The extraction of large order results from realistic QFTs becomes most complicated when \textit{renormalization} comes into play. For instance, the relationship between \textit{renormalons}, which are avatars of renormalization at large order, and \textit{instantons} \cite{lautrup1977high,Zichichi:1979gj}, classical field configurations, which are in close correspondence with the large order behavior of the theory \cite{lipatov1977divergence}, remains elusive \cite{suslov2005divergent}.

Zero-dimensional QFT serves as a toy-model for realistic QFT calculations. Especially, the behavior of zero-dimensional QFT at large order is of interest, as calculations in these regimes for realistic QFTs are extremely delicate. The utility of zero-dimensional QFT as a reasonable toy-model comes mainly from the interpretation of observables as \textit{combinatorial generating functions} of the number of Feynman diagrams. It is therefore a reasonable starting point to study the divergence of the perturbation expansions in realistic QFTs. Zero-dimensional QFT has been extensively studied \cite{hurst1952enumeration,cvitanovic1978number,bender1978asymptotic,goldberg1991tree,argyres2001zero,molinari2006enumeration}. 

The purpose of this article is to study the \textit{asymptotics} of the observables in zero-dimensional QFT and the associated graph counting problems. The focus will be on the \textit{renormalization} of zero-dimensional QFT and the asymptotics of the renormalization constants, which will provide the asymptotic number of \textit{skeleton} Feynman diagrams.

As Feynman diagram techniques can also be used in quantum mechanics, the analysis provided here can also give insights into large order expansions of observables in classical theories \cite{bender1969anharmonic,le2012large}.

With techniques from \cite{borinsky2016generating} complete asymptotic expansions for all observables in zero-dimensional field theory and the respective enumeration problems can be obtained. These techniques are related to the theory of resurgence \cite{ecalle1981fonctions}. To establish the connection of renormalization and restricted graph counting, the lattice structure of renormalization can be used \cite{borinsky2016lattices}. In this article, these tools will be combined and applied to various examples in zero-dimensional QFT and graph counting.

In Section \ref{sec:formalint}, basic properties of the formal integrals, which are associated to zero-dimensional QFTs, will be introduced and a representation of these expressions as local expansions of a \textit{generalized hyperelliptic curve} will be derived. These expansions will enable us to rigorously extract complete asymptotic expansions by purely algebraic means. In Section \ref{sec:ring}, the ring of factorially divergent power series will be introduced, which allows very efficient extraction of complete asymptotic expansions from composed formal quantities. A short overview over some observables in zero-dimensional QFT will be given in Section \ref{sec:overview} and, to motivate the renormalization procedure, the structure of the Hopf algebra of Feynman diagrams will be sketched in Section \ref{sec:hopfalgebra}. In Section \ref{sec:applications}, various examples will be given. Explicit asymptotics of the number of skeleton diagrams will be provided for $\varphi^3$, $\varphi^4$, QED, quenched QED and Yukawa theory.
\section{Formal integrals}
\label{sec:formalint}

The starting point for zero-dimensional QFT is the \textit{path integral}, which becomes an ordinary integral in the zero-dimensional case. For instance, in a scalar theory the \textit{partition function} is given by
\begin{align} \label{eqn:formalfuncintegral1} Z(\hbar) &:= \int_\R \frac{dx}{\sqrt{2 \pi \hbar } } e^{\frac{1}{\hbar} \left( -\frac{x^2}{2a} + V(x) \right) }, \end{align}
where $V \in x^3 \R[[x]]$, the \textit{potential}, is some power series with the first three coefficients in $x$ vanishing and $a$ is a strictly positive parameter. The whole exponent $\Sact(x) = -\frac{x^2}{2a} + V(x)$ is the \textit{action}.

The integral \eqref{eqn:formalfuncintegral1} is ill-defined for general $V(x)$. If we substitute, for example, $V(x) = \frac{x^4}{4!}$, it is not integrable over $\R$. Furthermore, the power series expansion makes only limited sense as the actual function $Z(\hbar)$ will have a singularity at $\hbar=0$ - even in cases where the expression is integrable. One way to continue is to modify the integration contour, such that the integrand vanishes fast enough at the border of the integration domain. The disadvantage of this method is that the integration contour must be chosen on a case by case basis. 

This article is mainly concerned with the coefficients of the expansion in $\hbar$ of the integral \eqref{eqn:formalfuncintegral1}. We wish to give meaning to such an expressions in a way that highlights its \textit{power series} nature. Moreover, we like to free ourselves from restrictions in choices of $V(x)$ as far as possible.
We therefore treat the integral \eqref{eqn:formalfuncintegral1} as a \textit{formal} expression, which is not required to yield a proper function in $\hbar$, but a formal power series in this parameter.

The procedure to obtain a power series expansion from this formal integral is well-known and widely used \cite{itzykson2006quantum}: The potential $V(x)$ is treated as a perturbation around the Gaussian kernel and the remaining integrand is expanded. 
The `integration' will be solely performed by applying the identity 
\begin{align*} \int_\R \frac{dx}{\sqrt{2 \pi \hbar}} e^{-\frac{x^2}{2 a \hbar}} x^{2n} &= \sqrt{a} (a\hbar)^n (2n-1)!! & & n \geq 0. \end{align*}
This procedure mimics the calculation of amplitudes in higher dimensions, as the above identity is the zero-dimensional version of Wick's theorem \cite{itzykson2006quantum}. This way, it directly incorporates the interpretation of the coefficients of the power series as \textit{Feynman diagrams}. 
Unfortunately, these \textit{formal integrals} seem not to have been studied in detail as isolated mathematical entities. We will expand on these notions in further detail in a future publication \cite{borinsky2017graph}. For know, we will give a translation of the formal integral to a well defined formal power series. This will serve as a definition of a formal integral in the scope of this article.

We expand the exponent of $V(x)$ and exchange integration and summation and thereby \textit{define} the zero-dimensional path integral as the following expression:
\begin{defn}
\label{def:formalintegral}
Let $\mathcal{F} : x^2\R[[x]] \rightarrow \R[[\hbar]]$ be the operator that maps $\Sact(x) \in x^2 \R[[x]]$, a power series with vanishing constant and linear terms as well as a strictly negative quadratic term, $\Sact(x) = - \frac{x^2}{2a} + V(x)$, to $\mathcal{F}[\Sact(x)]\in\R[[\hbar]]$ a power series in $\hbar$, such that 
\begin{align} \label{eqn:formalintegralpwrsrs} \mathcal{F}[\Sact(x)](\hbar) = \sqrt{a} \sum_{n=0}^\infty \left(a \hbar\right)^{n} (2n-1)!! [x^{2n}] e^{\frac{1}{\hbar}V(x)}. \end{align}
\end{defn}
This gives a well defined power series in $\hbar$, because $[x^{2n}] e^{\frac{1}{\hbar}V(x)}$ is a polynomial in $\hbar^{-1}$ of order smaller than $n$.

An advantage of applying this definition rather then using the integral itself is that Definition \ref{def:formalintegral} gives an unambiguous procedure to obtain result for a given potential, whereas the integration depends heavily on the choice of the integration contour. 

The most important property of $\Fop$, and another motivation to define it, is that it enumerates \textit{multigraphs}. 
\subsection{Diagrammatic interpretation}
\begin{prop}
\label{prop:diagraminterpretation}
If $\Sact(x) = -\frac{x^2}{2 a} + \sum_{k=3}^\infty \frac{\lambda_k}{k!} x^k$ with $a>0$, then 
\begin{align*} \Fop[\Sact(x)](\hbar) = \sqrt{a}\sum_{\Gamma} \hbar^{|E(\Gamma)| - |V(\Gamma)|} \frac{ a^{|E(\Gamma)|} \prod_{v \in V(\Gamma)} \lambda_{n_v}}{|\Aut \Gamma|} \end{align*}
where the sum is over \textit{all} multigraphs $\Gamma$ without two or one-valent vertices, where $E(\Gamma)$ and $V(\Gamma)$ are the sets of edges and vertices of $\Gamma$ and $n_v$ is the valency of vertex $v$.
\end{prop}
This identity can also be used as definition of $\Fop$.
It is well-known that the terms in the expansion of the integral \eqref{eqn:formalfuncintegral1} and therefore also the terms of $\Fop[\Sact(x)](\hbar)$ can be interpreted as a sum over Feynman diagrams weighted by their symmetry factor \cite{cvitanovic1978number}. To calculate the $n$-th coefficient of \eqref{eqn:formalfuncintegral1} or \eqref{eqn:formalintegralpwrsrs} with $\Sact(x) = - \frac{x^2}{2a} + V(x)$ and $V(x) = \sum_{k=3} \frac{\lambda_k}{k!} x^k$, 
\begin{enumerate}
\item
draw all graphs with \textit{excess} $n$. The excess of a diagram $\Gamma$ is given by $|E(\Gamma)|-|V(\Gamma)|$, the number of edges minus the number of vertices. 
For connected graphs the excess is equal to the number of \textit{loops} minus 1. We say a graph has $n$ loops if it has $n$ independent cycles. The number of loops is also the first \textit{Betti number} of the graph.
\item
Apply Feynman rules: For each individual graph $\Gamma$ calculate the product $\prod_{v \in V(\Gamma)} \lambda_{|v|}$, where each vertex contributes with a factor $\lambda_{|v|}$ with $|v|$ the valency of the vertex. Subsequently, multiply by $a^{|E(\Gamma)|}$.
\item
Calculate the cardinality of the automorphism group of the graph. $k$-fold double-edges give additional automorphisms of order $k!$ and $k$-fold self-loops give additional automorphisms of order $2^k k!$, which both commute with the rest of the automorphism group. Divide the result of the previous calculation by this cardinality. 
\item
Sum all monomials and multiply the obtained polynomial by a normalization factor of $\sqrt{a}$.
\end{enumerate}
We may write the power series expansion of $\Fop[\Sact(x)](\hbar)$ in a diagrammatic way as follows:
\begin{gather} \begin{gathered} \label{eqn:generalexpansion} \Fop[\Sact(x)](\hbar) = \phi_{\Sact} \Big( 1 + \frac18 {  \ifmmode \usebox{\fghandle} \else \newsavebox{\fghandle} \savebox{\fghandle}{ \begin{tikzpicture}[x=1ex,y=1ex,baseline={([yshift=-.5ex]current bounding box.center)}] \coordinate (v0); \coordinate [right=1.5 of v0] (v1); \coordinate [left=.7 of v0] (i0); \coordinate [right=.7 of v1] (o0); \draw (v0) -- (v1); \filldraw (v0) circle (1pt); \filldraw (v1) circle (1pt); \draw (i0) circle(.7); \draw (o0) circle(.7); \end{tikzpicture} } \fi } + \frac{1}{12} {  \ifmmode \usebox{\fgbananathree} \else \newsavebox{\fgbananathree} \savebox{\fgbananathree}{ \begin{tikzpicture}[x=1ex,y=1ex,baseline={([yshift=-.5ex]current bounding box.center)}] \coordinate (vm); \coordinate [left=1 of vm] (v0); \coordinate [right=1 of vm] (v1); \draw (v0) -- (v1); \draw (vm) circle(1); \filldraw (v0) circle (1pt); \filldraw (v1) circle (1pt); \end{tikzpicture} } \fi } + \frac{1}{8} {  \ifmmode \usebox{\fgtadpoletwo} \else \newsavebox{\fgtadpoletwo} \savebox{\fgtadpoletwo}{ \begin{tikzpicture}[x=1ex,y=1ex,baseline={([yshift=-.5ex]current bounding box.center)}] \coordinate (vm); \coordinate [left=.7 of vm] (v0); \coordinate [right=.7 of vm] (v1); \draw (v0) circle(.7); \draw (v1) circle(.7); \filldraw (vm) circle (1pt); \end{tikzpicture} } \fi } \\ + \frac{1}{128} {  \ifmmode \usebox{\fgtwohandles} \else \newsavebox{\fgtwohandles} \savebox{\fgtwohandles}{ \begin{tikzpicture}[x=1ex,y=1ex,baseline={([yshift=-.5ex]current bounding box.center)}] \coordinate (v); \coordinate [above=1.2 of v] (v01); \coordinate [right=1.5 of v01] (v11); \coordinate [left=.7 of v01] (i01); \coordinate [right=.7 of v11] (o01); \draw (v01) -- (v11); \filldraw (v01) circle (1pt); \filldraw (v11) circle (1pt); \draw (i01) circle(.7); \draw (o01) circle(.7); \coordinate [below=1.2 of v] (v02); \coordinate [right=1.5 of v02] (v12); \coordinate [left=.7 of v02] (i02); \coordinate [right=.7 of v12] (o02); \draw (v02) -- (v12); \filldraw (v02) circle (1pt); \filldraw (v12) circle (1pt); \draw (i02) circle(.7); \draw (o02) circle(.7); \end{tikzpicture} } \fi } + \frac{1}{288} {  \ifmmode \usebox{\fgtwobananasthree} \else \newsavebox{\fgtwobananasthree} \savebox{\fgtwobananasthree}{ \begin{tikzpicture}[x=1ex,y=1ex,baseline={([yshift=-.5ex]current bounding box.center)}] \coordinate (v); \coordinate [above=1.2 of v](v01); \coordinate [right=2 of v01] (v11); \coordinate [right=1 of v01] (vm1); \draw (v01) -- (v11); \draw (vm1) circle(1); \filldraw (v01) circle (1pt); \filldraw (v11) circle (1pt); \coordinate [below=1.2 of v](v02); \coordinate [right=2 of v02] (v12); \coordinate [right=1 of v02] (vm2); \draw (v02) -- (v12); \draw (vm2) circle(1); \filldraw (v02) circle (1pt); \filldraw (v12) circle (1pt); \end{tikzpicture} } \fi } + \frac{1}{96} {  \ifmmode \usebox{\fgbananathreeandhandle} \else \newsavebox{\fgbananathreeandhandle} \savebox{\fgbananathreeandhandle}{ \begin{tikzpicture}[x=1ex,y=1ex,baseline={([yshift=-.5ex]current bounding box.center)}] \coordinate (v); \coordinate [above=1.2 of v] (vm1); \coordinate [left=1 of vm1] (v01); \coordinate [right=1 of vm1] (v11); \draw (v01) -- (v11); \draw (vm1) circle(1); \filldraw (v01) circle (1pt); \filldraw (v11) circle (1pt); \coordinate [below=1.2 of v] (vm2); \coordinate [left=.75 of vm2] (v02); \coordinate [right=.75 of vm2] (v12); \coordinate [left=.7 of v02] (i02); \coordinate [right=.7 of v12] (o02); \draw (v02) -- (v12); \filldraw (v02) circle (1pt); \filldraw (v12) circle (1pt); \draw (i02) circle(.7); \draw (o02) circle(.7); \end{tikzpicture} } \fi } + \frac{1}{48} {  \ifmmode \usebox{\fgpropellerthree} \else \newsavebox{\fgpropellerthree} \savebox{\fgpropellerthree}{ \begin{tikzpicture}[x=1ex,y=1ex,baseline={([yshift=-.5ex]current bounding box.center)}] \coordinate (v) ; \def \n {3}; \def \rad {1.2}; \def \rud {1.9}; \foreach \s in {1,...,\n} { \def \angle {360/\n*(\s - 1)}; \coordinate (s) at ([shift=({\angle}:\rad)]v); \coordinate (u) at ([shift=({\angle}:\rud)]v); \draw (v) -- (s); \filldraw (s) circle (1pt); \draw (u) circle (.7); } \filldraw (v) circle (1pt); \end{tikzpicture} } \fi } \\ + \frac{1}{16} {  \ifmmode \usebox{\fgdblhandle} \else \newsavebox{\fgdblhandle} \savebox{\fgdblhandle}{ \begin{tikzpicture}[x=1ex,y=1ex,baseline={([yshift=-.5ex]current bounding box.center)}] \coordinate (v) ; \def \n {2}; \def \rad {.7}; \def \rud {2}; \def \rid {2.7}; \foreach \s in {1,...,\n} { \def \angle {360/\n*(\s - 1)}; \coordinate (s) at ([shift=({\angle}:\rad)]v); \coordinate (u) at ([shift=({\angle}:\rud)]v); \coordinate (t) at ([shift=({\angle}:\rid)]v); \draw (s) -- (u); \filldraw (u) circle (1pt); \filldraw (s) circle (1pt); \draw (t) circle (.7); } \draw (v) circle(\rad); \end{tikzpicture} } \fi } + \frac{1}{16} {  \ifmmode \usebox{\fgdbleye} \else \newsavebox{\fgdbleye} \savebox{\fgdbleye}{ \begin{tikzpicture}[x=1ex,y=1ex,baseline={([yshift=-.5ex]current bounding box.center)}] \coordinate (v0); \coordinate[above left=1.5 of v0] (v1); \coordinate[below left=1.5 of v0] (v2); \coordinate[above right=1.5 of v0] (v3); \coordinate[below right=1.5 of v0] (v4); \draw (v1) to[bend left=80] (v2); \draw (v1) to[bend right=80] (v2); \draw (v3) to[bend right=80] (v4); \draw (v3) to[bend left=80] (v4); \draw (v1) -- (v3); \draw (v2) -- (v4); \filldraw (v1) circle(1pt); \filldraw (v2) circle(1pt); \filldraw (v3) circle(1pt); \filldraw (v4) circle(1pt); \end{tikzpicture} } \fi } + \frac{1}{8} {  \ifmmode \usebox{\fgcloseddunce} \else \newsavebox{\fgcloseddunce} \savebox{\fgcloseddunce}{ \begin{tikzpicture}[x=1ex,y=1ex,baseline={([yshift=-.5ex]current bounding box.center)}] \coordinate (v0); \coordinate[right=.5 of v0] (v3); \coordinate[right=1.5 of v3] (v4); \coordinate[above left=1.5 of v0] (v1); \coordinate[below left=1.5 of v0] (v2); \coordinate[right=.7 of v4] (o); \draw (v3) to[bend left=20] (v2); \draw (v3) to[bend right=20] (v1); \draw (v1) to[bend right=80] (v2); \draw (v1) to[bend left=80] (v2); \draw (v3) -- (v4); \filldraw (v1) circle(1pt); \filldraw (v2) circle(1pt); \filldraw (v3) circle(1pt); \filldraw (v4) circle(1pt); \draw (o) circle(.7); \end{tikzpicture} } \fi } + \frac{1}{24} {  \begin{tikzpicture}[x=1ex,y=1ex,baseline={([yshift=-.5ex]current bounding box.center)}] \coordinate (v) ; \def \n {3}; \def \rad {1.2}; \foreach \s in {1,...,\n} { \def \angle {360/\n*(\s - 1)}; \coordinate (s) at ([shift=({\angle}:\rad)]v); \draw (v) -- (s); \filldraw (s) circle (1pt); } \draw (v) circle (\rad); \filldraw (v) circle (1pt); \end{tikzpicture} } \\ + \frac{1}{96} {  \ifmmode \usebox{\fgbananathreeandtadpoletwo} \else \newsavebox{\fgbananathreeandtadpoletwo} \savebox{\fgbananathreeandtadpoletwo}{ \begin{tikzpicture}[x=1ex,y=1ex,baseline={([yshift=-.5ex]current bounding box.center)}] \coordinate (v); \coordinate [above=1.2 of v](vm1); \coordinate [left=1 of vm1] (v01); \coordinate [right=1 of vm1] (v11); \draw (v01) -- (v11); \draw (vm1) circle(1); \filldraw (v01) circle (1pt); \filldraw (v11) circle (1pt); \coordinate [below=1.2 of v] (vm2); \coordinate [left=.7 of vm2] (v02); \coordinate [right=.7 of vm2] (v12); \draw (v02) circle(.7); \draw (v12) circle(.7); \filldraw (vm2) circle (1pt); \end{tikzpicture} } \fi } + \frac{1}{64} {  \ifmmode \usebox{\fghandleandtadpoletwo} \else \newsavebox{\fghandleandtadpoletwo} \savebox{\fghandleandtadpoletwo}{ \begin{tikzpicture}[x=1ex,y=1ex,baseline={([yshift=-.5ex]current bounding box.center)}] \coordinate (v); \coordinate [above=1.2 of v] (vm1); \coordinate [left=.75 of vm1] (v01); \coordinate [right=.75 of vm1] (v11); \coordinate [left=.7 of v01] (i01); \coordinate [right=.7 of v11] (o01); \draw (v01) -- (v11); \filldraw (v01) circle (1pt); \filldraw (v11) circle (1pt); \draw (i01) circle(.7); \draw (o01) circle(.7); \coordinate [below=1.2 of v] (vm2); \coordinate [left=.7 of vm2] (v02); \coordinate [right=.7 of vm2] (v12); \draw (v02) circle(.7); \draw (v12) circle(.7); \filldraw (vm2) circle (1pt); \end{tikzpicture} } \fi } + \frac{1}{8} {  \ifmmode \usebox{\fgtadpoletwoclosed} \else \newsavebox{\fgtadpoletwoclosed} \savebox{\fgtadpoletwoclosed}{ \begin{tikzpicture}[x=1ex,y=1ex,baseline={([yshift=-.5ex]current bounding box.center)}] \coordinate (vm); \coordinate [left=.7 of vm] (v0); \coordinate [right=.7 of vm] (v1); \coordinate [above=.7 of v0] (v2); \coordinate [above=.7 of v1] (v3); \draw (v0) circle(.7); \draw (v1) circle(.7); \draw (v3) arc(0:180:.7) (v2); \filldraw (vm) circle (1pt); \filldraw (v2) circle (1pt); \filldraw (v3) circle (1pt); \end{tikzpicture} } \fi } + \frac{1}{16} {  \ifmmode \usebox{\fgcontractedpropellerthree} \else \newsavebox{\fgcontractedpropellerthree} \savebox{\fgcontractedpropellerthree}{ \begin{tikzpicture}[x=1ex,y=1ex,baseline={([yshift=-.5ex]current bounding box.center)}] \coordinate (v) ; \def \n {3}; \def \rad {1.2}; \def \rud {1.9}; \foreach \s in {2,...,\n} { \def \angle {360/\n*(\s - 1)}; \coordinate (s) at ([shift=({\angle}:\rad)]v); \coordinate (u) at ([shift=({\angle}:\rud)]v); \draw (v) -- (s); \filldraw (s) circle (1pt); \draw (u) circle (.7); } \filldraw (v) circle (1pt); \coordinate (s) at ([shift=({0}:.7)]v); \draw (s) circle (.7); \end{tikzpicture} } \fi } + \frac{1}{8} {  \ifmmode \usebox{\fgcontracteddblhandle} \else \newsavebox{\fgcontracteddblhandle} \savebox{\fgcontracteddblhandle}{ \begin{tikzpicture}[x=1ex,y=1ex,baseline={([yshift=-.5ex]current bounding box.center)}] \coordinate (v0); \coordinate [right=1.5 of v0] (v1); \coordinate [left=.7 of v0] (i0); \coordinate [right=.7 of v1] (o0); \coordinate [right=.7 of o0] (v2); \coordinate [right=.7 of v2] (o1); \draw (v0) -- (v1); \filldraw (v0) circle (1pt); \filldraw (v1) circle (1pt); \filldraw (v2) circle (1pt); \draw (i0) circle(.7); \draw (o0) circle(.7); \draw (o1) circle(.7); \end{tikzpicture} } \fi } + \frac{1}{12} {  \ifmmode \usebox{\fgbananathreehandle} \else \newsavebox{\fgbananathreehandle} \savebox{\fgbananathreehandle}{ \begin{tikzpicture}[x=1ex,y=1ex,baseline={([yshift=-.5ex]current bounding box.center)}] \coordinate (vm); \coordinate [left=1 of vm] (v0); \coordinate [right=1 of vm] (v1); \coordinate [right=1.5 of v1] (v2); \coordinate [right=.7 of v2] (o); \draw (v0) -- (v1); \draw (v1) -- (v2); \draw (vm) circle(1); \draw (o) circle(.7); \filldraw (v0) circle (1pt); \filldraw (v1) circle (1pt); \filldraw (v2) circle (1pt); \end{tikzpicture} } \fi } + \frac{1}{8} {  \ifmmode \usebox{\fgcontractedcloseddunce} \else \newsavebox{\fgcontractedcloseddunce} \savebox{\fgcontractedcloseddunce}{ \begin{tikzpicture}[x=1ex,y=1ex,baseline={([yshift=-.5ex]current bounding box.center)}] \coordinate (v0); \coordinate[right=.5 of v0] (v3); \coordinate[above left=1.5 of v0] (v1); \coordinate[below left=1.5 of v0] (v2); \coordinate[right=.7 of v3] (o); \draw (v3) to[bend left=20] (v2); \draw (v3) to[bend right=20] (v1); \draw (v1) to[bend right=80] (v2); \draw (v1) to[bend left=80] (v2); \filldraw (v1) circle(1pt); \filldraw (v2) circle(1pt); \filldraw (v3) circle(1pt); \draw (o) circle(.7); \end{tikzpicture} } \fi } \\ + \frac{1}{128} {  \ifmmode \usebox{\fgtwotadpoletwos} \else \newsavebox{\fgtwotadpoletwos} \savebox{\fgtwotadpoletwos}{ \begin{tikzpicture}[x=1ex,y=1ex,baseline={([yshift=-.5ex]current bounding box.center)}] \coordinate (v); \coordinate [above=1.2 of v] (vm1); \coordinate [left=.7 of vm1] (v01); \coordinate [right=.7 of vm1] (v11); \draw (v01) circle(.7); \draw (v11) circle(.7); \filldraw (vm1) circle (1pt); \coordinate [below=1.2 of v] (vm2); \coordinate [left=.7 of vm2] (v02); \coordinate [right=.7 of vm2] (v12); \draw (v02) circle(.7); \draw (v12) circle(.7); \filldraw (vm2) circle (1pt); \end{tikzpicture} } \fi } + \frac{1}{48} {  \ifmmode \usebox{\fgbananafour} \else \newsavebox{\fgbananafour} \savebox{\fgbananafour}{ \begin{tikzpicture}[x=1ex,y=1ex,baseline={([yshift=-.5ex]current bounding box.center)}] \coordinate (vm); \coordinate [left=1 of vm] (v0); \coordinate [right=1 of vm] (v1); \draw (v0) to[bend left=45] (v1); \draw (v0) to[bend right=45] (v1); \draw (vm) circle(1); \filldraw (v0) circle (1pt); \filldraw (v1) circle (1pt); \end{tikzpicture} } \fi } + \frac{1}{16} {  \ifmmode \usebox{\fgthreebubble} \else \newsavebox{\fgthreebubble} \savebox{\fgthreebubble}{ \begin{tikzpicture}[x=1ex,y=1ex,baseline={([yshift=-.5ex]current bounding box.center)}] \coordinate (vm); \coordinate [left=.7 of vm] (v0); \coordinate [right=.7 of vm] (v1); \coordinate [left=.7 of v0] (vc1); \coordinate [right=.7 of v1] (vc2); \draw (vc1) circle(.7); \draw (vc2) circle(.7); \draw (vm) circle(.7); \filldraw (v0) circle (1pt); \filldraw (v1) circle (1pt); \end{tikzpicture} } \fi } + \frac{1}{12} {  \ifmmode \usebox{\fgbananathreewithbubble} \else \newsavebox{\fgbananathreewithbubble} \savebox{\fgbananathreewithbubble}{ \begin{tikzpicture}[x=1ex,y=1ex,baseline={([yshift=-.5ex]current bounding box.center)}] \coordinate (vm); \coordinate [left=1 of vm] (v0); \coordinate [right=1 of vm] (v1); \coordinate [right=.7 of v1] (o); \draw (v0) -- (v1); \draw (vm) circle(1); \draw (o) circle(.7); \filldraw (v0) circle (1pt); \filldraw (v1) circle (1pt); \end{tikzpicture} } \fi } + \frac{1}{16} {  \ifmmode \usebox{\fgdblcontractedpropellerthree} \else \newsavebox{\fgdblcontractedpropellerthree} \savebox{\fgdblcontractedpropellerthree}{ \begin{tikzpicture}[x=1ex,y=1ex,baseline={([yshift=-.5ex]current bounding box.center)}] \coordinate (v) ; \def \rad {1.5}; \coordinate (s1) at ([shift=(0:1.2)]v); \coordinate (s2) at ([shift=(120:\rad)]v); \coordinate (s3) at ([shift=(240:\rad)]v); \coordinate [right=.7 of s1] (o); \draw (v) to[out=180,in=210] (s2) to[out=30,in=60] (v); \draw (v) to[out=300,in=330] (s3) to[out=150,in=180] (v); \draw (v) -- (s1); \filldraw (v) circle (1pt); \filldraw (s1) circle (1pt); \draw (o) circle(.7); \end{tikzpicture} } \fi } + \frac{1}{48} {  \ifmmode \usebox{\fgthreerose} \else \newsavebox{\fgthreerose} \savebox{\fgthreerose}{ \begin{tikzpicture}[x=1ex,y=1ex,baseline={([yshift=-.5ex]current bounding box.center)}] \coordinate (v) ; \def \rad {1.5}; \coordinate (s1) at ([shift=(0:\rad)]v); \coordinate (s2) at ([shift=(120:\rad)]v); \coordinate (s3) at ([shift=(240:\rad)]v); \draw (v) to[out=60,in=90] (s1) to[out=-90,in=0-60] (v); \draw (v) to[out=180,in=210] (s2) to[out=30,in=60] (v); \draw (v) to[out=300,in=330] (s3) to[out=150,in=180] (v); \filldraw (v) circle (1pt); \end{tikzpicture} } \fi } + \cdots \Big) \\ = \sqrt{a} \Big( 1 + \left( \left( \frac{1}{8} + \frac{1}{12} \right)\lambda_3^2 a^3 + \frac{1}{8} \lambda_4 a^2 \right) \hbar \\ + \left( \frac{385}{1152} \lambda_3^4 a^6 + \frac{35}{64} \lambda_3^2 \lambda_4 a^5 + \frac{35}{384} \lambda_4^2 a^4 + \frac{7}{48} \lambda_3 \lambda_5 a^4 + \frac{1}{48} \lambda_6 a^3       \right) \hbar^2 + \cdots \Big) \end{gathered} \end{gather}
where $\phi_{\Sact}$ is a linear map which applies the Feynman rules, encoded in $\Sact$, to every graph. 
The expression $\Fop[\Sact(x)](\hbar) = \sum_{n=0}^\infty \hbar^n P_n(\lambda_3 a^{\frac32},\lambda_4 a^{\frac42},\ldots)$ is a sequence of polynomials $P_n$ of degree $2n$.

Of course, drawing all diagrams for a specific model and applying the zero-dimensional Feynman rules is not a very convenient way to calculate the power series $\Fop[\Sact(x)](\hbar)$ order by order. 
A more efficient way is to derive differential equations from the formal integral expression and solve these recursively \cite{cvitanovic1978number, argyres2001zero}. In some cases these differential equations can be solved exactly \cite{argyres2001zero} or sufficiently simple closed forms for the respective coefficients can be found. 
For example, this is possible for the zero-dimensional version of $\varphi^3$-theory:

\begin{expl}[The partition function of $\varphi^3$-theory]
\label{expl:phi3theoryexpansion}
In $\varphi^3$-theory the potential takes the form $V(x)= \frac{x^3}{3!}$, that means $\Sact(x)= -\frac{x^2}{2}+\frac{x^3}{3!}$. From Definition \ref{def:formalintegral} it follows that,
\begin{align*} Z^{\varphi^3}(\hbar) &= \Fop\left[-\frac{x^2}{2}+\frac{x^3}{3!}\right](\hbar) = \sum_{n=0}^\infty \hbar^n (2n-1)!! [x^{2n}] e^{\frac{x^3}{3!\hbar }} = \sum_{n=0}^\infty \hbar^{n} \frac{(6n-1)!!}{(3!)^{2n} (2n)!}, \end{align*}
where we were able to expand the expression, because
\begin{align*} [x^{6n}] e^{\frac{x^3}{3!\hbar }} &= \frac{1}{(3!)^{2n} \hbar^{2n} (2n)!} && \\ [x^{6n+k}] e^{\frac{x^3}{3!\hbar }} &= 0 && k=1,2,3,4,5. \end{align*}
The diagrammatic expansion starts with
\begin{gather*} Z^{\varphi^3}(\hbar) = \phi_{\Sact} \Big( 1 + \frac18 {  \ifmmode \usebox{\fghandle} \else \newsavebox{\fghandle} \savebox{\fghandle}{ \begin{tikzpicture}[x=1ex,y=1ex,baseline={([yshift=-.5ex]current bounding box.center)}] \coordinate (v0); \coordinate [right=1.5 of v0] (v1); \coordinate [left=.7 of v0] (i0); \coordinate [right=.7 of v1] (o0); \draw (v0) -- (v1); \filldraw (v0) circle (1pt); \filldraw (v1) circle (1pt); \draw (i0) circle(.7); \draw (o0) circle(.7); \end{tikzpicture} } \fi } + \frac{1}{12} {  \ifmmode \usebox{\fgbananathree} \else \newsavebox{\fgbananathree} \savebox{\fgbananathree}{ \begin{tikzpicture}[x=1ex,y=1ex,baseline={([yshift=-.5ex]current bounding box.center)}] \coordinate (vm); \coordinate [left=1 of vm] (v0); \coordinate [right=1 of vm] (v1); \draw (v0) -- (v1); \draw (vm) circle(1); \filldraw (v0) circle (1pt); \filldraw (v1) circle (1pt); \end{tikzpicture} } \fi } \\ + \frac{1}{128} {  \ifmmode \usebox{\fgtwohandles} \else \newsavebox{\fgtwohandles} \savebox{\fgtwohandles}{ \begin{tikzpicture}[x=1ex,y=1ex,baseline={([yshift=-.5ex]current bounding box.center)}] \coordinate (v); \coordinate [above=1.2 of v] (v01); \coordinate [right=1.5 of v01] (v11); \coordinate [left=.7 of v01] (i01); \coordinate [right=.7 of v11] (o01); \draw (v01) -- (v11); \filldraw (v01) circle (1pt); \filldraw (v11) circle (1pt); \draw (i01) circle(.7); \draw (o01) circle(.7); \coordinate [below=1.2 of v] (v02); \coordinate [right=1.5 of v02] (v12); \coordinate [left=.7 of v02] (i02); \coordinate [right=.7 of v12] (o02); \draw (v02) -- (v12); \filldraw (v02) circle (1pt); \filldraw (v12) circle (1pt); \draw (i02) circle(.7); \draw (o02) circle(.7); \end{tikzpicture} } \fi } + \frac{1}{288} {  \ifmmode \usebox{\fgtwobananasthree} \else \newsavebox{\fgtwobananasthree} \savebox{\fgtwobananasthree}{ \begin{tikzpicture}[x=1ex,y=1ex,baseline={([yshift=-.5ex]current bounding box.center)}] \coordinate (v); \coordinate [above=1.2 of v](v01); \coordinate [right=2 of v01] (v11); \coordinate [right=1 of v01] (vm1); \draw (v01) -- (v11); \draw (vm1) circle(1); \filldraw (v01) circle (1pt); \filldraw (v11) circle (1pt); \coordinate [below=1.2 of v](v02); \coordinate [right=2 of v02] (v12); \coordinate [right=1 of v02] (vm2); \draw (v02) -- (v12); \draw (vm2) circle(1); \filldraw (v02) circle (1pt); \filldraw (v12) circle (1pt); \end{tikzpicture} } \fi } + \frac{1}{96} {  \ifmmode \usebox{\fgbananathreeandhandle} \else \newsavebox{\fgbananathreeandhandle} \savebox{\fgbananathreeandhandle}{ \begin{tikzpicture}[x=1ex,y=1ex,baseline={([yshift=-.5ex]current bounding box.center)}] \coordinate (v); \coordinate [above=1.2 of v] (vm1); \coordinate [left=1 of vm1] (v01); \coordinate [right=1 of vm1] (v11); \draw (v01) -- (v11); \draw (vm1) circle(1); \filldraw (v01) circle (1pt); \filldraw (v11) circle (1pt); \coordinate [below=1.2 of v] (vm2); \coordinate [left=.75 of vm2] (v02); \coordinate [right=.75 of vm2] (v12); \coordinate [left=.7 of v02] (i02); \coordinate [right=.7 of v12] (o02); \draw (v02) -- (v12); \filldraw (v02) circle (1pt); \filldraw (v12) circle (1pt); \draw (i02) circle(.7); \draw (o02) circle(.7); \end{tikzpicture} } \fi } + \frac{1}{48} {  \ifmmode \usebox{\fgpropellerthree} \else \newsavebox{\fgpropellerthree} \savebox{\fgpropellerthree}{ \begin{tikzpicture}[x=1ex,y=1ex,baseline={([yshift=-.5ex]current bounding box.center)}] \coordinate (v) ; \def \n {3}; \def \rad {1.2}; \def \rud {1.9}; \foreach \s in {1,...,\n} { \def \angle {360/\n*(\s - 1)}; \coordinate (s) at ([shift=({\angle}:\rad)]v); \coordinate (u) at ([shift=({\angle}:\rud)]v); \draw (v) -- (s); \filldraw (s) circle (1pt); \draw (u) circle (.7); } \filldraw (v) circle (1pt); \end{tikzpicture} } \fi } \\ + \frac{1}{16} {  \ifmmode \usebox{\fgdblhandle} \else \newsavebox{\fgdblhandle} \savebox{\fgdblhandle}{ \begin{tikzpicture}[x=1ex,y=1ex,baseline={([yshift=-.5ex]current bounding box.center)}] \coordinate (v) ; \def \n {2}; \def \rad {.7}; \def \rud {2}; \def \rid {2.7}; \foreach \s in {1,...,\n} { \def \angle {360/\n*(\s - 1)}; \coordinate (s) at ([shift=({\angle}:\rad)]v); \coordinate (u) at ([shift=({\angle}:\rud)]v); \coordinate (t) at ([shift=({\angle}:\rid)]v); \draw (s) -- (u); \filldraw (u) circle (1pt); \filldraw (s) circle (1pt); \draw (t) circle (.7); } \draw (v) circle(\rad); \end{tikzpicture} } \fi } + \frac{1}{16} {  \ifmmode \usebox{\fgdbleye} \else \newsavebox{\fgdbleye} \savebox{\fgdbleye}{ \begin{tikzpicture}[x=1ex,y=1ex,baseline={([yshift=-.5ex]current bounding box.center)}] \coordinate (v0); \coordinate[above left=1.5 of v0] (v1); \coordinate[below left=1.5 of v0] (v2); \coordinate[above right=1.5 of v0] (v3); \coordinate[below right=1.5 of v0] (v4); \draw (v1) to[bend left=80] (v2); \draw (v1) to[bend right=80] (v2); \draw (v3) to[bend right=80] (v4); \draw (v3) to[bend left=80] (v4); \draw (v1) -- (v3); \draw (v2) -- (v4); \filldraw (v1) circle(1pt); \filldraw (v2) circle(1pt); \filldraw (v3) circle(1pt); \filldraw (v4) circle(1pt); \end{tikzpicture} } \fi } + \frac{1}{8} {  \ifmmode \usebox{\fgcloseddunce} \else \newsavebox{\fgcloseddunce} \savebox{\fgcloseddunce}{ \begin{tikzpicture}[x=1ex,y=1ex,baseline={([yshift=-.5ex]current bounding box.center)}] \coordinate (v0); \coordinate[right=.5 of v0] (v3); \coordinate[right=1.5 of v3] (v4); \coordinate[above left=1.5 of v0] (v1); \coordinate[below left=1.5 of v0] (v2); \coordinate[right=.7 of v4] (o); \draw (v3) to[bend left=20] (v2); \draw (v3) to[bend right=20] (v1); \draw (v1) to[bend right=80] (v2); \draw (v1) to[bend left=80] (v2); \draw (v3) -- (v4); \filldraw (v1) circle(1pt); \filldraw (v2) circle(1pt); \filldraw (v3) circle(1pt); \filldraw (v4) circle(1pt); \draw (o) circle(.7); \end{tikzpicture} } \fi } + \frac{1}{24} {  \begin{tikzpicture}[x=1ex,y=1ex,baseline={([yshift=-.5ex]current bounding box.center)}] \coordinate (v) ; \def \n {3}; \def \rad {1.2}; \foreach \s in {1,...,\n} { \def \angle {360/\n*(\s - 1)}; \coordinate (s) at ([shift=({\angle}:\rad)]v); \draw (v) -- (s); \filldraw (s) circle (1pt); } \draw (v) circle (\rad); \filldraw (v) circle (1pt); \end{tikzpicture} } + \ldots \Big) \\ = 1 +\left( \frac{1}{8} + \frac{1}{12} \right) \hbar + \frac{385}{1152} \hbar^2 + \ldots                    \end{gather*}
which is the same as the expansion in \eqref{eqn:generalexpansion} with $a= \lambda_3 = 1$ and all other $\lambda_k=0$.
\end{expl}

\begin{expl}[Generating function of all multigraphs with given excess]
\label{expl:allgraphs}
The generating function of all graphs without one or two-valent vertices is given by the partition function of the `theory' with the potential $V(x)=\sum_{k=3}\frac{x^k}{k!} = e^x - 1 - x - \frac{x^2}{2}$. Therefore,
\begin{align*} Z^\text{all}(\hbar) = \Fop\left[ -x^2 - x - 1 + e^x \right](\hbar) &= \sum_{n=0}^\infty \hbar^n (2n-1)!! [x^{2n}] e^{\frac{1}{\hbar} \left( e^x - 1 - x - \frac{x^2}{2} \right)}. \end{align*}
Here, as in many cases where $V(x)$ is not merely a monomial, the extraction of coefficients is more difficult. 
Still, the power series expansion in $\hbar$ can be computed conveniently with the methods which will be established in the next section.
The diagrammatic expansion is equivalent to the one given in \eqref{eqn:generalexpansion} with $a=1$ and all the $\lambda_k=1$:
\begin{align*} Z^\text{all}(\hbar) &= 1 + \frac{1}{3} \hbar + \frac{41}{36} \hbar^2 + \cdots \end{align*}
Whereas this example has no direct interpretation in QFT, except maybe for the case of gravity, where vertices with arbitrary valency appear, it shows that formal integrals are quite powerful at enumerating general graphs. Hence, the techniques of zero-dimensional QFT and formal integrals can be applied to a much broader class of topics, which evolve around graph enumeration. Especially promising is the application to the theory of complex networks \cite{albert2002statistical}. We will elaborate on applications of formal integrals to these problems in a future publication \cite{borinsky2017graph}.
\end{expl}

\begin{expl}[Zero-dimensional sine-Gordon model]
\label{expl:sinegordon}
For a more exotic zero-dimensional QFT take $\Sact(x)= -\frac{\sin^2(x)}{2}$ or $V(x) = \frac{x^2}{2} -\frac{\sin^2(x)}{2} = 4\frac{x^4}{4!} - 16 \frac{x^6}{6!}+ 64 \frac{x^8}{8!} +\cdots$. This can be seen as the potential of a zero-dimensional version of the sine-Gordon model \cite{cherman2014decoding}.
\begin{align*} \Fop\left[ -\frac{\sin^2(x)}{2} \right](\hbar) &= \sum_{n=0}^\infty \hbar^n (2n-1)!! [x^{2n}] e^{\frac{1}{\hbar} \left( \frac{x^2}{2} -\frac{\sin^2(x)}{2} \right)}. \end{align*}
The diagrammatic expansion starts with
\begin{gather*} Z^{\text{sine-Gordon}}(\hbar) = \phi_{\Sact} \Big( 1 + \frac{1}{8} {  \ifmmode \usebox{\fgtadpoletwo} \else \newsavebox{\fgtadpoletwo} \savebox{\fgtadpoletwo}{ \begin{tikzpicture}[x=1ex,y=1ex,baseline={([yshift=-.5ex]current bounding box.center)}] \coordinate (vm); \coordinate [left=.7 of vm] (v0); \coordinate [right=.7 of vm] (v1); \draw (v0) circle(.7); \draw (v1) circle(.7); \filldraw (vm) circle (1pt); \end{tikzpicture} } \fi } + \frac{1}{128} {  \ifmmode \usebox{\fgtwotadpoletwos} \else \newsavebox{\fgtwotadpoletwos} \savebox{\fgtwotadpoletwos}{ \begin{tikzpicture}[x=1ex,y=1ex,baseline={([yshift=-.5ex]current bounding box.center)}] \coordinate (v); \coordinate [above=1.2 of v] (vm1); \coordinate [left=.7 of vm1] (v01); \coordinate [right=.7 of vm1] (v11); \draw (v01) circle(.7); \draw (v11) circle(.7); \filldraw (vm1) circle (1pt); \coordinate [below=1.2 of v] (vm2); \coordinate [left=.7 of vm2] (v02); \coordinate [right=.7 of vm2] (v12); \draw (v02) circle(.7); \draw (v12) circle(.7); \filldraw (vm2) circle (1pt); \end{tikzpicture} } \fi } + \frac{1}{48} {  \ifmmode \usebox{\fgbananafour} \else \newsavebox{\fgbananafour} \savebox{\fgbananafour}{ \begin{tikzpicture}[x=1ex,y=1ex,baseline={([yshift=-.5ex]current bounding box.center)}] \coordinate (vm); \coordinate [left=1 of vm] (v0); \coordinate [right=1 of vm] (v1); \draw (v0) to[bend left=45] (v1); \draw (v0) to[bend right=45] (v1); \draw (vm) circle(1); \filldraw (v0) circle (1pt); \filldraw (v1) circle (1pt); \end{tikzpicture} } \fi } + \frac{1}{16} {  \ifmmode \usebox{\fgthreebubble} \else \newsavebox{\fgthreebubble} \savebox{\fgthreebubble}{ \begin{tikzpicture}[x=1ex,y=1ex,baseline={([yshift=-.5ex]current bounding box.center)}] \coordinate (vm); \coordinate [left=.7 of vm] (v0); \coordinate [right=.7 of vm] (v1); \coordinate [left=.7 of v0] (vc1); \coordinate [right=.7 of v1] (vc2); \draw (vc1) circle(.7); \draw (vc2) circle(.7); \draw (vm) circle(.7); \filldraw (v0) circle (1pt); \filldraw (v1) circle (1pt); \end{tikzpicture} } \fi } + \frac{1}{48} {  \ifmmode \usebox{\fgthreerose} \else \newsavebox{\fgthreerose} \savebox{\fgthreerose}{ \begin{tikzpicture}[x=1ex,y=1ex,baseline={([yshift=-.5ex]current bounding box.center)}] \coordinate (v) ; \def \rad {1.5}; \coordinate (s1) at ([shift=(0:\rad)]v); \coordinate (s2) at ([shift=(120:\rad)]v); \coordinate (s3) at ([shift=(240:\rad)]v); \draw (v) to[out=60,in=90] (s1) to[out=-90,in=0-60] (v); \draw (v) to[out=180,in=210] (s2) to[out=30,in=60] (v); \draw (v) to[out=300,in=330] (s3) to[out=150,in=180] (v); \filldraw (v) circle (1pt); \end{tikzpicture} } \fi } + \ldots \Big) \\ = 1 + \frac{4}{8}\hbar + \left( \frac{4^2}{128} + \frac{4^2}{48} + \frac{4^2}{16} - \frac{16}{48} \right) \hbar^2 + \ldots \\ = 1 + \frac12 \hbar + \frac{9}{8} \hbar^2 + \ldots \end{gather*}
which is equal to the expansion in \eqref{eqn:generalexpansion} with $\lambda_{2k} = (-1)^k 2^{2k-2}$ and $\lambda_{2k+1}=0$.

\end{expl}

\begin{expl}[Stirling's QFT]
The following example is widely used in physics, although we did not find any reference with the interpretation of a zero-dimensional QFT. As a matrix model it is known as Penner's model \cite{penner1988perturbative}. It agrees with Stirling's asymptotic expansion of the $\Gamma$-function \cite[A. D]{kontsevich1992intersection}:

From Euler's integral for the $\Gamma$-function we can deduce with the change of variables $t \rightarrow N e^x$ ,
\begin{align*} \frac{\Gamma(N)}{\sqrt{\frac{2\pi}{N}} \left(\frac{N}{e}\right)^N } &= \frac{1}{\sqrt{\frac{2\pi}{N}} \left(\frac{N}{e}\right)^N } \int_0^\infty dt t^{N-1} e^{-t} = \frac{e^{N}}{\sqrt{\frac{2\pi}{N}}} \int_\R dx e^{-N e^x + Nx } \\ &= \int_\R \frac{dx}{\sqrt{2\pi \frac{1}{N}}} e^{N\left(1-x-e^x\right) } , \end{align*}
This is the correction term of Stirling's formula expressed as a zero-dimensional QFT. Note, that the integral is actually convergent in this case, whereas the expansion in $\frac{1}{N}$ is not.
Therefore,
\begin{align*} Z^\text{Stirling}\left(\frac{1}{N}\right)&:= \Fop[1+x-e^x]\left(\frac{1}{N}\right) \\ &= \sum_{n=0}^\infty \frac{1}{N^n} (2n-1)!! [x^{2n}] e^{N \left( 1+x+\frac{x^2}{2}-e^x\right)}. \end{align*}
We can use the well-known Stirling expansion in terms of the Bernoulli numbers $B_k$ to state the power series more explicitly \cite{whittaker1996course}:
\begin{align*} \Fop[1+x-e^x]\left(\frac{1}{N}\right) = e^{ \sum_{k=1}^\infty \frac{B_{k+1}}{k (k+1)} \frac{1}{N^{k}}}, \end{align*}

Interestingly, Proposition \ref{prop:diagraminterpretation} provides us with a combinatorial interpretation of the Stirling expansion. 
We can directly use expansion \eqref{eqn:generalexpansion} to calculate the first terms:
\begin{gather*} Z^\text{Stirling} \left(\frac{1}{N}\right):= \phi_{\Sact} \Big( 1 + \frac18 {  \ifmmode \usebox{\fghandle} \else \newsavebox{\fghandle} \savebox{\fghandle}{ \begin{tikzpicture}[x=1ex,y=1ex,baseline={([yshift=-.5ex]current bounding box.center)}] \coordinate (v0); \coordinate [right=1.5 of v0] (v1); \coordinate [left=.7 of v0] (i0); \coordinate [right=.7 of v1] (o0); \draw (v0) -- (v1); \filldraw (v0) circle (1pt); \filldraw (v1) circle (1pt); \draw (i0) circle(.7); \draw (o0) circle(.7); \end{tikzpicture} } \fi } + \frac{1}{12} {  \ifmmode \usebox{\fgbananathree} \else \newsavebox{\fgbananathree} \savebox{\fgbananathree}{ \begin{tikzpicture}[x=1ex,y=1ex,baseline={([yshift=-.5ex]current bounding box.center)}] \coordinate (vm); \coordinate [left=1 of vm] (v0); \coordinate [right=1 of vm] (v1); \draw (v0) -- (v1); \draw (vm) circle(1); \filldraw (v0) circle (1pt); \filldraw (v1) circle (1pt); \end{tikzpicture} } \fi } + \frac{1}{8} {  \ifmmode \usebox{\fgtadpoletwo} \else \newsavebox{\fgtadpoletwo} \savebox{\fgtadpoletwo}{ \begin{tikzpicture}[x=1ex,y=1ex,baseline={([yshift=-.5ex]current bounding box.center)}] \coordinate (vm); \coordinate [left=.7 of vm] (v0); \coordinate [right=.7 of vm] (v1); \draw (v0) circle(.7); \draw (v1) circle(.7); \filldraw (vm) circle (1pt); \end{tikzpicture} } \fi } + \ldots \Big) \\ = 1 + \left( \frac18 (-1)^2 + \frac{1}{12} (-1)^2 + \frac{1}{8} (-1)^1 \right) \frac{1}{N} \\ + \left( \frac{385}{1152} (-1)^4 + \frac{35}{64} (-1)^3 + \frac{35}{384} (-1)^2 + \frac{7}{48} (-1)^2 + \frac{1}{48} (-1)^1 \right) \frac{1}{N^2} + \ldots \\ =1 + \frac{1}{12} \frac{1}{N} + \frac{1}{288} \frac{1}{N^2} + \ldots \end{gather*}
which results in the well-known asymptotic expansion of the gamma function \cite{whittaker1996course},
\begin{align*} \Gamma(N) \underset{N\rightarrow \infty}{\sim} \sqrt{\frac{2\pi}{N}} \left(\frac{N}{e}\right)^n \left( 1 +\frac{1}{12 N} + \frac{1}{288 N^2}+\ldots\right). \end{align*}
Moreover, taking the logarithm of $\Fop[1+x-e^x]\left(\frac{1}{N}\right)$ and using the fact that the $n$-th Bernoulli number vanishes if $n$ is odd and greater than $1$, gives us the combinatorial identities,
\begin{align*} \frac{B_{2n}}{2 n (2n-1)} &= \sum_{|L(\Gamma)| = 2n} \frac{(-1)^{|V(\Gamma)|}}{|\Aut\Gamma|} & 0 &= \sum_{|L(\Gamma)| = 2n+1} \frac{(-1)^{|V(\Gamma)|}}{|\Aut\Gamma|}, \end{align*}
where the sum is over all \textit{connected} graphs $\Gamma$ with a fixed number of loops, denoted by $|L(\Gamma)|$.
\end{expl}
\subsection{Representation as affine hyperelliptic curve}

Calculating the coefficients of the power series given in Definition \ref{def:formalintegral} using the expression in eq.\ \eqref{eqn:formalintegralpwrsrs} directly is inconvenient, because an intermediate bivariate quantity $e^{\frac{1}{\hbar}V(x)}$ needs to be expanded in $x$ and in $\hbar^{-1}$.

A form that is computationally more accessible can be achieved by \textit{a formal change of variables}. 
Recall that we set $\Sact(x) = -\frac{x^2}{2a} +V(x)$. Expanding the exponential in eq.\ \eqref{eqn:formalintegralpwrsrs} gives
\begin{align*} \mathcal{F}[\Sact(x)](\hbar) &= \sqrt{a} \sum_{n=0}^\infty \sum_{k=0}^\infty \hbar^{n-k} a^{n} (2n-1)!! [x^{2n}] \frac{V(x)^k}{k!}. \intertext{This can be seen as the zero-dimensional analog of Dyson's series \cite{itzykson2006quantum}. Shifting the summation over $n$ and substituting $V(x) = \frac{x^2}{2a } + \Sact(x)$ results in,} &= \sum_{n=0}^\infty \sum_{k=0}^\infty 2^{-k} a^{n+\frac12} \hbar^{n} \frac{(2(n+k)-1)!!}{k!} [x^{2n}] \left(1 + \frac{2 a}{x^2} \Sact(x) \right)^k \intertext{ Because $2^{-k} \frac{(2(n+k)-1)!!}{(2n-1)!! k!} = { n+k-\frac12 \choose k }$ , it follows that }  &= \sum_{n=0}^\infty a^{n+\frac12} \hbar^{n} (2n-1)!! [x^{2n}] \sum_{k=0}^\infty { n+k-\frac12 \choose k } \left(1 + \frac{2a}{x^2}\Sact(x) \right)^k, \intertext{and using $\sum_{k=0}^\infty { \alpha +k-1 \choose k } x^k = \frac{1}{(1-x)^\alpha} $ gives, } &= \sum_{n=0}^\infty \hbar^{n} (2n-1)!! [x^{2n}] \left( \frac{x}{\sqrt{-2 \Sact(x)}} \right)^{2n+1}. \end{align*}
By the Lagrange inversion formula $[y^n] g(y) = \frac{1}{n} [x^{n-1}] \left(\frac{x}{f(x)}\right)^n$, where $f(g(y)) = y$, this is equivalent to
\begin{prop}
\label{prop:formalchangeofvar}
If $\Sact(x) = -\frac{x^2}{2a} + V(x)$, then
\begin{align} \mathcal{F}[\Sact(x)](\hbar)&= \sum_{n=0}^\infty \hbar^{n} (2n+1)!! [y^{2n+1}] x(y) \end{align}
where $x(y)$ is the unique power series solution of $y = \sqrt{-2 S(x(y))}$, where the positive branch of the square root is taken. 
\end{prop}

Note, that this can be seen as a formal change of variables for the formal integral \eqref{eqn:formalfuncintegral1}. The advantage of using the Lagrange inversion formula is that it makes clear that the formal change of variables in Proposition \ref{prop:formalchangeofvar} does not depend on the analyticity or injectiveness properties of $\Sact(x)$.

Care must be taken to ensure that $x(y)$ is interpreted as a formal power series in $\R[[y]]$, whereas $\Sact(x)$ is in $\R[[x]]$. We hope that the slight abuse of notation, where we interpret $x$ as a power series or as a variable is transparent for the reader. 

If $\Sact(x)$ is a polynomial, the equation $y = \sqrt{-2 S(x(y))}$ can be interpreted as the definition of an \textit{affine hyperelliptic curve}, 
\begin{align} \frac{y^2}{2} = - \Sact(x) \end{align}
with at least one singular point or \textit{ordinary double point} at the origin, because $\Sact(x) = -\frac{x^2}{2a}+\cdots$. If $\Sact(x)$ is not a polynomial, but an entire function, it is a \textit{generalized affine hyperelliptic curve}.

This interpretation shows a surprising similarity to the theory of \textit{topological recursion} \cite{eynard2007invariants}. The affine complex curve is called the spectral curve in this realm, as it is associated to the eigenvalue distribution of a random matrix model. 
In the theory of topological recursion the \textit{branch-cut} singularities of the expansion of the curve play a vital role. They will also be important for the extraction of asymptotics from formal integrals presented in the next section. 
\begin{figure}
\centering
\begin{subfigure}[t]{0.4\textwidth}
\begin{tikzpicture} \begin{axis}[ xlabel={$x$}, ylabel={$y$}, xmin=-5, xmax=5, ymin=-5, ymax=5, axis on top, width=\figurewidth, height=\figureheight, xmajorgrids, ymajorgrids ] \addplot [black] table {%
-5 -8.16496580927726
-4.89795918367347 -7.94716036730067
-4.79591836734694 -7.73116240380083
-4.69387755102041 -7.51698638636357
-4.59183673469388 -7.30464712324638
-4.48979591836735 -7.09415977661449
-4.38775510204082 -6.88553987649383
-4.28571428571429 -6.67880333549125
-4.18367346938776 -6.47396646433575
-4.08163265306122 -6.27104598829997
-3.97959183673469 -6.07005906456591
-3.87755102040816 -5.87102330060463
-3.77551020408163 -5.67395677364627
-3.6734693877551 -5.47887805132332
-3.57142857142857 -5.28580621357825
-3.46938775510204 -5.09476087593504
-3.36734693877551 -4.90576221424379
-3.26530612244898 -4.71883099101839
-3.16326530612245 -4.53398858349917
-3.06122448979592 -4.3512570135858
-2.95918367346939 -4.17065897980091
-2.85714285714286 -3.99221789146155
-2.75510204081633 -3.81595790525484
-2.6530612244898 -3.64190396443527
-2.55102040816327 -3.47008184088574
-2.44897959183673 -3.30051818031135
-2.3469387755102 -3.1332405508664
-2.24489795918367 -2.9682774955503
-2.14285714285714 -2.80565858874847
-2.04081632653061 -2.64541449734037
-1.93877551020408 -2.48757704684978
-1.83673469387755 -2.33217929317319
-1.73469387755102 -2.17925560049212
-1.63265306122449 -2.02884172605635
-1.53061224489796 -1.88097491261865
-1.42857142857143 -1.7356939894113
-1.3265306122449 -1.59303948268164
-1.22448979591837 -1.4530537369536
-1.12244897959184 -1.31578104835766
-1.02040816326531 -1.18126781157854
-0.918367346938775 -1.04956268221574
-0.816326530612245 -0.9207167566435
-0.714285714285714 -0.794783771805981
-0.612244897959184 -0.671820327801919
-0.510204081632653 -0.551886136618102
-0.408163265306122 -0.435044300983467
-0.306122448979592 -0.321361628061933
-0.204081632653061 -0.210908983617208
-0.102040816326531 -0.103761693411383
0 -0
0 0
0.0689655172413793 0.068168201228676
0.137931034482759 0.134722897075969
0.206896551724138 0.199634747866348
0.275862068965517 0.262872955775982
0.344827586206897 0.32440514870615
0.413793103448276 0.384197251400797
0.482758620689655 0.4422133420096
0.551724137931034 0.498415491985822
0.620689655172414 0.552763586831413
0.689655172413793 0.605215124744396
0.758620689655172 0.655724989665408
0.827586206896552 0.704245194535252
0.896551724137931 0.750724589729521
0.96551724137931 0.795108530585403
1.03448275862069 0.837338496621082
1.10344827586207 0.877351653391487
1.17241379310345 0.915080345820703
1.24137931034483 0.950451509158481
1.31034482758621 0.983385980230529
1.37931034482759 1.01379768711839
1.44827586206897 1.04159268943377
1.51724137931034 1.06666803340216
1.58620689655172 1.08891037526024
1.6551724137931 1.10819431185432
1.72413793103448 1.12438033709903
1.79310344827586 1.13731231452828
1.86206896551724 1.14681431554406
1.93103448275862 1.15268661381773
2 1.15470053837925
}; \addplot [black, dashed] table {%
2 1.15470053837925
2.05263157894737 1.15347936794798
2.10526315789474 1.14972353952158
2.15789473684211 1.14328056083509
2.21052631578947 1.13397607335728
2.26315789473684 1.12160889445706
2.31578947368421 1.10594447628423
2.36842105263158 1.08670610780792
2.42105263157895 1.06356281657479
2.47368421052632 1.0361122994584
2.52631578947368 1.00385610219666
2.57894736842105 0.966162205631642
2.63157894736842 0.922206110592714
2.68421052631579 0.870872892386199
2.73684210526316 0.810582674828232
2.78947368421053 0.738949630783656
2.84210526315789 0.652023664684755
2.89473684210526 0.542233890915719
2.94736842105263 0.390388484187592
3 0
}; \addplot [black, dashed] table {%
-5 8.16496580927726
-4.89795918367347 7.94716036730067
-4.79591836734694 7.73116240380083
-4.69387755102041 7.51698638636357
-4.59183673469388 7.30464712324638
-4.48979591836735 7.09415977661449
-4.38775510204082 6.88553987649383
-4.28571428571429 6.67880333549125
-4.18367346938776 6.47396646433575
-4.08163265306122 6.27104598829997
-3.97959183673469 6.07005906456591
-3.87755102040816 5.87102330060463
-3.77551020408163 5.67395677364627
-3.6734693877551 5.47887805132332
-3.57142857142857 5.28580621357825
-3.46938775510204 5.09476087593504
-3.36734693877551 4.90576221424379
-3.26530612244898 4.71883099101839
-3.16326530612245 4.53398858349917
-3.06122448979592 4.3512570135858
-2.95918367346939 4.17065897980091
-2.85714285714286 3.99221789146155
-2.75510204081633 3.81595790525484
-2.6530612244898 3.64190396443527
-2.55102040816327 3.47008184088574
-2.44897959183673 3.30051818031135
-2.3469387755102 3.1332405508664
-2.24489795918367 2.9682774955503
-2.14285714285714 2.80565858874847
-2.04081632653061 2.64541449734037
-1.93877551020408 2.48757704684978
-1.83673469387755 2.33217929317319
-1.73469387755102 2.17925560049212
-1.63265306122449 2.02884172605635
-1.53061224489796 1.88097491261865
-1.42857142857143 1.7356939894113
-1.3265306122449 1.59303948268164
-1.22448979591837 1.4530537369536
-1.12244897959184 1.31578104835766
-1.02040816326531 1.18126781157854
-0.918367346938775 1.04956268221574
-0.816326530612245 0.9207167566435
-0.714285714285714 0.794783771805981
-0.612244897959184 0.671820327801919
-0.510204081632653 0.551886136618102
-0.408163265306122 0.435044300983467
-0.306122448979592 0.321361628061933
-0.204081632653061 0.210908983617208
-0.102040816326531 0.103761693411383
0 0
0 -0
0.0689655172413793 -0.068168201228676
0.137931034482759 -0.134722897075969
0.206896551724138 -0.199634747866348
0.275862068965517 -0.262872955775982
0.344827586206897 -0.32440514870615
0.413793103448276 -0.384197251400797
0.482758620689655 -0.4422133420096
0.551724137931034 -0.498415491985822
0.620689655172414 -0.552763586831413
0.689655172413793 -0.605215124744396
0.758620689655172 -0.655724989665408
0.827586206896552 -0.704245194535252
0.896551724137931 -0.750724589729521
0.96551724137931 -0.795108530585403
1.03448275862069 -0.837338496621082
1.10344827586207 -0.877351653391487
1.17241379310345 -0.915080345820703
1.24137931034483 -0.950451509158481
1.31034482758621 -0.983385980230529
1.37931034482759 -1.01379768711839
1.44827586206897 -1.04159268943377
1.51724137931034 -1.06666803340216
1.58620689655172 -1.08891037526024
1.6551724137931 -1.10819431185432
1.72413793103448 -1.12438033709903
1.79310344827586 -1.13731231452828
1.86206896551724 -1.14681431554406
1.93103448275862 -1.15268661381773
2 -1.15470053837925
2 -1.15470053837925
2.05263157894737 -1.15347936794798
2.10526315789474 -1.14972353952158
2.15789473684211 -1.14328056083509
2.21052631578947 -1.13397607335728
2.26315789473684 -1.12160889445706
2.31578947368421 -1.10594447628423
2.36842105263158 -1.08670610780792
2.42105263157895 -1.06356281657479
2.47368421052632 -1.0361122994584
2.52631578947368 -1.00385610219666
2.57894736842105 -0.966162205631642
2.63157894736842 -0.922206110592714
2.68421052631579 -0.870872892386199
2.73684210526316 -0.810582674828232
2.78947368421053 -0.738949630783656
2.84210526315789 -0.652023664684755
2.89473684210526 -0.542233890915719
2.94736842105263 -0.390388484187592
3 -0
}; \addplot [black, dotted] table {%
-5 1.15470053837925
5 1.15470053837925
}; \end{axis} \end{tikzpicture}
\subcaption{Plot of the elliptic curve $\frac{y^2}{2} = \frac{x^2}{2} - \frac{x^3}{3!}$, which can be associated to the perturbative expansion of zero-dimensional $\varphi^3$-theory. 
The dominant singularity can be found at $(x,y)=\left(2,\frac{2}{\sqrt{3}}\right)$.
}
\label{fig:curve_phi3}
\end{subfigure}
\quad
\begin{subfigure}[t]{0.4\textwidth}
\begin{tikzpicture} \begin{axis}[ xlabel={$x$}, ylabel={$y$}, xmin=-5, xmax=5, ymin=-5, ymax=5, axis on top, width=\figurewidth, height=\figureheight, xmajorgrids, ymajorgrids ] \addplot [black] table {%
-1.5707963267949 -1
-1.46246554563663 -0.99413795715436
-1.35413476447836 -0.976620555710087
-1.24580398332009 -0.947653171182802
-1.13747320216182 -0.907575419670957
-1.02914242100355 -0.856857176167589
-0.920811639845284 -0.796093065705644
-0.812480858687015 -0.725995491923131
-0.704150077528747 -0.647386284781828
-0.595819296370478 -0.561187065362382
-0.487488515212209 -0.46840844069979
-0.37915773405394 -0.370138155339914
-0.270826952895672 -0.267528338529221
-0.162496171737403 -0.161781996552765
-0.0541653905791344 -0.0541389085854175
0.0541653905791344 0.0541389085854175
0.162496171737403 0.161781996552765
0.270826952895672 0.267528338529221
0.379157734053941 0.370138155339915
0.487488515212209 0.46840844069979
0.595819296370478 0.561187065362382
0.704150077528747 0.647386284781828
0.812480858687016 0.725995491923131
0.920811639845284 0.796093065705644
1.02914242100355 0.856857176167589
1.13747320216182 0.907575419670957
1.24580398332009 0.947653171182803
1.35413476447836 0.976620555710087
1.46246554563663 0.99413795715436
1.5707963267949 1
}; \addplot [black, dashed] table {%
-1.5707963267949 1
-1.46246554563663 0.99413795715436
-1.35413476447836 0.976620555710087
-1.24580398332009 0.947653171182802
-1.13747320216182 0.907575419670957
-1.02914242100355 0.856857176167589
-0.920811639845284 0.796093065705644
-0.812480858687015 0.725995491923131
-0.704150077528747 0.647386284781828
-0.595819296370478 0.561187065362382
-0.487488515212209 0.46840844069979
-0.37915773405394 0.370138155339914
-0.270826952895672 0.267528338529221
-0.162496171737403 0.161781996552765
-0.0541653905791344 0.0541389085854175
0.0541653905791344 -0.0541389085854175
0.162496171737403 -0.161781996552765
0.270826952895672 -0.267528338529221
0.379157734053941 -0.370138155339915
0.487488515212209 -0.46840844069979
0.595819296370478 -0.561187065362382
0.704150077528747 -0.647386284781828
0.812480858687016 -0.725995491923131
0.920811639845284 -0.796093065705644
1.02914242100355 -0.856857176167589
1.13747320216182 -0.907575419670957
1.24580398332009 -0.947653171182803
1.35413476447836 -0.976620555710087
1.46246554563663 -0.99413795715436
1.5707963267949 -1
}; \addplot [black, dashed] table {%
-5 0.958924274663138
-4.93001625156724 0.9764125014873
-4.86003250313449 0.989120479855391
-4.79004875470173 0.996985994982286
-4.72006500626897 0.999970539457967
-4.65008125783621 0.998059501769267
-4.58009750940346 0.991262237833841
-4.5101137609707 0.979612025196152
-4.44013001253794 0.963165900109718
-4.37014626410518 0.942004378303225
-4.30016251567243 0.916231060797541
-4.23017876723967 0.885972126703468
-4.16019501880691 0.851375715483339
-4.09021127037416 0.812611201700743
-4.0202275219414 0.769868365808998
-3.95024377350864 0.723356465037984
-3.88026002507588 0.673303208927992
-3.81027627664313 0.619953644526123
-3.74029252821037 0.563568956702987
-3.67030877977761 0.504425189463042
-3.60032503134486 0.442811894508673
-3.5303412829121 0.379030713674278
-3.46035753447934 0.313393902170381
-3.39037378604658 0.246222799867555
-3.32039003761383 0.177846258104333
-3.25040628918107 0.108599029721034
-3.18042254074831 0.0388201302014438
-3.11043879231556 -0.0311488220542363
-3.0404550438828 -0.100965278199608
-2.97047129545004 -0.170287435967506
-2.90048754701728 -0.238775913040378
-2.83050379858453 -0.306095408573256
-2.76052005015177 -0.371916344734953
-2.69053630171901 -0.435916480230948
-2.62055255328625 -0.497782487908603
-2.5505688048535 -0.557211488721169
-2.48058505642074 -0.61391253454073
-2.41060130798798 -0.667608032560636
-2.34061755955523 -0.718035104313931
-2.27063381112247 -0.764946872654404
-2.20065006268971 -0.808113670399547
-2.13066631425695 -0.847324164718247
-2.0606825658242 -0.8823863917585
-1.99069881739144 -0.913128696449913
-1.92071506895868 -0.939400572879963
-1.85073132052593 -0.961073401129759
-1.78074757209317 -0.978041076961984
-1.71076382366041 -0.9902205312782
-1.64078007522765 -0.997552136802433
-1.5707963267949 -1
}; \addplot [black, dashed] table {%
-5 -0.958924274663138
-4.93001625156724 -0.9764125014873
-4.86003250313449 -0.989120479855391
-4.79004875470173 -0.996985994982286
-4.72006500626897 -0.999970539457967
-4.65008125783621 -0.998059501769267
-4.58009750940346 -0.991262237833841
-4.5101137609707 -0.979612025196152
-4.44013001253794 -0.963165900109718
-4.37014626410518 -0.942004378303225
-4.30016251567243 -0.916231060797541
-4.23017876723967 -0.885972126703468
-4.16019501880691 -0.851375715483339
-4.09021127037416 -0.812611201700743
-4.0202275219414 -0.769868365808998
-3.95024377350864 -0.723356465037984
-3.88026002507588 -0.673303208927992
-3.81027627664313 -0.619953644526123
-3.74029252821037 -0.563568956702987
-3.67030877977761 -0.504425189463042
-3.60032503134486 -0.442811894508673
-3.5303412829121 -0.379030713674278
-3.46035753447934 -0.313393902170381
-3.39037378604658 -0.246222799867555
-3.32039003761383 -0.177846258104333
-3.25040628918107 -0.108599029721034
-3.18042254074831 -0.0388201302014438
-3.11043879231556 0.0311488220542363
-3.0404550438828 0.100965278199608
-2.97047129545004 0.170287435967506
-2.90048754701728 0.238775913040378
-2.83050379858453 0.306095408573256
-2.76052005015177 0.371916344734953
-2.69053630171901 0.435916480230948
-2.62055255328625 0.497782487908603
-2.5505688048535 0.557211488721169
-2.48058505642074 0.61391253454073
-2.41060130798798 0.667608032560636
-2.34061755955523 0.718035104313931
-2.27063381112247 0.764946872654404
-2.20065006268971 0.808113670399547
-2.13066631425695 0.847324164718247
-2.0606825658242 0.8823863917585
-1.99069881739144 0.913128696449913
-1.92071506895868 0.939400572879963
-1.85073132052593 0.961073401129759
-1.78074757209317 0.978041076961984
-1.71076382366041 0.9902205312782
-1.64078007522765 0.997552136802433
-1.5707963267949 1
}; \addplot [black, dashed] table {%
1.5707963267949 1
1.64078007522765 0.997552136802433
1.71076382366041 0.9902205312782
1.78074757209317 0.978041076961984
1.85073132052593 0.961073401129759
1.92071506895868 0.939400572879963
1.99069881739144 0.913128696449913
2.0606825658242 0.8823863917585
2.13066631425695 0.847324164718246
2.20065006268971 0.808113670399547
2.27063381112247 0.764946872654404
2.34061755955523 0.718035104313931
2.41060130798798 0.667608032560636
2.48058505642074 0.613912534540729
2.5505688048535 0.557211488721168
2.62055255328626 0.497782487908603
2.69053630171901 0.435916480230948
2.76052005015177 0.371916344734952
2.83050379858453 0.306095408573256
2.90048754701728 0.238775913040378
2.97047129545004 0.170287435967506
3.0404550438828 0.100965278199608
3.11043879231556 0.0311488220542358
3.18042254074831 -0.0388201302014443
3.25040628918107 -0.108599029721034
3.32039003761383 -0.177846258104333
3.39037378604658 -0.246222799867555
3.46035753447934 -0.313393902170381
3.5303412829121 -0.379030713674278
3.60032503134486 -0.442811894508673
3.67030877977761 -0.504425189463043
3.74029252821037 -0.563568956702987
3.81027627664313 -0.619953644526123
3.88026002507588 -0.673303208927992
3.95024377350864 -0.723356465037984
4.0202275219414 -0.769868365808998
4.09021127037416 -0.812611201700743
4.16019501880691 -0.851375715483339
4.23017876723967 -0.885972126703468
4.30016251567243 -0.916231060797541
4.37014626410519 -0.942004378303225
4.44013001253794 -0.963165900109718
4.5101137609707 -0.979612025196152
4.58009750940346 -0.991262237833841
4.65008125783621 -0.998059501769267
4.72006500626897 -0.999970539457967
4.79004875470173 -0.996985994982286
4.86003250313449 -0.989120479855391
4.93001625156724 -0.9764125014873
5 -0.958924274663138
}; \addplot [black, dashed] table {%
1.5707963267949 -1
1.64078007522765 -0.997552136802433
1.71076382366041 -0.9902205312782
1.78074757209317 -0.978041076961984
1.85073132052593 -0.961073401129759
1.92071506895868 -0.939400572879963
1.99069881739144 -0.913128696449913
2.0606825658242 -0.8823863917585
2.13066631425695 -0.847324164718246
2.20065006268971 -0.808113670399547
2.27063381112247 -0.764946872654404
2.34061755955523 -0.718035104313931
2.41060130798798 -0.667608032560636
2.48058505642074 -0.613912534540729
2.5505688048535 -0.557211488721168
2.62055255328626 -0.497782487908603
2.69053630171901 -0.435916480230948
2.76052005015177 -0.371916344734952
2.83050379858453 -0.306095408573256
2.90048754701728 -0.238775913040378
2.97047129545004 -0.170287435967506
3.0404550438828 -0.100965278199608
3.11043879231556 -0.0311488220542358
3.18042254074831 0.0388201302014443
3.25040628918107 0.108599029721034
3.32039003761383 0.177846258104333
3.39037378604658 0.246222799867555
3.46035753447934 0.313393902170381
3.5303412829121 0.379030713674278
3.60032503134486 0.442811894508673
3.67030877977761 0.504425189463043
3.74029252821037 0.563568956702987
3.81027627664313 0.619953644526123
3.88026002507588 0.673303208927992
3.95024377350864 0.723356465037984
4.0202275219414 0.769868365808998
4.09021127037416 0.812611201700743
4.16019501880691 0.851375715483339
4.23017876723967 0.885972126703468
4.30016251567243 0.916231060797541
4.37014626410519 0.942004378303225
4.44013001253794 0.963165900109718
4.5101137609707 0.979612025196152
4.58009750940346 0.991262237833841
4.65008125783621 0.998059501769267
4.72006500626897 0.999970539457967
4.79004875470173 0.996985994982286
4.86003250313449 0.989120479855391
4.93001625156724 0.9764125014873
5 0.958924274663138
}; \addplot [black, dotted] table {%
-5 1
5 1
}; \addplot [black, dotted] table {%
-5 -1
5 -1
}; \end{axis} \end{tikzpicture}
\subcaption{Plot of the generalized hyperelliptic curve $\frac{y^2}{2} = \frac{\sin^2(x)}{2}$ with dominant singularities at 
$(x,y)=\left(\pm \frac{\pi}{2}, \pm 1\right)$.}
\label{fig:curve_sine}
\end{subfigure}
\caption{Examples of curves associated to formal integrals}
\end{figure}
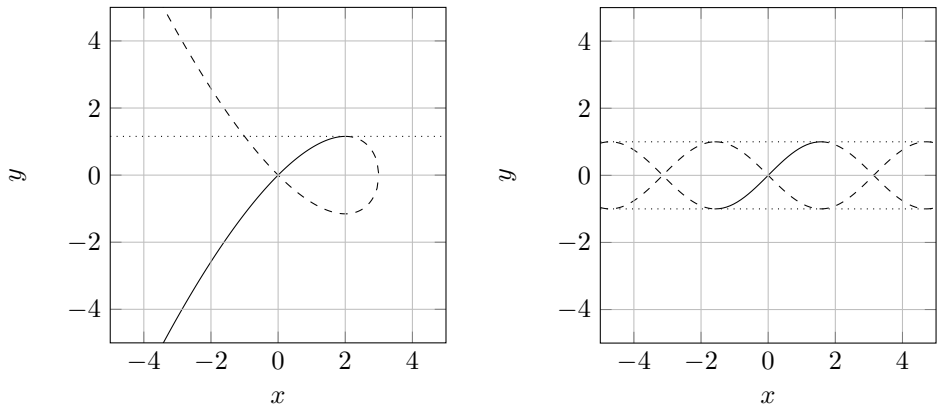
\begin{expl}[$\varphi^3$-theory as the expansion of a complex curve]
\label{expl:phi3elliptic}
For $\varphi^3$-theory the complex curve takes the form
\begin{align*} \frac{y^2}{2} = \frac{x^2}{2} - \frac{x^3}{3!}. \end{align*}
This is the elliptic curve depicted in Figure \ref{fig:curve_phi3}. 
It is clearly visible that solving for $x$ will result in a multivalued function. 
With $x(y)$, we mean the power series expansion at the origin associated to the locally increasing branch. This branch is depicted as solid line.
Moreover, we see that this expansion will have a finite radius of convergence, which is dictated by the location of the branch-cut singularity which is attained at $y=\frac{2}{\sqrt{3}}$.
\end{expl}
\begin{expl}[Sine-Gordon model as expansion of a complex curve]
\label{expl:sinegordoncurve}
Consider again the action $\Sact(x)=-\frac{\sin^2(x)}{2}$ discussed in Example \ref{expl:sinegordon}. The complex curve takes the form,
\begin{align*} \frac{y^2}{2} = \frac{\sin^2(x)}{2}. \end{align*}
This curve is depicted in Figure \ref{fig:curve_sine}.
We may solve for $x(y) = \arcsin(y)$, which is the local solution around $y=0$, which is positive for $y\rightarrow 0^+$. This local solution is drawn as black line in Figure \ref{fig:curve_sine}. Obviously, $x(y)$ has singularities for $y= \pm 1$. 
From Proposition \ref{prop:formalchangeofvar} it follows that,
\begin{align*} \mathcal{F}\left[-\frac{\sin^2(x)}{2}\right](\hbar) &= \sum_{n=0}^\infty \hbar^{n} (2n+1)!! [y^{2n+1}] \arcsin(y) \\ &= \sum_{n=0}^\infty \hbar^{n} (2n-1)!! [y^{2n}] \frac{1}{\sqrt{1-y^2}}. \end{align*}
The last equality follows because $\arcsin'(y)=\frac{1}{\sqrt{1-y^2}}$. We will use this result in Section \ref{sec:QED} to express the partition function of zero-dimensional QED using $\Fop\left[-\frac{\sin^2(x)}{2}\right]$.
\end{expl}

The representation of the coefficients of $\mathcal{F}[\Sact(x)](\hbar)$ as expansion of a generalized hyperelliptic curve can be used to calculate them efficiently. The expansion of $x(y)$ must fulfill the differential equation
\begin{align*} \frac{\partial x}{\partial y} = -\frac{y}{S'(x(y))}. \end{align*}
Using the initial condition $x(0)=0$ and $\frac{\partial x}{\partial y} > 0$, while expanding this as a power series results in the correct branch. 
\begin{expl}
For the coefficients of $Z^\text{all}(\hbar)$, where $\Sact(x) = -x^2-x-1+e^x$, we obtain the differential equation for $x(y)$:
\begin{align*} \frac{\partial x}{\partial y} = \frac{y}{1+2x-e^x} \end{align*}
The coefficients of $x(y)$ can be calculated by basic iterative methods. These coefficients can be translated into coefficients of $\Fop[-x^2-x-1+e^x](\hbar)$ using Proposition \ref{prop:formalchangeofvar}.
\end{expl}

\subsection{Asymptotics from singularity analysis}
One approach to calculate asymptotics of expressions such as the integral \eqref{eqn:formalfuncintegral1} is to perform the coefficient extraction with a Cauchy integral and to approximate the result using the method of steepest decent or saddle point method:
\begin{gather*} [\hbar^n] \int_\R \frac{dx}{\sqrt{2 \pi \hbar } } e^{\frac{1}{\hbar} \left( -\frac{x^2}{2a} + V(x) \right) } = \\ \oint_{|\hbar|=\epsilon} \frac{d\hbar}{\hbar^{n-1}} \int_\R \frac{dx}{\sqrt{2 \pi \hbar } } e^{\frac{1}{\hbar} \left( -\frac{x^2}{2a} + V(x) \right) } = \\ \oint_{|\hbar|=\epsilon} d\hbar \int_\R \frac{dx}{\sqrt{2 \pi } } e^{\frac{1}{\hbar} \left( -\frac{x^2}{2a} + V(x) \right) - (n-\frac32) \log \hbar } \end{gather*}
See for instance \cite{cvitanovic1978number}, where this technique was applied to $\varphi^3$-theory. This method was also applied to higher dimensional path integrals to obtain the asymptotics for realistic QFTs \cite{lipatov1977divergence}. The saddle points are solutions to the classical equations of motion and are called instantons in the realm of QFT.

The approach requires us to manipulate the integrand and to pick the right contour for the integration. The main disadvantage is that this procedure will result in a complicated asymptotic expansion. 

There exists a powerful method called hyperasymptotics \cite{berry1991hyperasymptotics} to obtain large order asymptotics of integrals such as \eqref{eqn:formalfuncintegral1}. This procedure is very general, as it also provides exponentially suppressed contributions as a systematic expansion. The expansion of a specific exponential order results in an expressions involving Dingle's terminants \cite{dingle1973asymptotic}. Unfortunately, these expressions can be quite complicated \cite{berry1991hyperasymptotics}.

We will take a slightly different approach, as we aim to obtain a complete asymptotic expansion in $n$: We will compute the large $n$ asymptotics of the coefficients $a_n$ of $\Fop[\Sact(x)](\hbar) = \sum_{n=0}a_n \hbar^n$ using \textit{singularity analysis} of the function $x(y)$. Singularity analysis has proven itself to be a powerful tool for asymptotics extraction even for implicitly defined power series such as $x(y)$ \cite{flajolet2009analytic}. As $x(y)$ can be interpreted as a variant of the Borel transform of $\Fop[\Sact(x)](\hbar)$, this approach is in the spirit of resurgence \cite{ecalle1981fonctions}, where singularities of the Borel transform are associated to the factorial divergence of expansions.

We will briefly repeat the necessary steps to compute the asymptotics of the implicitly defined power series $x(y)$. For a detailed account on singularity analysis, we refer to \cite[Ch. VI]{flajolet2009analytic}.

By Darboux's method, the asymptotics of the power series $x(y)$ are determined by the behavior of the function $x(y)$ near its \textit{dominant singularities}.
The dominant singularities of a function are the singularities which lie on the boundary of the disc of convergence of its expansion near the origin. 

Finding the actual location of the dominant singularity can be quite complicated. 
In our case we generally would need to calculate the \textit{monodromy} of the complex curve $\frac{y^2}{2}=-\Sact(x)$. 
However, in many examples the location of the dominant singularities is more or less obvious.

We will assume that the locations of the dominant singularities of $x(y)$ are known and that these singularities are of simple \textit{square root} type. Let $(\tau_i, \rho_i)$ be the coordinates of such a singularity. That means that $x(y)$ is non-analytic for $y \rightarrow \rho_i$ and that $\lim_{y\rightarrow \rho_i} x(y)= \tau_i$. The requirement that the singularity is of \textit{square root} type is equivalent to the requirement that the curve $\frac{y^2}{2}=-\Sact(x)$ is regular at $(x,y) = (\tau_i, \rho_i)$. 
\begin{expl}
\label{expl:phi3elliptic_singularity}
Consider the graph of the elliptic curve depicted in Figure \ref{fig:curve_phi3} from $\varphi^3$-theory. It is clear that $x(y)$ has a singularity at a fixed value of $y = \frac{2}{\sqrt{3}}$ indicated by the dotted line. This is in fact the unique dominant singularity in this example. 
The exponential growth of the coefficients of $x(y) = \sum_{n=0}^\infty c_n y^n$ is governed by the radius of convergence, $c_n \sim r^{-n} = \left(\frac{2}{\sqrt{3}} \right)^{-n}$. 
More precise asymptotics of the coefficients are determined by the singular expansion around the dominant singularity. 
In Figure \ref{fig:curve_phi3}, the point $(x_0,y_0) = (2, \frac{2}{\sqrt{3}})$ is the dominant singularity of $x(y)$ as well as a \textit{critical point} or \textit{saddle point} of the function $y(x)$ as expected by the implicit function theorem. This saddle point coincides with a saddle point of $\Sact(x)$. Although $x(y)$ has a singularity at this point, the curve stays regular.
\end{expl}
Having found the dominant singularity, it is surprisingly easy to obtain a complete asymptotic expansion for the large order behavior of the coefficients of $\Fop[\Sact(x)](\hbar)$. 
\begin{thm}
\label{thm:comb_int_asymp}
If $\Sact(x) \in x^2\R[[x]]$, such that the local solution $x(y)$ around the origin of $\frac{y^2}{2} = -\Sact(x)$ has only square-root type dominant singularities at the regular points $(\tau_i, \rho_i)$ with $i \in I$, then the Poincaré asymptotic expansion of the coefficients of $\Fop[\Sact(x)](\hbar)$ is given by
\begin{align} \label{eqn:asymptotic_expantion_combint} [\hbar^n] \Fop[\Sact(x)](\hbar) &= \sum_{k=0}^{R-1} \sum_{i\in I} w_{i,k} A_i^{-(n-k)} \Gamma(n-k) + \bigO\left(\sum_{i\in I} A_i^{-n} \Gamma(n-R)\right), \end{align}
for all $R \geq 0$, where $A_i = - S(\tau_i)$, the $\bigO$-notation refers to the $n\rightarrow \infty$ limit and
\begin{align} \label{eqn:asympgeneral} w_{i,n} &= \frac{1}{2\pi \mathrm{i}}[\hbar^n] \Fop[\Sact(x+\tau_i) - \Sact(\tau_i)](\hbar). \end{align}
\end{thm}

The fact that the asymptotics of expansions such as $\Fop[\Sact(x)](\hbar)$ are governed by the expansion around the saddle points of $\Sact(x)$ is well-known \cite{flajolet2009analytic}. The exact shape of the asymptotic expansion can be seen as slightly distorted or `resurged' version of the original expansion was also observed for instance in \cite{basar2013resurgence}, but we were not able to find a proof in the literature. Therefore a proof is provided. It is an application of the Legendre inversion formula together with a Lemma from the theory of hypergeometric functions. It is technical and postponed to \ref{sec:singanalysis}.

As was illustrated in Example \ref{expl:phi3elliptic_singularity}, a square-root type singularity of $x(y)$ coincides with a saddle point of $\Sact(x)$. This way Theorem \ref{thm:comb_int_asymp} works in a very similar way to the saddle point method.

To actually find the location of the dominant singularity in non-trivial cases, powerful techniques of singularity analysis of implicitly defined functions can be applied. For instance, if $\Sact(x)$ is a polynomial, a systematic treatment given in \cite[Chap. VII]{flajolet2009analytic} can be used. With minor modifications this can also be applied to an entire function $\Sact(x)$ for the non-degenerate case \cite{banderier2015formulae}.

Note that the quadratic coefficient of $\Sact(x+\tau_i) - \Sact(\tau_i)$ in the argument for $\Fop$ in eq. \eqref{eqn:asympgeneral} is not necessarily negative. The regularity of the complex curve only guarantees that it is non-zero. We need to generalize Definition \ref{def:formalintegral} to also allow positive quadratic coefficients. 
With this generalization the choice of the branch for the square-root in eq.\ \eqref{eqn:formalintegralpwrsrs} becomes ambiguous and we have to determine the correct branch by analytic continuation. In the scope of this article, we will only need a special case of Theorem \ref{thm:comb_int_asymp}, which remedies this ambiguity:
\begin{crll}
\label{crll:comb_int_asymp}
If $\Sact(x) = -\frac{x^2}{2} + \cdots$ is the power series expansion of an entire real function, which has simple critical points only on the real line, then there are not more than two dominant singularities associated with local minima of $\Sact(x)$ at $x= \tau_i$. The minima must have the same ordinate $\Sact(\tau_i) = -A$ to qualify as dominant singularities. 
The coefficients of the asymptotic expansion are given by
\begin{align} w_{i,n} &= \frac{1}{2\pi}[\hbar^n] \Fop[\Sact(\tau_i)- \Sact(x+\tau_i)](-\hbar), \end{align}
where the argument of $\Fop$ has a strictly negative quadratic coefficient. 
\end{crll}
\begin{proof}
If $S(x)$ is a real entire function, whose derivative vanishes on the real line, we can analytically continue $x(y)$ to a star-shaped domain excluding at most two rays on the real line. On the real line $x(y)$ can have singularities. By definition $x=0$ is a local maximum of $S(x)$. It follows from Rolle's theorem that the next critical point encountered on the real line must be a local minimum. These minima are the only candidates for dominant singularities of $x(y)$. Using Theorem \ref{thm:comb_int_asymp}, we obtain $\frac{1}{2\pi i}[\hbar^n] \Fop[\Sact(x+\tau_i) - \Sact(\tau_i)](\hbar)$ as expansions around the minima. The power series $\Sact(x+\tau_i) - \Sact(\tau_i)$ starts with a positive quadratic term, resulting in a prefactor of $\sqrt{-1}$. Taking the square root in the upper half plane results in the correct expansion. Flipping the sign in the argument and in the expansion parameter absorbs the imaginary unit in eq.\ \eqref{eqn:asympgeneral}.
\end{proof}

\begin{expl}
\label{expl:phi3theoryasymptotics}
Let $\Sact(x) = -\frac{x^2}{2} + \frac{x^3}{3!}$ as in Examples \ref{expl:phi3theoryexpansion}, \ref{expl:phi3elliptic} and \ref{expl:phi3elliptic_singularity}. The location of the dominant singularity at $x= \tau = 2$ can be obtained by solving $\Sact'(\tau)=0$ (see Figure \ref{fig:curve_phi3}). There is only one non-trivial solution. Therefore, this is the only dominant singularity of $x(y)$. We have $A= -\Sact(2) = \frac23$. It follows that,
\begin{gather*} [\hbar^n] \Fop\left[-\frac{x^2}{2} + \frac{x^3}{3!}\right](\hbar) = \\ \sum_{k=0}^{R-1} w_{k} \left(\frac{2}{3} \right)^{-(n-k)} \Gamma(n-k) + \bigO\left(\left(\frac{2}{3} \right)^{-n} \Gamma(n-R)\right) \qquad \forall R\in \N_0, \end{gather*}
where
\begin{align*} w_k &= \frac{1}{2\pi}[\hbar^k] \Fop[\Sact(2) - \Sact(x+2)](-\hbar)      = \frac{1}{2\pi}[\hbar^k] \Fop\left[-\frac{x^2}{2}-\frac{x^3}{3!}\right](-\hbar). \end{align*}
Because generally $\Fop[\Sact(x)](\hbar) = \Fop[\Sact(-x)](\hbar)$, the large $n$ asymptotics of the power series
\begin{align*} \Fop\left[ -\frac{x^2}{2} + \frac{x^3}{3!} \right](\hbar) &= \sum_{n=0}z_n \hbar^n = 1 + \frac{5}{24} \hbar + \frac{385}{1152} \hbar^2 + \frac{85085}{82944} \hbar^3 \\ &+ \frac{37182145}{7962624} \hbar^4 + \frac{5391411025}{191102976} \hbar^5 + \ldots \end{align*}
are given by the same sequence with negative expansion parameter:
\begin{align*} z_n &\underset{n\rightarrow \infty}{\sim} \frac{1}{2\pi} \left( \left(\frac{2}{3} \right)^{-n} \Gamma(n) - \frac{5}{24} \left(\frac{2}{3} \right)^{-n+1} \Gamma(n-1) + \frac{385}{1152} \left(\frac{2}{3} \right)^{-n+2} \Gamma(n-2) \right. \\ & \left. - \frac{85085}{82944} \left(\frac{2}{3} \right)^{-n+3} \Gamma(n-3) + \frac{37182145}{7962624} \left(\frac{2}{3} \right)^{-n+4} \Gamma(n-4) + \ldots \right) \end{align*}
This is an occurrence of the quite general self-replicating or resurgent phenomenon of the asymptotics of power series \cite{ecalle1981fonctions}.
\end{expl}

Restricting the dominant singularities in Corollary \ref{crll:comb_int_asymp} to be regular points of the complex curve is necessary. Otherwise, it cannot be guaranteed that a critical point actually coincides with a dominant singularity of $x(y)$. We will illustrate this in
\begin{expl}
\label{expl:counterexpl}
Let $\Sact(x) = -\frac{(1-e^{-x})^2}{2}$. This action has saddle points at $\tau_k = 2\pi i k$ for all $k\in\Z$. Because $\Sact(\tau_k)=0$ using Corollary \ref{crll:comb_int_asymp} naively would imply that, $[\hbar^n] \Fop\left[ -\frac{(1-e^{-x})^2}{2} \right] \sim \left(\frac{1}{0}\right)^n$, which is clearly nonsensical. 
On the other hand, we can solve $\frac{y^2}{2} = \frac{(1-e^{-x})^2}{2}$ for $x(y)= \log\frac{1}{1-y}$. Using Proposition \ref{prop:formalchangeofvar} immediately results in
\begin{align*} \Fop\left[ -\frac{(1-e^{-x})^2}{2} \right] &= \sum_{n=0}^\infty \hbar^n (2n+1)!! [y^{2n+1}] \log\frac{1}{1-y} \\ &= \sum_{n=0}^\infty \hbar^n (2n-1)!! [y^{2n}] \frac{1}{1-y} = \sum_{n=0}^\infty \hbar^n (2n-1)!!, \end{align*}
which naturally has a sound asymptotic description. The dominant singularity of $x(y)$ is obviously at $y=1$. An association of the asymptotics with saddle points of $\frac{(1-e^{-x})^2}{2}$ is not possible in this case, due to the irregularity of the complex curve at the saddle points.
\end{expl}
\section{The ring of factorially divergent power series}
\label{sec:ring}
For the applications to diagram counting and zero-dimensional QFT, the algebraic properties of power series with asymptotics as in eq.\ \eqref{eqn:asymptotic_expantion_combint} will be exploited. They also provide a compact notation for lengthy asymptotic expressions.
These properties can be derived from the broader theory of \textit{resurgence} established by Jean Ecalle \cite{ecalle1981fonctions}. Resurgence has been successfully applied to various physical problems \cite{garoufalidis2012asymptotics} and provides a promising approach to tackle non-perturbative phenomena in QFT \cite{alvarez2004langer,dunne2012resurgence}. A recent review is given by David Sauzin in \cite{mitschi2016divergent}. 
For the present considerations only a toy version of the complete resurgence machinery will be required. We will use the notation of \cite{borinsky2016generating}, where all the necessary tools for the present application were established. We will introduce this mathematical toolbox briefly:
\begin{defn}
Define $\fring{x}{A}{\beta}$ with $A \in \R_{>0}$ to be the subset of the ring of power series $f \in \R[[x]]$, whose coefficients have a Poincaré asymptotic expansion of the form,
\begin{align} \label{eqn:basic_asymp} f_n = \sum_{k=0}^{R-1} c_{k} A^{-n-\beta+k} \Gamma(n+\beta-k) + \bigO(A^{-n} \Gamma(n+\beta-R)), \end{align}
with coefficients $c_{k} \in \R$ and $\beta \in \R$. This subset forms a subring of $\R[[x]]$ as was shown in \cite{borinsky2016generating}. 
\end{defn}

A linear operator can be defined on $\fring{x}{A}{\beta}$ which maps a power series to its asymptotic expansion. This operator is called \textit{alien derivative} in the context of resurgence' alien calculus. 
\begin{defn}
\label{def:asymp}
Let $\asyOpV{A}{}{x}: \fring{x}{A}{\beta} \rightarrow x^{-\beta}\R[[x]]$ be the operator which maps a power series $f(x) = \sum_{n=0}^\infty f_n x^n$ to the generalized power series $\left(\asyOpV{A}{}{x} f\right)(x) = x^{-\beta}\sum_{k=0}^\infty c_k x^{k}$ such that, 
\begin{align} f_n &= \sum_{k=0}^{R-1} c_{k} A^{-n-\beta+k} \Gamma(n+\beta-k) + \bigO(A^{-n} \Gamma(n+\beta-R)) &&\forall R\geq 0. \end{align}
\end{defn}
To adapt to the context of path integrals and to the notation of resurgence two changes were made with respect to the notation of \cite{borinsky2016generating}: The exponent $A$ is given in reciprocal form as is common in resurgence. Moreover, the monomial $x^{-\beta}$ is included into the definition of the $\asyOp$-operator, which maps to power series with a fixed monomial prefactor or equivalently generalized Laurent series. This change simplifies the notation of the chain rule for compositions of power series heavily. 

\begin{expl}
Let $f(x) = \sum_{n=m+1}^\infty \Gamma(n+m) x^n$. It follows that $f \in \fring{x}{1}{m}$ and $\left(\asyOpV{1}{}{x} f\right)(x) = \frac{1}{x^m}$.
\end{expl}
\begin{expl}
For certain QED-type theories, we will need sequences which do not behave as an integer shift of the $\Gamma$-function. If for instance, 
$f(x) = \sum_{n=0}^\infty (2n-1)!! x^n = \frac{1}{\sqrt{2 \pi}} \sum_{n=0}^\infty 2^{n+\frac12} \Gamma\left(n+\frac12 \right) x^n$, then  
$\left(\asyOpV{\frac12}{}{x} f\right)(x) = \frac{1}{\sqrt{2 \pi \hbar}}$ in accord to Definition \ref{def:asymp}.
\end{expl}
For more concrete examples of power series which are elements of $\fring{x}{A}{\beta}$ that count combinatorial objects, we refer to \cite{borinsky2016generating}.

The $\asyOp$-operator is a derivative, which obeys the following identities for $f,g \in \fring{x}{A}{\beta}$. Cited from \cite{borinsky2016generating}:
\begin{align*} &\asyOpV{A}{\beta}{x} (f(x) + g(x))& &=& &\asyOpV{A}{\beta}{x} f(x) + \asyOpV{A}{\beta}{x} g(x) && \textit{Linearity} \\ &\asyOpV{A}{\beta}{x} (f(x) g(x)) & &=& &g(x) \asyOpV{A}{\beta}{x} f(x) + f(x) \asyOpV{A}{\beta}{x} g(x) && \textit{Leibniz rule} \\ &\asyOpV{A}{\beta}{x} f(g(x)) & &=& &f'(g(x)) \asyOpV{A}{\beta}{x} g(x) + e^{A \left( \frac{1}{x} - \frac{1}{\xi}\right)} (\asyOpV{A}{\beta}{\xi} f(\xi)) \big|_{\xi=g(x)} && \textit{Chain rule} \\ &\asyOpV{A}{\beta}{x} g^{-1}(x) & &=& -& {g^{-1}}'(x) e^{A \left( \frac{1}{x} - \frac{1}{\xi}\right)} (\asyOpV{A}{\beta}{\xi} g(\xi)) \big|_{\xi=g^{-1}(x)} && \textit{Inverse} \\ & & &=& -& e^{A \left( \frac{1}{x} - \frac{1}{\xi}\right)} \left. \frac{\asyOpV{A}{\beta}{\xi} g(\xi)} { \partial_\xi g(\xi) } \right|_{\xi=g^{-1}(x)} &&    \end{align*}
where $f'(x)$ denotes the usual derivative of $f(x)$. We require $g_0=0$ and $g_1=1$ for the chain rule and the inverse. 

With this notation at hand, the asymptotics of a formal integral, which fulfills the restrictions of Corollary \ref{crll:comb_int_asymp}, may be written in compact form as,
\begin{align*} \asyOpV{A}{}{\hbar} \Fop\left[\Sact(x)\right](\hbar) = \frac{1}{2\pi} \sum_{i \in I} \Fop\left[\Sact(\tau_i) - \Sact(x+\tau_i)\right](-\hbar), \end{align*}
where $\tau_i$ are the locations of the dominant saddle points, $A = -\Sact(\tau_i)$ and 
$\Fop\left[\Sact(x)\right](\hbar) \in \fring{\hbar}{A}{0}$.
The important property is that $\Fop$-expressions are stable under application of the $\asyOp$-derivative. This makes the calculation of the asymptotics as easy as calculating the expansion at low-order.
\begin{expl}
\label{expl:phi3theorycompactasymptotics}
The asymptotics deduced in Example \ref{expl:phi3theoryasymptotics} can be written in compact form as,
\begin{align*} \asyOpV{\frac23}{}{\hbar} \Fop\left[-\frac{x^2}{2} + \frac{x^3}{3!}\right](\hbar) &= \frac{1}{2\pi} \Fop\left[-\frac{x^2}{2}+\frac{x^3}{3!}\right](-\hbar), \end{align*}
where $\Fop\left[-\frac{x^2}{2}+\frac{x^3}{3!}\right](\hbar) \in \fring{\hbar}{\frac23}{0}$.
\end{expl}
\section{Overview of zero-dimensional QFT}
\label{sec:overview}
The zero-dimensional partition function of a scalar theory with interaction given by $V(x)$ is written as a formal integral,
\begin{align*} Z(\hbar, j) := \int \limits \frac{d x}{\sqrt{2 \pi \hbar}} e^{\frac{1}{\hbar} \left( - \frac{x^2}{2} + V(x) + x j \right) }. \end{align*}
This integral is to be understood as a formal expansion in $\hbar$ and $j$. The discussion from Section \ref{sec:formalint} does not immediately apply here, because of the additional $x j$ term, which was \textit{not} allowed in Definition \ref{def:formalintegral}. 
We can always transform the expression above into the canonical form as in Definition \ref{def:formalintegral} by formally shifting the integration variable,
\begin{align*} Z(\hbar, j) &= e^{\frac{-\frac{x_0^2}{2} + V(x_0) + x_0 j}{\hbar} } \int \limits_\mathbb{R} \frac{d x}{\sqrt{2 \pi \hbar}} e^{\frac{1}{\hbar} \left( -\frac{x^2}{2} + V(x+x_0) -V(x_0) -x V'(x_0) \right)} \\ &= e^{\frac{-\frac{x_0^2}{2} + V(x_0) + x_0 j}{\hbar} } \Fop\left[ -\frac{x^2}{2} + V(x+x_0) -V(x_0) -x V'(x_0)\right](\hbar) \end{align*}
where $x_0=x_0(j)$ is the unique power series solution of $x_0(j) = V'(x_0(j)) + j$.

The exponential prefactor enumerates all (possibly disconnected) \textit{tree diagrams} with the prescribed vertex structure and the $\Fop$-term enumerates all diagrams with at least one cycle in each connected component. It is useful to separate the tree-level diagrams as they contribute with negative powers in $\hbar$, which spoils the simple treatment in the formalism of power series. Trees and diagrams with at least one cycle are isolated after restricting to connected diagrams, which are generated by the \textit{free energy} of the theory:
\begin{align*} W(\hbar, j) &:= \hbar \log Z(\hbar, j) = \\ &=-\frac{x_0^2}{2} + V(x_0) + x_0 j +\hbar \log \mathcal{F} \left[ -\frac{x^2}{2} + V(x+x_0) -V(x_0) -x V'(x_0)\right] \left(\hbar \right), \end{align*}
where we conventionally multiply by $\hbar$ to go from counting by excess to counting by loop number and $x_0=x_0(j)$. The reason that taking the logarithm transforms the disconnected to the connected generating function of diagrams is that Feynman diagrams, which are weighted by their symmetry factors, are \textit{labeled} combinatorial objects \cite{flajolet2009analytic}.

The next step is to perform a Legendre transformation to get access to the effective action $G$, which is a generating function in $\hbar$ and $\varphi_c$, 
\begin{align*} G(\hbar,\varphi_c) &:= W - j \varphi_c \\ \varphi_c &:= \partial_j W. \end{align*}
The equation $\varphi_c = \partial_j W$ needs to be solved for $j$ to obtain $G$ as a generating function in $\hbar$ and $\varphi_c$. Explicitly, this is only necessary if the potential allows \textit{tadpole diagrams}. By tadpole diagrams we essentially mean diagrams which only have one external leg. The Legendre transformation can be interpreted as a combinatorial operation which maps a set of connected graphs to the respective set of \textit{2-edge-connected} graphs or \textit{1-particle irreducible} (1PI) graphs. See \cite{jackson2016robust} for a combinatorial treatment of the Legendre transformation in this context.

The coefficients of $G$ in $\varphi_c$ are called proper Green functions of the theory. More specifically, the first derivative $\partial_{\varphi_c} G |_{\varphi_c=0}$ is called the generating function of \textit{(proper) tadpoles}, the second derivative $\partial_{\varphi_c}^2 G |_{\varphi_c=0}$ is called \textit{(1PI) propagator} and higher derivatives $\partial_{\varphi_c}^k G |_{\varphi_c=0}$ are called \textit{proper $k$-point function}. 

A further step in the analysis of zero-dimensional QFT is the calculation of the \textit{renormalization constants}. The calculation is slightly artificial in zero-dimensional QFT, as there are no explicit divergences to renormalize. 
Without momentum dependence every `integral' for a graph is convergent. Thus renormalization has to be defined in analogy with higher dimensional models. 
To motivate the renormalization procedure for zero-dimensional QFT, we will use the Hopf algebra structure of Feynman diagrams.
\section{The Hopf algebra structure of renormalization}
\label{sec:hopfalgebra}

In this section, the important properties of the Hopf algebra of Feynman diagrams \cite{connes2001renormalization} will be briefly recalled. In \cite{borinsky2016lattices}, a detailed analysis of the Hopf algebra of Feynman diagrams with emphasis on the use for zero-dimensional combinatorial QFTs is given. Consult \cite{manchon2004hopf} for a general review of Hopf algebras of Feynman diagrams. 

\begin{defn}
Let ${\hopffg}$ be the $\R$-algebra generated by all mutually non-isomorphic, superficially divergent, 1PI Feynman diagrams of a certain QFT. The multiplication of generators is given by the disjoint union: $m: \hopffg \otimes \hopffg \rightarrow \hopffg, \gamma_1 \otimes \gamma_2 \mapsto \gamma_1 \cup \gamma_2$. It is extended linearly to all elements in the vector space $\hopffg$. ${\hopffg}$ has a unit $\unit:\mathbb{R} \mapsto \mathbb{I} \mathbb{R} \subset \hopffg $, where $\mathbb{I}$ is associated to the empty diagram and a counit $\counit: {\hopffg} \rightarrow \mathbb{R}$, which vanishes on every generator of ${\hopffg}$ except $\mathbb{I}$. The coproduct on the generators is defined as follows: 
    \begin{align*} &\Delta \Gamma := \sum \limits_{ \gamma \in \sdsubdiags(\Gamma) } \gamma \otimes \Gamma/\gamma& &:& &{\hopffg} \rightarrow {\hopffg} \otimes {\hopffg}, \end{align*}
    where the complete contraction $\Gamma/\Gamma$ is set to $\mathbb{I}$ and the vacuous contraction $\Gamma/\emptyset$ to $\Gamma$.

    $\sdsubdiags(\Gamma)$ denotes the set of all superficially divergent subgraphs of $\Gamma$:
\begin{align*} \sdsubdiags(\Gamma):= \left\{\gamma \subset \Gamma \text{ such that } \gamma = \prod \limits_i \gamma_i \text{ and } \gamma_i \text{ is 1PI and } \omega( \gamma_i ) \leq 0 \right\}, \end{align*}
    where $\omega(\gamma_i)$ is the superficial degree of divergence of the 1PI subgraph $\gamma_i$ obtained by power counting \cite{weinberg1960high}.
\end{defn}

As a Hopf algebra $\hopffg$ is equipped with a antipode $S: \hopffg \rightarrow \hopffg$, which is defined recursively by the identity $u \circ \epsilon = m \circ (\id \otimes S) \circ \Delta$.

We define the element $X^r$ in the completion of $\overline{\hopffg}$ of $\hopffg$ as 
\begin{align*} X^r := \pm 1 + \sum_{\substack{\text{1PI graphs }\Gamma \\\text{s.t. } \res\Gamma =r}} \frac{\Gamma}{|\Aut\Gamma|}, \end{align*}
where the sum is over all 1PI Feynman graphs with a designated \textit{residue} $r$. The residue of a graph is the external leg structure of the graph or equivalently the vertex that is created if all edges of the graph are contracted to a point. The negative sign is only assumed if $r$ is of propagator-type.

The most important property of $\hopffg$ from the perspective of renormalization is given by the identity \cite[Thm. 5]{kreimer2006anatomy},
\begin{align} \label{eqn:hopfdysonidentity} \Delta X^r = \sum_{L=0}^\infty X^r Q^{2L} \otimes X^r \big|_L, \end{align}
where $\big|_L$ is the projection to graphs of the fixed number of loops $L$ and $Q$ is the \textit{invariant charge}. 
This identity holds in \textit{renormalizable} QFTs.
For a QFT with a single vertex, $Q$ is defined as 
\begin{align} \label{eqn:invariantchargedef} Q:=\left(\frac{X^v}{\sqrt{\prod_{e\in v} (-X^e)}}\right)^{\frac{1}{|v|-2}}, \end{align}
where the product is over all edge-type residues that are incident to the vertex-type $v$ and $|v|$ denotes the number of edges incident to the vertex-type. Algebraic operations are to be understood in the combinatorial sense. For instance, arbitrary powers are translated using the generalized binomial theorem,
$(\mathbb{I}+X)^\alpha = \sum_{n=0}^\infty { \alpha \choose n} X^n$.

In theories with more than one vertex-type all possible different definitions of $Q$ must agree in a certain sense as dictated by the \textit{Slavnov-Taylor-Identities} \cite{kreimer2006anatomy,van2007renormalization}.

\textit{Feynman rules} can be defined conveniently as elements of the \textit{group of characters} of $\hopffg$. Characters are \textit{morphisms of algebras} $\phi: \hopffg \rightarrow \mathbb{A}$ from $\hopffg$ to some algebra $\mathbb{A}$, which map the unit $\mathbb{I}_{\hopffg}$ to the unit of $\mathbb{A}$, $\phi(\mathbb{I}_{\hopffg}) = \mathbb{I}_\mathbb{A}$ and respect the multiplication: $\phi(\Gamma_1 \Gamma_2) = \phi(\Gamma_1)\phi(\Gamma_2)$ for two elements $\Gamma_1, \Gamma_2 \in \hopffg$.
The zero-dimensional Feynman rules $\phi:\hopffg \rightarrow \R[[\hbar]]$ are simply given by,
\begin{align*} \phi \{\Gamma\}(\hbar) = \hbar^{|L(\Gamma)|}, \end{align*}
where $|L(\Gamma)|$ denotes the number of loops of $\Gamma$. The argument of maps from $\hopffg$ to some other space will be embraced with curly brackets to avoid confusion with the proceeding power series argument. 
By definition,
\begin{align*} g^r(\hbar):= \phi\{ X^r \}(\hbar) = \pm 1 + \sum_{\substack{\text{1PI graphs }\Gamma \\\text{s.t. } \res\Gamma =r}} \frac{\hbar^{|L(\Gamma)|}}{|\Aut\Gamma|}, \end{align*}
is the generating function of \textit{labeled} 1PI graphs with residue $r$. 
If $r_k$ is the residue with $k$ external legs in a scalar theory, $g^{r_k}(\hbar)$ is the $k$-th derivative of the 
effective action in the respective QFT, that means the proper $k$-point function: $g^{r_k}(\hbar) = \partial_{\varphi_c}^k G |_{\varphi_c=0}$.
In higher dimensions with more general Feynman rules, $\phi$ will map to a power series with coefficients in a function space. 

In a \textit{kinematic renormalization scheme} the \textit{counterterm map} for given Feynman rules $\phi: \hopffg \rightarrow \mathbb{A}$ is given by the character,
\begin{align*} S^\phi_R: \hopffg \rightarrow \mathbb{A}\\ S^\phi_R = R \circ \phi \circ S, \end{align*}
where $R: \mathbb{A} \rightarrow \mathbb{A}$ is the projection operator with $\ker R = \mathbb{A}_+$ which, loosely speaking, maps the `convergent part' of the Feynman rules in $\mathbb{A}_+$ to zero. The process of splitting $\mathbb{A}$ is a \textit{Birkhoff decomposition} \cite{connes2001renormalization}. The choice of a specific splitting is equivalent with the choice of a certain \textit{renormalization scheme}. 
The renormalized Feynman rules or \textit{Bogoliubov map} is then given by,
\begin{align*} \phi^R_\text{ren} := S^\phi_R * \phi, \end{align*}
where $*$ is the star product, $\phi_1 * \phi_2 := m \circ ( \phi_1 \otimes \phi_2) \circ \Delta$ on the group of characters of $\hopffg$ \cite{manchon2004hopf}.
The images of $X^r$ under $S^\phi_R$, $z_r := S^\phi_R\{X^r\}$ are called \textit{counterterms} or \textit{$z$-factors}.

For the renormalization of a realistic QFT, there are infinitely many options for the choice of a renormalization scheme. Typically, a one dimensional subset of renormalization schemes is chosen. This one dimensional subset can be used to setup the renormalization group, which maps different schemes onto each other. In zero-dimensional QFT on the other hand, the target algebra takes the simple form $\mathbb{A} = \R[[\hbar]]$. The only sensible renormalization scheme, which respects the grading of the Hopf algebra and reflects the usual properties of realistic renormalization, is the choice $R = \id$. 

In zero-dimensional QFT, we therefore fix our renormalization scheme to $R=\id$:
\begin{align*} S^\phi: \hopffg \rightarrow \R[[\hbar]]\\ S^\phi = \phi \circ S \end{align*}
The renormalized Feynman rules take a very simple form in this case,
\begin{align*} \phi_\text{ren} = S^\phi * \phi = (\phi \circ S) * \phi = u \circ \epsilon, \end{align*}
where $u \circ \epsilon$ maps every element of $\hopffg$ to zero except for the unit $u \circ \epsilon(\mathbb{I}) = 1$.

The renormalized Green functions in zero-dimensional QFT are therefore 
\begin{align*} g^r_\text{ren}&= \phi_\text{ren}\{ X^r \}(\hbar) = \pm 1 \end{align*}
with $+$ sign for vertex-type residue $r$ or $-$ for propagator-type residue $r$. We can take the identities above as renormalization conditions, which dictate the form of our counterterms, as has been done in \cite{cvitanovic1978number} and \cite{argyres2001zero}.

By using the identity in eq.\ \eqref{eqn:hopfdysonidentity} and the definition of the antipode $S$, renormalized Feynman rules can be written as
\begin{align*} \pm 1 &= (S^\phi * \phi)\{X^r\}(\hbar) = (m \circ (S^\phi \otimes \phi) \circ \Delta X^r )(\hbar) = (m \circ (\phi \otimes S^\phi) \circ \Delta X^r )(\hbar) \\ &=\sum_{L=0}^\infty \phi\{ X^r \}(\hbar) \left(\phi\{ Q \}(\hbar)\right)^{2L} [\hbar^{L}] S^\phi \{X^r\}(\hbar) = \phi\{ X^r \}(\hbar) S^\phi \{X^r\}(\hbar \phi\{ Q \}(\hbar)^2) \\ &= g^r(\hbar) z_r(\hbar \alpha(\hbar)) \end{align*}
where $\alpha(\hbar):= \phi\{ Q \}(\hbar)^2$ is the square of the invariant charge.

We will adapt the classic and critical insight to renormalization theory: We will interpret the expansion parameter $\hbar$ as a function of an \textit{renormalized} expansion parameter \cite{gell1954quantum}. We do this by defining $\hbar(\hbar_\text{ren})$ as the unique power series solution of $\hbar_\text{ren} = \hbar(\hbar_\text{ren}) \alpha(\hbar(\hbar_\text{ren}))$, which gives the identity,
\begin{align} z_r(\hbar_\text{ren}) = \pm \frac{1}{g^r(\hbar(\hbar_\text{ren}))} \end{align}
with $+$ sign for vertex-type residue $r$ or $-$ for propagator-type residue $r$.
Therefore, we can obtain the $z$-factors in zero-dimensional QFT from the proper Green functions and from the solution of the equation for the renormalized expansion parameter $\hbar_\text{ren}$. This computation can be performed in $\R[[\hbar]]$ and $\R[[\hbar_\text{ren}]]$. The asymptotics of these quantities can be obtained explicitly using of the $\asyOp$-derivative.

The main advantage of the Hopf algebra formulation is that the combinatorial interpretation is more accessible. In \cite{borinsky2016generating}, it was proven by the author that, in theories with only a three-valent vertex, the generating function $z_r(\hbar_\text{ren})$ counts the number of \textit{skeleton} or \textit{primitive} diagrams if $r$ is of vertex-type. In theories with four-valent vertices $z_r(\hbar_\text{ren})$ `almost' counts the number of skeleton diagrams. These results were obtained by exploiting the inherent \textit{lattice structure} of Feynman diagrams and the close relation of this structure with the Hopf algebra.
\section{Application to zero-dimensional QFT}
\label{sec:applications}
\subsection{Notation and verification}
The coefficients of asymptotic expansions in the following section are given in the notation of Section \ref{sec:ring}. That means, a row in a table such as, 

\begin{minipage}{\linewidth}
\vspace{1ex}
\centering
\def\arraystretch{1.5}
\begin{tabular}{|c||c||c|c|c|c|c|c|}
\hline
&prefactor&$\hbar^{0}$&$\hbar^{1}$&$\hbar^{2}$&$\hbar^{3}$&$\hbar^{4}$&$\hbar^{5}$\\
\hline\hline
$\asyOpV{A}{0}{\hbar} f$&$ C \hbar^{-\beta}$&$c_0$&$c_1$&$c_2$&$c_3$&$c_4$&$c_5$\\
\hline
\end{tabular}
\vspace{1ex}
\end{minipage}
corresponds to an asymptotic expansion of the coefficients of the power series $f(\hbar)$:
\begin{align*} [\hbar^n] f(\hbar) = C \sum_{k=0}^{R-1} c_k A^{-n-\beta+k} \Gamma(n+\beta-k) + \bigO(A^{-n} \Gamma(n+\beta-R)). \end{align*}

The given low-order expansions were checked by explicitly counting diagrams with the program \texttt{feyngen} \cite{borinsky2014feynman}. 
All given expansions were computed up to at least $100$ coefficients using basic computer algebra.
Although the asymptotics were completely obtained by analytic means, numerical computations were used to verify the analytic results. 
All given asymptotic expansions were checked by computing the asymptotics from the original expansions using the Richardson-extrapolation of the first $100$ coefficients.

\subsection{$\varphi^3$-theory}

\paragraph{Disconnected diagrams}
We start with an analysis of the asymptotics of zero dimensional $\varphi^3$-theory, which has been analyzed in \cite{cvitanovic1978number} using differential equations. For the sake of completeness, we will repeat the calculation with different methods and obtain all-order asymptotics in terms of $\Fop$ expressions.

The partition function with sources is given by the formal integral,
\begin{align} Z^{\varphi^3}(\hbar, j) &:= \int \frac{dx}{\sqrt{2 \pi \hbar}} e^{\frac{1}{\hbar} \left( -\frac{x^2}{2} + \frac{x^3}{3!} + x j \right) } = 1 + \frac{j^2}{2\hbar} + \frac{j^3}{3! \hbar} + \frac{1}{2} j + \frac{5}{24}\hbar + \ldots \end{align}
Pictorially, this may be depicted as,
\begin{align*} Z^{\varphi^3}(\hbar,j) &= \phi_\Sact \Big( 1 +  \frac12 {  \ifmmode \usebox{\fgsimpleprop} \else \newsavebox{\fgsimpleprop} \savebox{\fgsimpleprop}{ \begin{tikzpicture}[x=1ex,y=1ex,baseline={([yshift=-.5ex]current bounding box.center)}] \coordinate (v) ; \coordinate [right=1.2 of v] (u); \draw (v) -- (u); \filldraw (v) circle (1pt); \filldraw (u) circle (1pt); \end{tikzpicture} } \fi } + \frac16 {  \ifmmode \usebox{\fgsimplethreevtx} \else \newsavebox{\fgsimplethreevtx} \savebox{\fgsimplethreevtx}{ \begin{tikzpicture}[x=1ex,y=1ex,baseline={([yshift=-.5ex]current bounding box.center)}] \coordinate (v) ; \def \n {3}; \def \rad {1.2}; \foreach \s in {1,...,5} { \def \angle {360/\n*(\s - 1)}; \coordinate (u) at ([shift=({\angle}:\rad)]v); \draw (v) -- (u); \filldraw (u) circle (1pt); } \filldraw (v) circle (1pt); \end{tikzpicture} } \fi } + \frac18 {  \begin{tikzpicture}[x=1ex,y=1ex,baseline={([yshift=-.5ex]current bounding box.center)}] \coordinate (v0) ; \coordinate [right=1.2 of v0] (u0); \coordinate [below=1.2 of v0] (v1) ; \coordinate [right=1.2 of v1] (u1); \draw (v0) -- (u0); \filldraw (v0) circle (1pt); \filldraw (u0) circle (1pt); \draw (v1) -- (u1); \filldraw (v1) circle (1pt); \filldraw (u1) circle (1pt); \end{tikzpicture} } + \frac12 {  \ifmmode \usebox{\fgonetadpolephithree} \else \newsavebox{\fgonetadpolephithree} \savebox{\fgonetadpolephithree}{ \begin{tikzpicture}[x=2ex,y=2ex,baseline={([yshift=-.5ex]current bounding box.center)}] \coordinate (v0) ; \coordinate [right=1 of v0] (v1); \coordinate [left=.7 of v0] (i0); \coordinate [left=.5 of v1] (vm); \draw (vm) circle(.5); \draw (i0) -- (v0); \filldraw (v0) circle(1pt); \filldraw (i0) circle (1pt); \end{tikzpicture} } \fi } + \frac18 {  \begin{tikzpicture}[x=2ex,y=2ex,baseline={([yshift=-.5ex]current bounding box.center)}] \coordinate (v00); \coordinate [below=1.2 of v00] (v01); \coordinate [right=1 of v00] (v10); \coordinate [right=1 of v01] (v11); \coordinate [left=.7 of v00] (i00); \coordinate [left=.7 of v01] (i01); \coordinate [left=.5 of v10] (vm0); \coordinate [left=.5 of v11] (vm1); \draw (vm0) circle(.45); \draw (vm1) circle(.45); \draw (i00) -- (v00); \draw (i01) -- (v01); \filldraw (v00) circle(1pt); \filldraw (v01) circle(1pt); \filldraw (i00) circle (1pt); \filldraw (i01) circle(1pt); \end{tikzpicture} } + \frac14 {  \ifmmode \usebox{\fgtwojoneloopbubblephithree} \else \newsavebox{\fgtwojoneloopbubblephithree} \savebox{\fgtwojoneloopbubblephithree}{ \begin{tikzpicture}[x=2ex,y=2ex,baseline={([yshift=-.5ex]current bounding box.center)}] \coordinate (v0) ; \coordinate [right=1 of v0] (v1); \coordinate [left=.7 of v0] (i0); \coordinate [right=.7 of v1] (o0); \coordinate [left=.5 of v1] (vm); \draw (vm) circle(.5); \draw (i0) -- (v0); \draw (o0) -- (v1); \filldraw (v0) circle(1pt); \filldraw (v1) circle(1pt); \filldraw (i0) circle (1pt); \filldraw (o0) circle (1pt); \end{tikzpicture} } \fi } \\ &+ \frac16 {  \ifmmode \usebox{\fgthreejoneltrianglephithree} \else \newsavebox{\fgthreejoneltrianglephithree} \savebox{\fgthreejoneltrianglephithree}{ \begin{tikzpicture}[x=1ex,y=1ex,baseline={([yshift=-.5ex]current bounding box.center)}] \coordinate (v) ; \def \n {3}; \def \rad {1}; \def \rud {2.2}; \foreach \s in {1,...,5} { \def \angle {360/\n*(\s - 1)}; \def \ungle {360/\n*\s}; \coordinate (s) at ([shift=({\angle}:\rad)]v); \coordinate (t) at ([shift=({\ungle}:\rad)]v); \coordinate (u) at ([shift=({\angle}:\rud)]v); \draw (s) -- (u); \filldraw (u) circle (1pt); \filldraw (s) circle (1pt); } \draw (v) circle(\rad); \end{tikzpicture} } \fi } + \frac14 {  \ifmmode \usebox{\fgthreejonelpropinsphithree} \else \newsavebox{\fgthreejonelpropinsphithree} \savebox{\fgthreejonelpropinsphithree}{ \begin{tikzpicture}[x=2ex,y=2ex,baseline={([yshift=-.5ex]current bounding box.center)}] \coordinate (v0) ; \coordinate [right=1 of v0] (v1); \coordinate [right=.7 of v1] (v2); \coordinate [left=.7 of v0] (i0); \coordinate [above right=.7 of v2] (o0); \coordinate [below right=.7 of v2] (o1); \coordinate [left=.5 of v1] (vm); \draw (vm) circle(.5); \draw (i0) -- (v0); \draw (v1) -- (v2); \draw (o0) -- (v2); \draw (o1) -- (v2); \filldraw (v0) circle(1pt); \filldraw (v1) circle(1pt); \filldraw (v2) circle(1pt); \filldraw (i0) circle (1pt); \filldraw (o0) circle (1pt); \filldraw (o1) circle (1pt); \end{tikzpicture} } \fi } + \frac18 {  \ifmmode \usebox{\fghandle} \else \newsavebox{\fghandle} \savebox{\fghandle}{ \begin{tikzpicture}[x=1ex,y=1ex,baseline={([yshift=-.5ex]current bounding box.center)}] \coordinate (v0); \coordinate [right=1.5 of v0] (v1); \coordinate [left=.7 of v0] (i0); \coordinate [right=.7 of v1] (o0); \draw (v0) -- (v1); \filldraw (v0) circle (1pt); \filldraw (v1) circle (1pt); \draw (i0) circle(.7); \draw (o0) circle(.7); \end{tikzpicture} } \fi } + \frac{1}{12} {  \ifmmode \usebox{\fgbananathree} \else \newsavebox{\fgbananathree} \savebox{\fgbananathree}{ \begin{tikzpicture}[x=1ex,y=1ex,baseline={([yshift=-.5ex]current bounding box.center)}] \coordinate (vm); \coordinate [left=1 of vm] (v0); \coordinate [right=1 of vm] (v1); \draw (v0) -- (v1); \draw (vm) circle(1); \filldraw (v0) circle (1pt); \filldraw (v1) circle (1pt); \end{tikzpicture} } \fi } + \ldots \Big) \end{align*}
with the additional Feynman rule that every one-valent vertex is assigned a factor of $j$.
After a shift and rescaling of the integration variable $Z^{\varphi^3}(\hbar, j)$ takes the form,
\begin{align} \begin{split} \label{eqn:Zphi3_as_Zphi0} Z^{\varphi^3}(\hbar, j) &=  e^{\frac{-\frac{x_0^2}{2} + \frac{x_0^3}{3!}+ x_0 j}{\hbar} } \frac{1}{(1-2j)^{\frac14} } \int \frac{dx}{\sqrt{2 \pi \frac{\hbar}{(1-2j)^{\frac32}}}} e^{\frac{(1-2j)^{\frac32}}{\hbar} \left( - \frac{x^2}{2 } + \frac{x^3}{3!} \right) } \\  &= e^{\frac{(1-2 j)^{\frac32}-1+3j}{3\hbar} } \frac{1}{(1-2j)^{\frac14} } Z^{\varphi^3}_0 \left(\frac{\hbar}{(1-2j)^{\frac32}} \right) \end{split} \end{align}
where $x_0 := 1- \sqrt{1-2j}$ and $Z^{\varphi^3}_0(\hbar) := Z^{\varphi^3}(\hbar, 0)$. The last equality gives a significant simplification, because we are effectively left with a univariate generating function. The combinatorial explanation for this is that we can always `dress' a graph without external legs - a vacuum graph - by attaching an arbitrary number of rooted trees to the edges of the original graph. We also note that $-\frac{x_0^2}{2} + \frac{x_0^3}{3!}+ x_0 j = \frac13 ((1-2 j)^{\frac32}-1+3j)$, sequence \texttt{A001147} in the OEIS \cite{oeis}, is the generating function of all connected trees build out of three valent vertices. 

The generating function of graphs without external legs is given by
\begin{align*} Z^{\varphi^3}_0 \left(\hbar \right) &= \Fop\left[-\frac{x^2}{2}+\frac{x^3}{3!}\right]\left(\hbar\right), \end{align*}
which has been discussed in Examples \ref{expl:phi3theoryexpansion}, \ref{expl:phi3elliptic}, \ref{expl:phi3elliptic_singularity}, \ref{expl:phi3theoryasymptotics} and \ref{expl:phi3theorycompactasymptotics}. The first coefficients of $Z^{\varphi^3}(\hbar, j)$ are given in Table \ref{tab:Zphi3}.
\begin{table}
\begin{subtable}[c]{\textwidth}
\centering
\tiny
\def\arraystretch{1.5}
\begin{tabular}{|c||c||c|c|c|c|c|c|}
\hline
&prefactor&$\hbar^{0}$&$\hbar^{1}$&$\hbar^{2}$&$\hbar^{3}$&$\hbar^{4}$&$\hbar^{5}$\\
\hline\hline
$\partial_j^{0} Z^{\varphi^3} \big|_{j=0}$&$\hbar^{0}$&$1$&$ \frac{5}{24}$&$ \frac{385}{1152}$&$ \frac{85085}{82944}$&$ \frac{37182145}{7962624}$&$ \frac{5391411025}{191102976}$\\
\hline
$\partial_j^{1} Z^{\varphi^3} \big|_{j=0}$&$\hbar^{0}$&$ \frac{1}{2}$&$ \frac{35}{48}$&$ \frac{5005}{2304}$&$ \frac{1616615}{165888}$&$ \frac{929553625}{15925248}$&$ \frac{167133741775}{382205952}$\\
\hline
$\partial_j^{2} Z^{\varphi^3} \big|_{j=0}$&$\hbar^{-1}$&$1$&$ \frac{35}{24}$&$ \frac{5005}{1152}$&$ \frac{1616615}{82944}$&$ \frac{929553625}{7962624}$&$ \frac{167133741775}{191102976}$\\
\hline
$\partial_j^{3} Z^{\varphi^3} \big|_{j=0}$&$\hbar^{-1}$&$ \frac{5}{2}$&$ \frac{385}{48}$&$ \frac{85085}{2304}$&$ \frac{37182145}{165888}$&$ \frac{26957055125}{15925248}$&$ \frac{5849680962125}{382205952}$\\
\hline
\end{tabular}
\subcaption{The first coefficients of the bivariate generating function $Z^{\varphi^3}(\hbar, j)$.}
\label{tab:Zphi3}
\end{subtable}
\begin{subtable}[c]{\textwidth}
\centering
\tiny
\def\arraystretch{1.5}
\begin{tabular}{|c||c||c|c|c|c|c|c|}
\hline
&prefactor&$\hbar^{0}$&$\hbar^{1}$&$\hbar^{2}$&$\hbar^{3}$&$\hbar^{4}$&$\hbar^{5}$\\
\hline\hline
$\asyOpV{\frac23}{0}{\hbar} \partial_j^{0} Z^{\varphi^3} \big|_{j=0}$&$\frac{\hbar^{0}}{2\pi}$&$1$&$- \frac{5}{24}$&$ \frac{385}{1152}$&$- \frac{85085}{82944}$&$ \frac{37182145}{7962624}$&$- \frac{5391411025}{191102976}$\\
\hline
$\asyOpV{\frac23}{0}{\hbar} \partial_j^{1} Z^{\varphi^3} \big|_{j=0}$&$\frac{\hbar^{-1}}{2\pi}$&$2$&$ \frac{1}{12}$&$- \frac{35}{576}$&$ \frac{5005}{41472}$&$- \frac{1616615}{3981312}$&$ \frac{185910725}{95551488}$\\
\hline
$\asyOpV{\frac23}{0}{\hbar} \partial_j^{2} Z^{\varphi^3} \big|_{j=0}$&$\frac{\hbar^{-2}}{2\pi}$&$4$&$ \frac{1}{6}$&$- \frac{35}{288}$&$ \frac{5005}{20736}$&$- \frac{1616615}{1990656}$&$ \frac{185910725}{47775744}$\\
\hline
$\asyOpV{\frac23}{0}{\hbar} \partial_j^{3} Z^{\varphi^3} \big|_{j=0}$&$\frac{\hbar^{-3}}{2\pi}$&$8$&$- \frac{5}{3}$&$ \frac{25}{144}$&$- \frac{1925}{10368}$&$ \frac{425425}{995328}$&$- \frac{37182145}{23887872}$\\
\hline
\end{tabular}
\subcaption{The first coefficients of the bivariate generating function $\asyOpV{\frac23}{0}{\hbar} Z^{\varphi^3}(\hbar, j)$.}
\label{tab:Zphi3asymp}
\end{subtable}
\caption{Partition function in $\varphi^3$-theory.}
\end{table}
Using theorem \ref{thm:comb_int_asymp} the generating function of the asymptotics of $Z^{\varphi^3}_0 $ were calculated in Example \ref{expl:phi3theoryasymptotics}. Written in the notation of Section \ref{sec:ring}. We have $Z^{\varphi^3}_0 \in \fring{\hbar}{\frac{2}{3}}{0}$ and
\begin{align*} \asyOpV{\frac23}{0}{\hbar} Z^{\varphi^3}_0 (\hbar) = \frac{1}{2\pi} \Fop\left[-\frac{x^2}{2}+\frac{x^3}{3!}\right]\left(-\hbar\right) = \frac{1}{2\pi} Z^{\varphi^3}_0 (-\hbar). \end{align*}
This very simple form for the generating function can of course be traced back to the simple structure of $\varphi^3$, which is almost invariant under the $\asyOp$-derivative. 

The bivariate generating function of the asymptotics is obtained by using the $\asyOp$-derivative on eq.\ \eqref{eqn:Zphi3_as_Zphi0} and applying the chain rule from Section \ref{sec:ring}:
\begin{align} \begin{split} \label{eqn:Zphi3asymp} \asyOpV{\frac23}{0}{\hbar} Z^{\varphi^3}(\hbar, j) &= e^{\frac{(1-2 j)^{\frac32}-1+3j}{3\hbar} } \frac{1}{(1-2j)^{\frac14} } e^{\frac23 \frac{1 - (1-2 j)^{\frac32}}{\hbar}} \asyOpV{\frac23}{0}{\widetilde \hbar} Z^{\varphi^3}_0 \left(\widetilde \hbar \right) \big|_{\widetilde \hbar = \frac{\hbar}{(1-2j)^{\frac32}}} \\ &= \frac{1}{2\pi} e^{\frac{1-(1-2 j)^{\frac32}+3j}{3\hbar} } \frac{1}{(1-2j)^{\frac14} } Z^{\varphi^3}_0 \left(-\frac{\hbar}{(1-2j)^{\frac32}} \right). \end{split} \end{align}
Note that the $\asyOp$-derivative commutes with expansions in $j$, as we leave the number of external legs fixed while taking the limit to large loop order.
The first coefficients of the asymptotics of $Z^{\varphi^3}(\hbar, j)$ are listed in Table \ref{tab:Zphi3asymp}.

We may also expand the expression for the asymptotics in eq.\ \eqref{eqn:Zphi3asymp} in $\hbar$ to obtain a generating function for the first coefficient of the asymptotic expansions of the derivatives by $j$:
\begin{align*} \asyOpV{\frac23}{0}{\hbar} Z^{\varphi^3}(\hbar, j) &=\frac{1}{2\pi} e^{\frac{2j}{\hbar}} \left( 1 + \left(-\frac{5}{24} + \frac14 \frac{2j}{\hbar} - \frac18 \frac{(2j)^2}{\hbar^2} \right) \hbar +\ldots\right) \\ \asyOpV{\frac23}{0}{\hbar} \partial_j^m Z^{\varphi^3}(\hbar, j) \big|_{j=0} &= \frac{1}{2\pi} \left(\frac{2}{\hbar}\right)^{m}\left(1 + \left( -\frac{5}{24} + \frac{3m}{8} - \frac{m^2}{8} \right) \hbar + \ldots \right) \end{align*}
By Definition \ref{def:asymp} this can be translated into an asymptotic expression for large order coefficients. With $\partial_j^m Z^{\varphi^3}(\hbar, j)\big|_{j=0} = \sum_{n=0}^\infty z_{m,n} \hbar^n$:
\begin{align*} z_{m,n} &= \sum_{k=0}^{R-1} c_{m,k} \left(\frac{2}{3} \right)^{-m-n+k} \Gamma(n+m-k) \\ &+\bigO\left( \left(\frac{2}{3} \right)^{-m-n+R} \Gamma(n+m-R)\right) , \intertext{for all $R \geq 0$, where $c_{m,k}= [\hbar^k] \hbar^m \asyOpV{\frac23}{0}{\hbar}\partial_j^m Z^{\varphi^3}(\hbar, j)\big|_{j=0}$ or more explicitly, } z_{m,n} & \underset{n\rightarrow \infty}{\sim} \frac{2^m}{2\pi} \left(\frac{2}{3} \right)^{-m-n} \Gamma(n+m) \\ & \times \left( 1 + \frac23 \left( -\frac{5}{24} + \frac{3m}{8} - \frac{m^2}{8} \right) \frac{1}{n+m-1} +\ldots \right), \end{align*}
which agrees with the coefficients, which were given in \cite{cvitanovic1978number} in a different notation.
\paragraph{Connected diagrams}
\begin{table}
\begin{subtable}[c]{\textwidth}
\centering
\tiny
\def\arraystretch{1.5}
\begin{tabular}{|c||c|c|c|c|c|c|}
\hline
&$\hbar^{0}$&$\hbar^{1}$&$\hbar^{2}$&$\hbar^{3}$&$\hbar^{4}$&$\hbar^{5}$\\
\hline\hline
$\partial_j^{0} W^{\varphi^3} \big|_{j=0}$&$0$&$0$&$ \frac{5}{24}$&$ \frac{5}{16}$&$ \frac{1105}{1152}$&$ \frac{565}{128}$\\
\hline
$\partial_j^{1} W^{\varphi^3} \big|_{j=0}$&$0$&$ \frac{1}{2}$&$ \frac{5}{8}$&$ \frac{15}{8}$&$ \frac{1105}{128}$&$ \frac{1695}{32}$\\
\hline
$\partial_j^{2} W^{\varphi^3} \big|_{j=0}$&$1$&$1$&$ \frac{25}{8}$&$15$&$ \frac{12155}{128}$&$ \frac{11865}{16}$\\
\hline
$\partial_j^{3} W^{\varphi^3} \big|_{j=0}$&$1$&$4$&$ \frac{175}{8}$&$150$&$ \frac{158015}{128}$&$11865$\\
\hline
\end{tabular}
\subcaption{Table of the first coefficients of the bivariate generating function $W^{\varphi^3}(\hbar, j)$.}
\label{tab:Wphi3}
\end{subtable}
\begin{subtable}[c]{\textwidth}
\centering
\tiny
\def\arraystretch{1.5}
\begin{tabular}{|c||c||c|c|c|c|c|c|}
\hline
&prefactor&$\hbar^{0}$&$\hbar^{1}$&$\hbar^{2}$&$\hbar^{3}$&$\hbar^{4}$&$\hbar^{5}$\\
\hline\hline
$\asyOpV{\frac23}{0}{\hbar} \partial_j^{0} W^{\varphi^3}\big|_{j=0}$&$\frac{\hbar^{1}}{2 \pi}$&$1$&$- \frac{5}{12}$&$ \frac{25}{288}$&$- \frac{20015}{10368}$&$ \frac{398425}{497664}$&$- \frac{323018725}{5971968}$\\
\hline
$\asyOpV{\frac23}{0}{\hbar} \partial_j^{1} W^{\varphi^3}\big|_{j=0}$&$\frac{\hbar^{0}}{2 \pi}$&$2$&$- \frac{5}{6}$&$- \frac{155}{144}$&$- \frac{17315}{5184}$&$- \frac{3924815}{248832}$&$- \frac{294332125}{2985984}$\\
\hline
$\asyOpV{\frac23}{0}{\hbar} \partial_j^{2} W^{\varphi^3}\big|_{j=0}$&$\frac{\hbar^{-1}}{2 \pi}$&$4$&$- \frac{11}{3}$&$- \frac{275}{72}$&$- \frac{31265}{2592}$&$- \frac{7249295}{124416}$&$- \frac{553369915}{1492992}$\\
\hline
$\asyOpV{\frac23}{0}{\hbar} \partial_j^{3} W^{\varphi^3}\big|_{j=0}$&$\frac{\hbar^{-2}}{2 \pi}$&$8$&$- \frac{46}{3}$&$- \frac{407}{36}$&$- \frac{51065}{1296}$&$- \frac{12501815}{62208}$&$- \frac{988327615}{746496}$\\
\hline
\end{tabular}
\subcaption{Table of the first coefficients of the bivariate generating function $\asyOpV{\frac23}{0}{\hbar} W^{\varphi^3}(\hbar, j)$.}
\label{tab:Wphi3asymp}
\end{subtable}
\caption{Free energy in $\varphi^3$-theory.}
\end{table}

The generating function of the connected graphs can be obtained by taking the logarithm:
\begin{align} \begin{split} \label{eqn:Wphi3_explicit} W^{\varphi^3} ( \hbar, j ) &:= \hbar \log Z^{\varphi^3}(\hbar, j) \\ &= \frac13 ((1-2 j)^{\frac32}-1+3j) + \frac14 \hbar \log \frac{1}{1-2j } + \hbar \log Z^{\varphi^3}_0 \left(\frac{\hbar}{(1-2j)^{\frac32}} \right) \end{split} \\ \notag &= \frac{5}{24}\hbar^{2} + \frac{1}{2} j \hbar + \frac{5}{8} j \hbar^{2} + \frac{1}{2} j^{2} + \frac{1}{2} j^{2} \hbar + \frac{25}{16} j^{2} \hbar^{2} + \ldots \end{align}
This can be written as the diagrammatic expansion,
\begin{align*} W^{\varphi^3}(\hbar,j) &=  \phi_\Sact \Big( \frac12 {  \ifmmode \usebox{\fgsimpleprop} \else \newsavebox{\fgsimpleprop} \savebox{\fgsimpleprop}{ \begin{tikzpicture}[x=1ex,y=1ex,baseline={([yshift=-.5ex]current bounding box.center)}] \coordinate (v) ; \coordinate [right=1.2 of v] (u); \draw (v) -- (u); \filldraw (v) circle (1pt); \filldraw (u) circle (1pt); \end{tikzpicture} } \fi } + \frac16 {  \ifmmode \usebox{\fgsimplethreevtx} \else \newsavebox{\fgsimplethreevtx} \savebox{\fgsimplethreevtx}{ \begin{tikzpicture}[x=1ex,y=1ex,baseline={([yshift=-.5ex]current bounding box.center)}] \coordinate (v) ; \def \n {3}; \def \rad {1.2}; \foreach \s in {1,...,5} { \def \angle {360/\n*(\s - 1)}; \coordinate (u) at ([shift=({\angle}:\rad)]v); \draw (v) -- (u); \filldraw (u) circle (1pt); } \filldraw (v) circle (1pt); \end{tikzpicture} } \fi } + \frac12 {  \ifmmode \usebox{\fgonetadpolephithree} \else \newsavebox{\fgonetadpolephithree} \savebox{\fgonetadpolephithree}{ \begin{tikzpicture}[x=2ex,y=2ex,baseline={([yshift=-.5ex]current bounding box.center)}] \coordinate (v0) ; \coordinate [right=1 of v0] (v1); \coordinate [left=.7 of v0] (i0); \coordinate [left=.5 of v1] (vm); \draw (vm) circle(.5); \draw (i0) -- (v0); \filldraw (v0) circle(1pt); \filldraw (i0) circle (1pt); \end{tikzpicture} } \fi } + \frac14 {  \ifmmode \usebox{\fgtwojoneloopbubblephithree} \else \newsavebox{\fgtwojoneloopbubblephithree} \savebox{\fgtwojoneloopbubblephithree}{ \begin{tikzpicture}[x=2ex,y=2ex,baseline={([yshift=-.5ex]current bounding box.center)}] \coordinate (v0) ; \coordinate [right=1 of v0] (v1); \coordinate [left=.7 of v0] (i0); \coordinate [right=.7 of v1] (o0); \coordinate [left=.5 of v1] (vm); \draw (vm) circle(.5); \draw (i0) -- (v0); \draw (o0) -- (v1); \filldraw (v0) circle(1pt); \filldraw (v1) circle(1pt); \filldraw (i0) circle (1pt); \filldraw (o0) circle (1pt); \end{tikzpicture} } \fi } \\ &+ \frac16 {  \ifmmode \usebox{\fgthreejoneltrianglephithree} \else \newsavebox{\fgthreejoneltrianglephithree} \savebox{\fgthreejoneltrianglephithree}{ \begin{tikzpicture}[x=1ex,y=1ex,baseline={([yshift=-.5ex]current bounding box.center)}] \coordinate (v) ; \def \n {3}; \def \rad {1}; \def \rud {2.2}; \foreach \s in {1,...,5} { \def \angle {360/\n*(\s - 1)}; \def \ungle {360/\n*\s}; \coordinate (s) at ([shift=({\angle}:\rad)]v); \coordinate (t) at ([shift=({\ungle}:\rad)]v); \coordinate (u) at ([shift=({\angle}:\rud)]v); \draw (s) -- (u); \filldraw (u) circle (1pt); \filldraw (s) circle (1pt); } \draw (v) circle(\rad); \end{tikzpicture} } \fi } + \frac14 {  \ifmmode \usebox{\fgthreejonelpropinsphithree} \else \newsavebox{\fgthreejonelpropinsphithree} \savebox{\fgthreejonelpropinsphithree}{ \begin{tikzpicture}[x=2ex,y=2ex,baseline={([yshift=-.5ex]current bounding box.center)}] \coordinate (v0) ; \coordinate [right=1 of v0] (v1); \coordinate [right=.7 of v1] (v2); \coordinate [left=.7 of v0] (i0); \coordinate [above right=.7 of v2] (o0); \coordinate [below right=.7 of v2] (o1); \coordinate [left=.5 of v1] (vm); \draw (vm) circle(.5); \draw (i0) -- (v0); \draw (v1) -- (v2); \draw (o0) -- (v2); \draw (o1) -- (v2); \filldraw (v0) circle(1pt); \filldraw (v1) circle(1pt); \filldraw (v2) circle(1pt); \filldraw (i0) circle (1pt); \filldraw (o0) circle (1pt); \filldraw (o1) circle (1pt); \end{tikzpicture} } \fi } + \frac18 {  \ifmmode \usebox{\fghandle} \else \newsavebox{\fghandle} \savebox{\fghandle}{ \begin{tikzpicture}[x=1ex,y=1ex,baseline={([yshift=-.5ex]current bounding box.center)}] \coordinate (v0); \coordinate [right=1.5 of v0] (v1); \coordinate [left=.7 of v0] (i0); \coordinate [right=.7 of v1] (o0); \draw (v0) -- (v1); \filldraw (v0) circle (1pt); \filldraw (v1) circle (1pt); \draw (i0) circle(.7); \draw (o0) circle(.7); \end{tikzpicture} } \fi } + \frac{1}{12} {  \ifmmode \usebox{\fgbananathree} \else \newsavebox{\fgbananathree} \savebox{\fgbananathree}{ \begin{tikzpicture}[x=1ex,y=1ex,baseline={([yshift=-.5ex]current bounding box.center)}] \coordinate (vm); \coordinate [left=1 of vm] (v0); \coordinate [right=1 of vm] (v1); \draw (v0) -- (v1); \draw (vm) circle(1); \filldraw (v0) circle (1pt); \filldraw (v1) circle (1pt); \end{tikzpicture} } \fi } + \ldots \Big) \end{align*}
where we now assign a $\hbar^{|L(\Gamma)|}$ to every graph, where $|L(\Gamma)|$ is the number of loops of $\Gamma$.
The large-$n$ asymptotics of the coefficients $w_n(j) = [\hbar^n] W^{\varphi^3} ( \hbar, j )$ can be obtained by using the chain rule for $\asyOp$:
\begin{align} \asyOpV{\frac23}{0}{\hbar} W^{\varphi^3} (\hbar, j) &= \hbar \left[e^{\frac{2}{3} \left(\frac{1}{\hbar} - \frac{1}{\widetilde \hbar} \right) } \asyOpV{\frac23}{0}{\widetilde \hbar} \log Z^{\varphi^3}_0 \left( \widetilde \hbar \right)\right]_{\widetilde \hbar =\frac{\hbar}{(1-2j)^{\frac32}}}. \end{align}
Some coefficients of the bivariate generating functions $W^{\varphi^3} (\hbar, j)$ and $\asyOpV{\frac23}{0}{\hbar} W^{\varphi^3} (\hbar, j)$ are given in Tables \ref{tab:Wphi3} and \ref{tab:Wphi3asymp}. Comparing Tables \ref{tab:Zphi3asymp} and \ref{tab:Wphi3asymp}, we can observe the classic result, proven by Wright \cite{wright1970asymptotic}, that the asymptotics of connected and disconnected graphs differ only by a subdominant contribution. 

With the expressions above, we have explicit generating functions for the connected \textit{$n$-point functions} and their all-order asymptotics. 
For instance,
\begin{align*} \left. W^{\varphi^3} \right|_{j=0} &= \hbar \log Z^{\varphi^3}_0(\hbar) & \asyOpV{\frac23}{0}{\hbar} \left. W^{\varphi^3} \right|_{j=0} &= \hbar \asyOpV{\frac23}{0}{\hbar} \log Z^{\varphi^3}_0(\hbar) \\ \left. \frac{\partial W^{\varphi^3} }{\partial j} \right|_{j=0} &= \frac12 \hbar + 3 \hbar^2 \partial_\hbar \log Z^{\varphi^3}_0(\hbar) & \asyOpV{\frac23}{0}{\hbar} \left. \frac{\partial W^{\varphi^3} }{\partial j} \right|_{j=0} &= \left(2 + 3 \hbar^2 \partial_\hbar \right) \asyOpV{\frac23}{0}{\hbar} \log Z^{\varphi^3}_0(\hbar) . \end{align*}
Every $n$-point function is a linear combination of $\log Z^{\varphi^3}(\hbar)$ and its derivatives and the asymptotics are linear combinations of $\asyOpV{\frac23}{0}{\hbar} \log Z^{\varphi^3}(\hbar) = \frac{1}{2\pi}\frac{Z^{\varphi^3}_0(-\hbar)}{Z^{\varphi^3}_0(\hbar)}$ and its derivatives. 

Of course, we could derive differential equations, which are fulfilled by $Z^{\varphi^3}_0(\hbar)$, $\log Z^{\varphi^3}(\hbar)$ and $\asyOpV{\frac23}{0}{\hbar} \log Z^{\varphi^3}(\hbar)$ to simplify the expressions above. This would have to be done in a very model specific manner. We will not pursue this path in the scope of this article, as we aim for providing machinery which can be used for general models. 

\paragraph{1PI diagrams}

The next object of interest is the effective action,
\begin{align} \G^{\varphi^3}(\hbar, \varphi_c) &= W^{\varphi^3}(\hbar, j(\hbar, \varphi_c)) - j(\hbar, \varphi_c) \varphi_c, \end{align}
which is the Legendre transform of $W$, where $j(\hbar, \varphi_c)$ is the solution of  
$\varphi_c = \partial_j W^{\varphi^3} \left( \hbar, j \right)$.
A small calculation reveals what for the special case of $\varphi^3$-theory this can be written explicitly in terms of 
$\varphi_c$. It is convenient to define $\gamma^{\varphi^3}_0(\hbar) := \frac{\G^{\varphi^3}(\hbar,0)}{\hbar} = \frac{W^{\varphi^3}(\hbar, j(\hbar, 0))}{\hbar}$. Eq.\ \eqref{eqn:Wphi3_explicit} gives us the more explicit form,
\begin{align*} \gamma^{\varphi^3}_0(\hbar)&= \frac{(1-2 j_0(\hbar))^{\frac32}-1+3j_0(\hbar)}{3\hbar} \\ &+ \frac14 \log \frac{1}{1-2j_0(\hbar) } + \log Z^{\varphi^3}_0 \left(\frac{\hbar}{(1-2j_0(\hbar))^{\frac32}}\right). \end{align*}
where $j_0(\hbar)=j(\hbar,0)$ is the unique power series solution of the equation $0= \frac{\partial W^{\varphi^3}}{\partial j}\left(\hbar, j_0(\hbar)\right)$, 
The bivariate generating function $\G^{\varphi^3}(\hbar, \varphi_c)$ is then,
\begin{align} \G^{\varphi^3}(\hbar, \varphi_c) &= -\frac{\varphi_c^2}{2} + \frac{\varphi_c^3}{3!} + \frac12 \hbar \log \frac{1}{1-\varphi_c} + \hbar \gamma^{\varphi^3}_0\left(\frac{\hbar}{(1-\varphi_c)^3}\right). \end{align}
The combinatorial interpretation of the identity is the following: A 1PI diagram either has no or only one loop, or it can be reduced to a vacuum diagram by removing all external legs and the attached vertices. 
This bivariate generating function can be depicted diagrammatically as,
\begin{align*} G^{\varphi^3}(\hbar,\varphi_c) &= \phi_\Sact \Big(  -\frac12 {  \ifmmode \usebox{\fgsimpleprop} \else \newsavebox{\fgsimpleprop} \savebox{\fgsimpleprop}{ \begin{tikzpicture}[x=1ex,y=1ex,baseline={([yshift=-.5ex]current bounding box.center)}] \coordinate (v) ; \coordinate [right=1.2 of v] (u); \draw (v) -- (u); \filldraw (v) circle (1pt); \filldraw (u) circle (1pt); \end{tikzpicture} } \fi } + \frac16 {  \ifmmode \usebox{\fgsimplethreevtx} \else \newsavebox{\fgsimplethreevtx} \savebox{\fgsimplethreevtx}{ \begin{tikzpicture}[x=1ex,y=1ex,baseline={([yshift=-.5ex]current bounding box.center)}] \coordinate (v) ; \def \n {3}; \def \rad {1.2}; \foreach \s in {1,...,5} { \def \angle {360/\n*(\s - 1)}; \coordinate (u) at ([shift=({\angle}:\rad)]v); \draw (v) -- (u); \filldraw (u) circle (1pt); } \filldraw (v) circle (1pt); \end{tikzpicture} } \fi } + \frac12 {  \ifmmode \usebox{\fgonetadpolephithree} \else \newsavebox{\fgonetadpolephithree} \savebox{\fgonetadpolephithree}{ \begin{tikzpicture}[x=2ex,y=2ex,baseline={([yshift=-.5ex]current bounding box.center)}] \coordinate (v0) ; \coordinate [right=1 of v0] (v1); \coordinate [left=.7 of v0] (i0); \coordinate [left=.5 of v1] (vm); \draw (vm) circle(.5); \draw (i0) -- (v0); \filldraw (v0) circle(1pt); \filldraw (i0) circle (1pt); \end{tikzpicture} } \fi } + \frac14 {  \ifmmode \usebox{\fgtwojoneloopbubblephithree} \else \newsavebox{\fgtwojoneloopbubblephithree} \savebox{\fgtwojoneloopbubblephithree}{ \begin{tikzpicture}[x=2ex,y=2ex,baseline={([yshift=-.5ex]current bounding box.center)}] \coordinate (v0) ; \coordinate [right=1 of v0] (v1); \coordinate [left=.7 of v0] (i0); \coordinate [right=.7 of v1] (o0); \coordinate [left=.5 of v1] (vm); \draw (vm) circle(.5); \draw (i0) -- (v0); \draw (o0) -- (v1); \filldraw (v0) circle(1pt); \filldraw (v1) circle(1pt); \filldraw (i0) circle (1pt); \filldraw (o0) circle (1pt); \end{tikzpicture} } \fi } + \frac16 {  \ifmmode \usebox{\fgthreejoneltrianglephithree} \else \newsavebox{\fgthreejoneltrianglephithree} \savebox{\fgthreejoneltrianglephithree}{ \begin{tikzpicture}[x=1ex,y=1ex,baseline={([yshift=-.5ex]current bounding box.center)}] \coordinate (v) ; \def \n {3}; \def \rad {1}; \def \rud {2.2}; \foreach \s in {1,...,5} { \def \angle {360/\n*(\s - 1)}; \def \ungle {360/\n*\s}; \coordinate (s) at ([shift=({\angle}:\rad)]v); \coordinate (t) at ([shift=({\ungle}:\rad)]v); \coordinate (u) at ([shift=({\angle}:\rud)]v); \draw (s) -- (u); \filldraw (u) circle (1pt); \filldraw (s) circle (1pt); } \draw (v) circle(\rad); \end{tikzpicture} } \fi } + \frac{1}{12} {  \ifmmode \usebox{\fgbananathree} \else \newsavebox{\fgbananathree} \savebox{\fgbananathree}{ \begin{tikzpicture}[x=1ex,y=1ex,baseline={([yshift=-.5ex]current bounding box.center)}] \coordinate (vm); \coordinate [left=1 of vm] (v0); \coordinate [right=1 of vm] (v1); \draw (v0) -- (v1); \draw (vm) circle(1); \filldraw (v0) circle (1pt); \filldraw (v1) circle (1pt); \end{tikzpicture} } \fi } + \ldots \Big), \end{align*}
where we assign a $\varphi_c$ instead of $j$ to every external leg.

Acting with the $\asyOp$-derivative on $\gamma^{\varphi^3}_0$ gives,
\begin{align*} \asyOpV{\frac23}{0}{\hbar} \gamma^{\varphi^3}_0(\hbar) &= \asyOpV{\frac23}{0}{\hbar} \frac{W^{\varphi^3}(\hbar, j_0(\hbar))}{\hbar} \\ &= \left(\asyOpV{\frac23}{0}{\hbar} \frac{W^{\varphi^3}(\hbar, j)}{\hbar} \right)_{j=j_0(\hbar)} + \frac{\frac{\partial W^{\varphi^3}}{\partial j}(\hbar, j)\big|_{j=j_0(\hbar)}}{\hbar} \asyOpV{\frac23}{0}{\hbar} j_0(\hbar) \intertext{where the second term vanishes by the definition of $j_0$. Therefore,} \asyOpV{\frac23}{0}{\hbar} \gamma^{\varphi^3}_0(\hbar) &= \left[e^{\frac{2}{3} \left(\frac{1}{\hbar} - \frac{1}{\widetilde \hbar} \right) } \asyOpV{\frac23}{0}{\widetilde \hbar} \log Z^{\varphi^3}_0 \left( \widetilde \hbar \right)\right]_{\widetilde \hbar =\frac{\hbar}{(1-2j_0(\hbar))^{\frac32}}}, \end{align*}
and
\begin{align} \begin{split} \asyOpV{\frac23}{0}{\hbar} \G^{\varphi^3}(\hbar, \varphi_c) &= \hbar \left[e^{\frac{2}{3} \left(\frac{1}{\hbar} - \frac{1}{\widetilde \hbar} \right) } \asyOpV{\frac23}{0}{\widetilde \hbar} \gamma^{\varphi^3}_0(\widetilde \hbar) \right]_{\widetilde \hbar = \frac{\hbar}{(1-\varphi_c)^3}} \\ &= \hbar \left[e^{\frac{2}{3} \left(\frac{1}{\hbar} - \frac{1}{\widetilde \hbar} \right) } \asyOpV{\frac23}{0}{\widetilde \hbar} \log Z^{\varphi^3}_0 \left( \widetilde \hbar \right)\right]_{\widetilde \hbar =\frac{\hbar}{(1-\varphi_c)^3\left(1-2j_0\left(\frac{\hbar}{(1-\varphi_c)^3}\right)\right)^{\frac32}}}. \end{split} \end{align}

\begin{table}
\begin{subtable}[c]{\textwidth}
\centering
\tiny
\def\arraystretch{1.5}
\begin{tabular}{|c||c|c|c|c|c|c|}
\hline
&$\hbar^{0}$&$\hbar^{1}$&$\hbar^{2}$&$\hbar^{3}$&$\hbar^{4}$&$\hbar^{5}$\\
\hline\hline
$\partial_{\varphi_c}^{0} G^{\varphi^3} \big|_{\varphi_c=0}$&$0$&$0$&$ \frac{1}{12}$&$ \frac{5}{48}$&$ \frac{11}{36}$&$ \frac{539}{384}$\\
\hline
$\partial_{\varphi_c}^{1} G^{\varphi^3} \big|_{\varphi_c=0}$&$0$&$ \frac{1}{2}$&$ \frac{1}{4}$&$ \frac{5}{8}$&$ \frac{11}{4}$&$ \frac{539}{32}$\\
\hline
$\partial_{\varphi_c}^{2} G^{\varphi^3} \big|_{\varphi_c=0}$&$-1$&$ \frac{1}{2}$&$1$&$ \frac{35}{8}$&$ \frac{55}{2}$&$ \frac{7007}{32}$\\
\hline
$\partial_{\varphi_c}^{3} G^{\varphi^3} \big|_{\varphi_c=0}$&$1$&$1$&$5$&$35$&$ \frac{605}{2}$&$ \frac{49049}{16}$\\
\hline
\end{tabular}
\subcaption{Table of the first coefficients of the bivariate generating function $\G^{\varphi^3}(\hbar, \varphi_c)$.}
\label{tab:Gphi3}
\end{subtable}
\begin{subtable}[c]{\textwidth}
\centering
\tiny
\def\arraystretch{1.5}
\begin{tabular}{|c||c||c|c|c|c|c|c|}
\hline
&prefactor&$\hbar^{0}$&$\hbar^{1}$&$\hbar^{2}$&$\hbar^{3}$&$\hbar^{4}$&$\hbar^{5}$\\
\hline\hline
$\asyOpV{\frac23}{0}{\hbar} \partial_{\varphi_c}^{0} G^{\varphi^3} \big|_{\varphi_c=0}$&$e^{-1} \frac{\hbar^{1}}{2 \pi}$&$1$&$- \frac{7}{6}$&$- \frac{11}{72}$&$- \frac{10135}{1296}$&$- \frac{536087}{31104}$&$- \frac{296214127}{933120}$\\
\hline
$\asyOpV{\frac23}{0}{\hbar} \partial_{\varphi_c}^{1} G^{\varphi^3} \big|_{\varphi_c=0}$&$e^{-1} \frac{\hbar^{0}}{2 \pi}$&$2$&$- \frac{7}{3}$&$- \frac{137}{36}$&$- \frac{10729}{648}$&$- \frac{1630667}{15552}$&$- \frac{392709787}{466560}$\\
\hline
$\asyOpV{\frac23}{0}{\hbar} \partial_{\varphi_c}^{2} G^{\varphi^3} \big|_{\varphi_c=0}$&$e^{-1} \frac{\hbar^{-1}}{2 \pi}$&$4$&$- \frac{26}{3}$&$- \frac{179}{18}$&$- \frac{15661}{324}$&$- \frac{2531903}{7776}$&$- \frac{637309837}{233280}$\\
\hline
$\asyOpV{\frac23}{0}{\hbar} \partial_{\varphi_c}^{3} G^{\varphi^3} \big|_{\varphi_c=0}$&$e^{-1} \frac{\hbar^{-2}}{2 \pi}$&$8$&$- \frac{100}{3}$&$- \frac{101}{9}$&$- \frac{18883}{162}$&$- \frac{3471563}{3888}$&$- \frac{940175917}{116640}$\\
\hline
\end{tabular}
\subcaption{Table of the first coefficients of the bivariate generating function $\asyOpV{\frac23}{0}{\hbar} G^{\varphi^3}(\hbar, \varphi_c)$.}
\label{tab:Gphi3asymp}
\end{subtable}
\caption{Effective action in $\varphi^3$-theory.}
\end{table}

This can be expanded in $\varphi_c$ to obtain the asymptotics of the 1PI or `proper' $n$-point functions. Some coefficients of the bivariate generating function $G^{\varphi^3}$ and its asymptotics are listed in Tables \ref{tab:Gphi3} and \ref{tab:Gphi3asymp}.

As for the disconnected diagrams, we can also expand $\asyOpV{\frac23}{0}{\hbar} G^{\varphi^3}(\hbar, \varphi_c)$ in $\hbar$ to obtain an asymptotic expansion for general $m$ with $\partial_{\varphi_c}^m G^{\varphi^3}(\hbar, \varphi_c)\big|_{\varphi_c=0} = \sum_{n=0}^\infty g_{m,n} \hbar^n$.
Expanding gives,
\begin{gather*} \asyOpV{\frac23}{0}{\hbar} \G^{\varphi^3}(\hbar, \varphi_c) = \hbar \frac{e^{-1+\frac{2 \varphi_c}{\hbar}}}{2\pi} \left( 1 - \frac16 \left( 7+ 3 (2\varphi_c)^2 \right) \hbar \right. \\ \left. -\frac{1}{72} \left( 11 + 126 (2\varphi_c) - 42 (2\varphi_c)^2 - 8 (2\varphi_c)^3 - 9 (2\varphi_c)^4 \right) \hbar^2 + \ldots \right). \intertext{Translated to an asymptotic expansion this becomes, } g_{m,n} \underset{n\rightarrow \infty}{\sim} \frac{e^{-1}}{2\pi} 2^m \left(\frac{2}{3} \right)^{-m-n} \Gamma(n+m-1) \\ \times\left( 1 - \frac19 \frac{ 7 - 3 m + 3 m^2}{n+m-2} -\frac{1}{162} \frac{ 11 + 210 m - 123 m^2 + 48 m^3 - 9 m^4 }{(n+m-3)(n+m-2)} + \ldots \right). \end{gather*}

\paragraph{Renormalization constants and skeleton diagrams}
To perform the \textit{charge renormalization} as explained in detail in Section \ref{sec:hopfalgebra}, the invariant charge in $\varphi^3$-theory needs to be defined in accordance to eq.\ \eqref{eqn:invariantchargedef},
\begin{align} \alpha(\hbar) := \left( \frac{ \partial_{\varphi_c}^3 G^{\varphi^3} |_{\varphi_c = 0} (\hbar) }{ \left(- \partial_{\varphi_c}^2 G^{\varphi^3} |_{\varphi_c = 0}(\hbar) \right)^{\frac32}} \right)^2 \end{align}
The exponents in the expression above are a consequence of the combinatorial fact, that a 1PI $\varphi^3$-graph has two additional vertices and three additional propagators for each additional loop.
We need to solve 
\begin{align*} \hbar_{\text{ren}} = \hbar(\hbar_\text{ren}) \alpha(\hbar(\hbar_\text{ren})) \end{align*}
for $\hbar(\hbar_{\text{ren}})$. The asymptotics in $\hbar_{\text{ren}}$ can be obtained by using the formula for the compositional inverse of $\asyOp$-derivative given in Section \ref{sec:ring} on this expression:
\begin{align*} \begin{split} (\asyOpV{\frac23}{0}{\hbar_{\text{ren}}} \hbar(\hbar_{\text{ren}}))&=       - \left[ e^{\frac23 \left(\frac{1}{\hbar_{\text{ren}}} - \frac{1}{\hbar}\right)} \frac{ \asyOpV{\frac23}{0}{\hbar} \left(\hbar \alpha(\hbar) \right) } { \partial_\hbar \left( \hbar \alpha(\hbar)\right) }\right]_{\hbar = \hbar(\hbar_\text{ren})}.           \end{split} \end{align*}
The $z$-factors are then obtained as explained in Section \ref{sec:hopfalgebra}. They fulfill the identities,
\begin{align*} -1 &= z_{\varphi_c^2}(\hbar_\text{ren}) \partial_{\varphi_c}^2 G^{\varphi^3} \big|_{\varphi_c = 0} \left(\hbar(\hbar_\text{ren})\right) \\ 1 &= z_{\varphi_c^3}(\hbar_\text{ren}) \partial_{\varphi_c}^3 G^{\varphi^3} \big|_{\varphi_c = 0} \left(\hbar(\hbar_\text{ren})\right) \end{align*}
\begin{table}
\begin{subtable}[c]{\textwidth}
\centering
\tiny
\def\arraystretch{1.5}
\begin{tabular}{|c||c|c|c|c|c|c|}
\hline
&$\hbar_{\text{ren}}^{0}$&$\hbar_{\text{ren}}^{1}$&$\hbar_{\text{ren}}^{2}$&$\hbar_{\text{ren}}^{3}$&$\hbar_{\text{ren}}^{4}$&$\hbar_{\text{ren}}^{5}$\\
\hline\hline
$\hbar(\hbar_{\text{ren}})$&$0$&$1$&$- \frac{7}{2}$&$6$&$- \frac{33}{2}$&$- \frac{345}{16}$\\
\hline
$z_{\varphi_c^2}(\hbar_{\text{ren}})$&$1$&$ \frac{1}{2}$&$- \frac{1}{2}$&$- \frac{1}{4}$&$-2$&$- \frac{29}{2}$\\
\hline
$z_{\varphi_c^3}(\hbar_{\text{ren}})$&$1$&$-1$&$- \frac{1}{2}$&$-4$&$-29$&$- \frac{545}{2}$\\
\hline
\end{tabular}
\subcaption{Table of the first coefficients of the renormalization quantities in $\varphi^3$-theory.}
\label{tab:phi3ren}
\end{subtable}
\begin{subtable}[c]{\textwidth}
\centering
\tiny
\def\arraystretch{1.5}
\begin{tabular}{|c||c||c|c|c|c|c|c|}
\hline
&prefactor&$\hbar_{\text{ren}}^{0}$&$\hbar_{\text{ren}}^{1}$&$\hbar_{\text{ren}}^{2}$&$\hbar_{\text{ren}}^{3}$&$\hbar_{\text{ren}}^{4}$&$\hbar_{\text{ren}}^{5}$\\
\hline\hline
$\left(\asyOpV{\frac23}{0}{\hbar_{\text{ren}}} \hbar \right)(\hbar_{\text{ren}})$&$e^{- \frac{10}{3}} \frac{\hbar^{-1}}{2 \pi}$&$-16$&$ \frac{412}{3}$&$- \frac{3200}{9}$&$ \frac{113894}{81}$&$ \frac{765853}{243}$&$ \frac{948622613}{14580}$\\
\hline
$\left(\asyOpV{\frac23}{0}{\hbar_{\text{ren}}} z_{\varphi_c^2} \right)(\hbar_{\text{ren}})$&$e^{- \frac{10}{3}} \frac{\hbar^{-1}}{2 \pi}$&$-4$&$ \frac{64}{3}$&$ \frac{76}{9}$&$ \frac{13376}{81}$&$ \frac{397486}{243}$&$ \frac{284898947}{14580}$\\
\hline
$\left(\asyOpV{\frac23}{0}{\hbar_{\text{ren}}} z_{\varphi_c^3} \right)(\hbar_{\text{ren}})$&$e^{- \frac{10}{3}} \frac{\hbar^{-2}}{2 \pi}$&$-8$&$ \frac{128}{3}$&$ \frac{152}{9}$&$ \frac{26752}{81}$&$ \frac{794972}{243}$&$ \frac{569918179}{14580}$\\
\hline
\end{tabular}
\subcaption{Table of the first coefficients of the asymptotics of the renormalization quantities in $\varphi^3$-theory.}
\label{tab:phi3renasymp}
\end{subtable}
\caption{Renormalization constants in $\varphi^3$-theory.}
\end{table}%
By application of the $\asyOp$-derivative, the asymptotics of $z_{\varphi_c^3}$ are:
\begin{align} \begin{split} \label{eqn:asymptoticsZ} \asyOpV{\frac23}{0}{\hbar_\text{ren}} z_{\varphi_c^3}(\hbar_\text{ren}) &= -\left[\left( \partial_{\varphi_c}^3 G^{\varphi^3} \big|_{\varphi_c = 0}(\hbar)\right)^{-2} e^{\frac23 \left(\frac{1}{\hbar_{\text{ren}}} - \frac{1}{\hbar}\right)} \left( \asyOpV{\frac23}{0}{\hbar} \partial_{\varphi_c}^3 G^{\varphi^3} \big|_{\varphi_c = 0}(\hbar)   \vphantom{ \left(\partial_\hbar \partial_{\varphi_c}^3 G^{\varphi^3} \big|_{\varphi_c = 0}(\hbar)\right) \frac{ \asyOpV{\frac23}{0}{\hbar} \left(\hbar \alpha(\hbar) \right) } { \partial_\hbar \left( \hbar \alpha(\hbar)\right) } }   \right. \right. \\ &- \left. \left.   \vphantom{ \left( \partial_{\varphi_c}^3 G^{\varphi^3} \big|_{\varphi_c = 0}(\hbar)\right)^{-2} e^{\frac23 \left(\frac{1}{\hbar_{\text{ren}}} - \frac{1}{\hbar}\right)} \left( \asyOpV{\frac23}{0}{\hbar} \partial_{\varphi_c}^3 G^{\varphi^3} \big|_{\varphi_c = 0}(\hbar) \right. }   \left(\partial_\hbar \partial_{\varphi_c}^3 G^{\varphi^3} \big|_{\varphi_c = 0}(\hbar)\right) \frac{ \asyOpV{\frac23}{0}{\hbar} \left(\hbar \alpha(\hbar) \right) } { \partial_\hbar \left( \hbar \alpha(\hbar)\right) } \right) \right]_{\hbar = \hbar(\hbar_\text{ren})}. \end{split} \end{align}
and for $z_{\varphi_c^2}$ analogously. Some coefficients of the renormalization constants and their asymptotics are given in Tables \ref{tab:phi3ren} and \ref{tab:phi3renasymp}. 

It was observed by Cvitanović et al.\ \cite{cvitanovic1978number} that $1-z_{\varphi_c^3}(\hbar_\text{ren})$ is the generating function of \textit{skeleton} diagrams. Skeleton diagrams are 1PI diagrams without any superficially divergent subgraphs. 
Utilizing the Hopf algebra of Feynman graphs and the interpretation of subgraph structures as algebraic lattices, this was proven by the author in \cite{borinsky2016generating}. 
Applying Definition \ref{def:asymp} directly gives a complete asymptotic expansion of the coefficients of $1-z_{\varphi_c^3}(\hbar_\text{ren})$,
\begin{gather*} [\hbar_\text{ren}^n]( 1-z_{\varphi_c^3}(\hbar_\text{ren})) = \sum_{k=0}^{R-1} c_{k} \left(\frac{2}{3}\right)^{-n-2+k} \Gamma(n+2-k) \\ + \bigO \left( \left(\frac{2}{3}\right)^{-n} \Gamma(n+2-R)\right) \intertext{for $R \geq 0$, where $c_k = [\hbar_\text{ren}^k] \left(-\hbar_\text{ren}^2 \asyOpV{\frac23}{0}{\hbar_\text{ren}} z_{\varphi_c^3}(\hbar_\text{ren}) \right)$. Or more explicit for large $n$,} [\hbar_\text{ren}^n]( 1-z_{\varphi_c^3}(\hbar_\text{ren})) \underset{n\rightarrow \infty}{\sim} \frac{e^{-\frac{10}{3}}}{2\pi} \left(\frac{2}{3}\right)^{-n-2} \Gamma(n+2) \left( 8 - \frac{2}{3}\frac{128}{3} \frac{1}{n+1} \right. \\ \left. - \left(\frac{2}{3}\right)^2 \frac{152}{9}\frac{1}{n(n+1)} - \left(\frac{2}{3}\right)^3 \frac{26752}{81}\frac{1}{(n-1)n(n+1)} +\ldots \right). \end{gather*}
The constant coefficient of $\asyOpV{\frac23}{0}{\hbar_\text{ren}} z_{\varphi_c^3}(\hbar_\text{ren})$ was also given in \cite{cvitanovic1978number}.

Using the first coefficients of $\asyOpV{\frac23}{0}{\hbar} \partial_{\varphi_c}^3 \G^{\varphi^3}\big|_{\varphi_c=0}$ and $-\asyOpV{\frac23}{0}{\hbar_\text{ren}} z_{\varphi_c^3}(\hbar_\text{ren})$, we may deduce that the proportion of skeleton diagrams in the set of all proper vertex diagrams is,
\begin{align*} \frac{ \frac{e^{-\frac{10}{3}}}{2\pi} \left(\frac{2}{3}\right)^{-n-2} \Gamma(n+2) \left( 8 - \frac23 \frac{128}{3} \frac{1}{n+1} + \ldots \right)}{ \frac{e^{-1}}{2\pi} \left(\frac{2}{3}\right)^{-n-2} \Gamma(n+2) \left( 8 - \frac23 \frac{100}{3} \frac{1}{n+1} + \ldots \right) }= e^{-\frac{7}{3}}\left(1-\frac{56}{9}\frac{1}{n} + \bigO\left(\frac{1}{n^2}\right)\right). \end{align*}
A random 1PI diagram in $\varphi^3$-theory is therefore a skeleton diagram with probability 
\begin{align*} e^{-\frac{7}{3}}\left(1-\frac{56}{9}\frac{1}{n} + \bigO\left(\frac{1}{n^2}\right)\right), \end{align*}
where $n$ is the loop number.

All results obtained in this section can be translated to the respective asymptotic results on cubic graphs. For instance, $\frac{1}{3!}(1-z_{\varphi_c^3}(\hbar_\text{ren}))$ is the generating function of cyclically four-connected graphs with one distinguished vertex. In \cite{wormald1985enumeration}, the first coefficient of the asymptotic expansion of those graphs is given, which agrees with our expansion. More connections to graph counting problems of this kind will be discussed in a future publication \cite{borinsky2017graph}.
\subsection{$\varphi^4$-theory}

In $\varphi^4$-theory the partition function is given by the formal integral,
\begin{align*} Z^{\varphi^4}(\hbar, j) &:= \int \frac{dx}{\sqrt{2 \pi \hbar}} e^{\frac{1}{\hbar} \left( -\frac{x^2}{2} + \frac{x^4}{4!} + x j \right) } \\ &= 1 + \frac{j^2}{2\hbar} + \frac{j^4}{4! \hbar} + \frac58 j^2 + \frac{1155}{128} j^4 + \frac{1}{8}\hbar + \ldots \end{align*}
In this case, it is not possible to completely absorb the $j$ dependents into the argument of $Z^{\varphi^4}_0$. We only can do so up to fourth order in $j$, which is still sufficient to obtain the generating functions which are necessary to calculate the renormalization constants:
\begin{align*} Z^{\varphi^4}(\hbar, j) &= \int \frac{dx}{\sqrt{2 \pi \hbar}} e^{\frac{1}{\hbar} \left( -\frac{x^2}{2} + \frac{x^4}{4!} + \hbar \log \cosh \frac{x j}{\hbar} \right) } \\ &= \int \frac{dx}{\sqrt{2 \pi \hbar}} e^{\frac{1}{\hbar} \left( -\frac{x^2}{2} + \frac{x^4}{4!} + \hbar \left( \frac{1}{2} \left( \frac{x j}{\hbar} \right)^2 - \frac{1}{12} \left(\frac{x j}{\hbar} \right)^4 + \bigO(j^6) \right) \right) } \\ &= \int \frac{dx}{\sqrt{2 \pi \hbar}} e^{\frac{1}{\hbar} \left( -\left(1 - \frac{j^2}{\hbar}\right) \frac{x^2}{2} + \left(1 - 2 \frac{j^4}{\hbar^3} \right)\frac{x^4}{4!} \right) } + \bigO(j^6) \\ &= \frac{1}{\sqrt{1 - \frac{j^2}{\hbar}}} Z^{\varphi^4}_0 \left(\hbar \frac{1 - 2 \frac{j^4}{\hbar^3}}{\left(1 - \frac{j^2}{\hbar}\right)^2}\right) + \bigO(j^6) \end{align*}
where $Z^{\varphi^4}_0(\hbar) := Z^{\varphi^4}(\hbar, 0) = \Fop\left[-\frac{x^2}{2} + \frac{x^4}{4!} \right](\hbar)$.

\begin{table}
\begin{subtable}[c]{\textwidth}
\centering
\tiny
\def\arraystretch{1.5}
\begin{tabular}{|c||c||c|c|c|c|c|c|}
\hline
&prefactor&$\hbar^{0}$&$\hbar^{1}$&$\hbar^{2}$&$\hbar^{3}$&$\hbar^{4}$&$\hbar^{5}$\\
\hline\hline
$\partial_j^{0} Z^{\varphi^4} \big|_{j=0}$&$\hbar^{0}$&$1$&$ \frac{1}{8}$&$ \frac{35}{384}$&$ \frac{385}{3072}$&$ \frac{25025}{98304}$&$ \frac{1616615}{2359296}$\\
\hline
$\partial_j^{2} Z^{\varphi^4} \big|_{j=0}$&$\hbar^{-1}$&$1$&$ \frac{5}{8}$&$ \frac{105}{128}$&$ \frac{5005}{3072}$&$ \frac{425425}{98304}$&$ \frac{11316305}{786432}$\\
\hline
$\partial_j^{4} Z^{\varphi^4} \big|_{j=0}$&$\hbar^{-2}$&$3$&$ \frac{35}{8}$&$ \frac{1155}{128}$&$ \frac{25025}{1024}$&$ \frac{8083075}{98304}$&$ \frac{260275015}{786432}$\\
\hline
\end{tabular}
\subcaption{The first coefficients of the bivariate generating function $Z^{\varphi^4}(\hbar, j)$.}
\label{tab:Zphi4}
\end{subtable}
\begin{subtable}[c]{\textwidth}
\centering
\tiny
\def\arraystretch{1.5}
\begin{tabular}{|c||c||c|c|c|c|c|c|}
\hline
&prefactor&$\hbar^{0}$&$\hbar^{1}$&$\hbar^{2}$&$\hbar^{3}$&$\hbar^{4}$&$\hbar^{5}$\\
\hline\hline
$\asyOpV{\frac32}{0}{\hbar} \partial_j^{0} Z^{\varphi^4} \big|_{j=0}$&$\frac{\hbar^{0}}{\sqrt{2}\pi}$&$1$&$- \frac{1}{8}$&$ \frac{35}{384}$&$- \frac{385}{3072}$&$ \frac{25025}{98304}$&$- \frac{1616615}{2359296}$\\
\hline
$\asyOpV{\frac32}{0}{\hbar} \partial_j^{2} Z^{\varphi^4} \big|_{j=0}$&$\frac{\hbar^{-2}}{\sqrt{2}\pi}$&$6$&$ \frac{1}{4}$&$- \frac{5}{64}$&$ \frac{35}{512}$&$- \frac{5005}{49152}$&$ \frac{85085}{393216}$\\
\hline
$\asyOpV{\frac32}{0}{\hbar} \partial_j^{4} Z^{\varphi^4} \big|_{j=0}$&$\frac{\hbar^{-4}}{\sqrt{2}\pi}$&$36$&$- \frac{9}{2}$&$ \frac{9}{32}$&$- \frac{35}{256}$&$ \frac{1155}{8192}$&$- \frac{15015}{65536}$\\
\hline
\end{tabular}
\subcaption{The first coefficients of the bivariate generating function $\asyOpV{\frac32}{0}{\hbar}Z^{\varphi^4}(\hbar, j)$.}
\label{tab:Zphi4asymp}
\end{subtable}
\caption{Partition function in $\varphi^4$-theory.}
\end{table}

The asymptotics of $Z^{\varphi^4}_0$ can be calculated directly by using Corollary \ref{crll:comb_int_asymp}: The action $\Sact(x) = -\frac{x^2}{2} + \frac{x^4}{4!}$ is real analytic and all critical points lie on the real axis. The non-trivial critical points of $\Sact(x) =-\frac{x^2}{2} + \frac{x^4}{4!}$ are $\tau_{\pm} = \pm \sqrt{3!}$. The value at the critical points is $\Sact(\tau_{\pm})= - \frac32$. These are the dominant singularities which both contribute. Therefore, $A= \frac{3}{2}$ and $\Sact(\tau_{\pm})-\Sact(x+\tau_{\pm})= -x^2 \pm \frac{x^3}{\sqrt{3!}} + \frac{x^4}{4!}$. 
\begin{align*} \asyOpV{\frac32}{0}{\hbar} Z^{\varphi^4}_0 (\hbar) &= \frac{1}{2 \pi} \left(\Fop\left[-x^2 + \frac{x^3}{\sqrt{3!}} + \frac{x^4}{4!}\right](-\hbar) + \Fop\left[-x^2 - \frac{x^3}{\sqrt{3!}} + \frac{x^4}{4!}\right](-\hbar) \right) \\ &= \frac{1}{\pi} \Fop\left[-x^2 + \frac{x^3}{\sqrt{3!}} + \frac{x^4}{4!}\right](-\hbar) \\ &= \frac{1}{\sqrt{2}\pi} \left( 1 - \frac{1}{8} \hbar + \frac{35}{384} \hbar^{2} - \frac{385}{3072} \hbar^{3} + \frac{25025}{98304} \hbar^{4}  + \ldots\right) \end{align*}
The combinatorial interpretation of this sequence is the following: Diagrams with three or four-valent vertices are weighted with a $\lambda_3=\sqrt{3!}$ for each three-valent vertex, $\lambda_4=1$ for each four-valent vertex, a factor $a=\frac12$ for each edge and a $(-1)$ for every loop in accordance to Proposition \ref{prop:diagraminterpretation}. The whole sequence is preceded by a factor of $\sqrt{a}=\frac{1}{\sqrt{2}}$ as required by the definition of $\Fop$.

The asymptotics for $Z^{\varphi^4}(\hbar, j)$ can again be obtained by utilizing the chain rule for $\asyOp$:
\begin{align*} \asyOpV{\frac32}{0}{\hbar} Z^{\varphi^4}(\hbar, j) &= \frac{1}{\sqrt{1 - \frac{j^2}{\hbar}}} \left[ e^{ \frac32 \left( \frac{1}{\hbar} - \frac{1}{\widetilde \hbar} \right) } (\asyOpV{\frac32}{0}{\widetilde \hbar} Z^{\varphi^4}_0 )\left(\widetilde \hbar \right) \right]_{\widetilde \hbar = \hbar \frac{1 - 2 \frac{j^4}{\hbar^3}}{(1 - \frac{j^2}{\hbar})^2}} + \bigO(j^6) \end{align*}
The first coefficients of $Z^{\varphi^4}(\hbar, j)$ are given in Table \ref{tab:Zphi4} and the respective asymptotic expansions in Table \ref{tab:Zphi4asymp}.

The generating function of the connected graphs is given by,
\begin{align*} W^{\varphi^4}(\hbar, j) &:= \hbar \log Z^{\varphi^4}(\hbar, j) \\ &= \frac12 \hbar \log \frac{1}{1-\frac{j^2}{\hbar}} + \hbar \log Z^{\varphi^4}_0 \left(\hbar \frac{1 - 2 \frac{j^4}{\hbar^3}}{(1 - \frac{j^2}{\hbar})^2}\right) + \bigO(j^6) \end{align*}
\begin{table}
\begin{subtable}[c]{\textwidth}
\centering
\tiny
\def\arraystretch{1.5}
\begin{tabular}{|c||c|c|c|c|c|c|}
\hline
&$\hbar^{0}$&$\hbar^{1}$&$\hbar^{2}$&$\hbar^{3}$&$\hbar^{4}$&$\hbar^{5}$\\
\hline\hline
$\partial_j^{0} W^{\varphi^4} \big|_{j=0}$&$0$&$0$&$ \frac{1}{8}$&$ \frac{1}{12}$&$ \frac{11}{96}$&$ \frac{17}{72}$\\
\hline
$\partial_j^{2} W^{\varphi^4} \big|_{j=0}$&$1$&$ \frac{1}{2}$&$ \frac{2}{3}$&$ \frac{11}{8}$&$ \frac{34}{9}$&$ \frac{619}{48}$\\
\hline
$\partial_j^{4} W^{\varphi^4} \big|_{j=0}$&$1$&$ \frac{7}{2}$&$ \frac{149}{12}$&$ \frac{197}{4}$&$ \frac{15905}{72}$&$ \frac{107113}{96}$\\
\hline
\end{tabular}
\subcaption{The first coefficients of the bivariate generating function $W^{\varphi^4}(\hbar, j)$.}
\label{tab:Wphi4}
\end{subtable}
\begin{subtable}[c]{\textwidth}
\centering
\tiny
\def\arraystretch{1.5}
\begin{tabular}{|c||c||c|c|c|c|c|c|}
\hline
&prefactor&$\hbar^{0}$&$\hbar^{1}$&$\hbar^{2}$&$\hbar^{3}$&$\hbar^{4}$&$\hbar^{5}$\\
\hline\hline
$\asyOpV{\frac32}{0}{\hbar} \partial_j^{0} W^{\varphi^4} \big|_{j=0}$&$\frac{\hbar^{1}}{\sqrt{2} \pi}$&$1$&$- \frac{1}{4}$&$ \frac{1}{32}$&$- \frac{89}{384}$&$ \frac{353}{6144}$&$- \frac{10623}{8192}$\\
\hline
$\asyOpV{\frac32}{0}{\hbar} \partial_j^{2} W^{\varphi^4} \big|_{j=0}$&$\frac{\hbar^{-1}}{\sqrt{2} \pi}$&$6$&$- \frac{3}{2}$&$- \frac{13}{16}$&$- \frac{73}{64}$&$- \frac{2495}{1024}$&$- \frac{84311}{12288}$\\
\hline
$\asyOpV{\frac32}{0}{\hbar} \partial_j^{4} W^{\varphi^4} \big|_{j=0}$&$\frac{\hbar^{-3}}{\sqrt{2} \pi}$&$36$&$-45$&$- \frac{111}{8}$&$- \frac{687}{32}$&$- \frac{25005}{512}$&$- \frac{293891}{2048}$\\
\hline
\end{tabular}
\subcaption{The first coefficients of the bivariate generating function $\asyOpV{\frac32}{0}{\hbar}W^{\varphi^4}(\hbar, j)$.}
\label{tab:Wphi4asymp}
\end{subtable}
\caption{Free energy in $\varphi^4$-theory.}
\end{table}
and the asymptotics are,
\begin{align*} \asyOpV{\frac32}{0}{\hbar} W^{\varphi^4}(\hbar, j) &= \hbar \left[ e^{ \frac32 \left( \frac{1}{\hbar} - \frac{1}{\widetilde \hbar} \right) } \asyOpV{\frac32}{0}{\hbar} \log Z^{\varphi^4}_0\left(\widetilde \hbar \right) \right]_{\widetilde \hbar = \hbar \frac{1 - 2 \frac{j^4}{\hbar^3}}{(1 - \frac{j^2}{\hbar})^2}} + \bigO(j^6). \end{align*}
The first coefficients of the original generating function and the generating function for the asymptotics are given in Tables \ref{tab:Wphi4} and \ref{tab:Wphi4asymp}.

The effective action, which is the Legendre transform of $W^{\varphi^4}$,
\begin{align*} G^{\varphi^4}(\hbar, \varphi_c) &= W^{\varphi^4}(\hbar, j) - j \varphi_c \end{align*}
where $\varphi_c := \partial_j W^{\varphi^4}$,
is easy to handle in this case, as there are no tadpole diagrams. Derivatives of $G^{\varphi^4}(\hbar, \varphi_c) $ with respect to $\varphi_c$ can be calculated by exploiting that $\varphi_c=0$ implies $j=0$.  
\begin{table}
\begin{subtable}[c]{\textwidth}
\centering
\tiny
\def\arraystretch{1.5}
\begin{tabular}{|c||c|c|c|c|c|c|}
\hline
&$\hbar^{0}$&$\hbar^{1}$&$\hbar^{2}$&$\hbar^{3}$&$\hbar^{4}$&$\hbar^{5}$\\
\hline\hline
$\partial_{\varphi_c}^{0} G^{\varphi^4} \big|_{\varphi_c=0}$&$0$&$0$&$ \frac{1}{8}$&$ \frac{1}{12}$&$ \frac{11}{96}$&$ \frac{17}{72}$\\
\hline
$\partial_{\varphi_c}^{2} G^{\varphi^4} \big|_{\varphi_c=0}$&$-1$&$ \frac{1}{2}$&$ \frac{5}{12}$&$ \frac{5}{6}$&$ \frac{115}{48}$&$ \frac{625}{72}$\\
\hline
$\partial_{\varphi_c}^{4} G^{\varphi^4} \big|_{\varphi_c=0}$&$1$&$ \frac{3}{2}$&$ \frac{21}{4}$&$ \frac{45}{2}$&$ \frac{1775}{16}$&$ \frac{4905}{8}$\\
\hline
\end{tabular}
\subcaption{The first coefficients of the bivariate generating function $G^{\varphi^4}(\hbar, j)$.}
\label{tab:Gphi4}
\end{subtable}
\begin{subtable}[c]{\textwidth}
\centering
\tiny
\def\arraystretch{1.5}
\begin{tabular}{|c||c||c|c|c|c|c|c|}
\hline
&prefactor&$\hbar^{0}$&$\hbar^{1}$&$\hbar^{2}$&$\hbar^{3}$&$\hbar^{4}$&$\hbar^{5}$\\
\hline\hline
$\asyOpV{\frac32}{0}{\hbar} \partial_{\varphi_c}^{0} G^{\varphi^4} \big|_{\varphi_c=0}$&$ \frac{\hbar^{1}}{\sqrt{2} \pi}$&$1$&$- \frac{1}{4}$&$ \frac{1}{32}$&$- \frac{89}{384}$&$ \frac{353}{6144}$&$- \frac{10623}{8192}$\\
\hline
$\asyOpV{\frac32}{0}{\hbar} \partial_{\varphi_c}^{2} G^{\varphi^4} \big|_{\varphi_c=0}$&$ \frac{\hbar^{-1}}{\sqrt{2} \pi}$&$6$&$- \frac{15}{2}$&$- \frac{45}{16}$&$- \frac{445}{64}$&$- \frac{22175}{1024}$&$- \frac{338705}{4096}$\\
\hline
$\asyOpV{\frac32}{0}{\hbar} \partial_{\varphi_c}^{4} G^{\varphi^4} \big|_{\varphi_c=0}$&$ \frac{\hbar^{-3}}{\sqrt{2} \pi}$&$36$&$-117$&$ \frac{369}{8}$&$- \frac{1671}{32}$&$- \frac{103725}{512}$&$- \frac{1890555}{2048}$\\
\hline
\end{tabular}
\subcaption{The first coefficients of the bivariate generating function $\asyOpV{\frac32}{0}{\hbar}G^{\varphi^4}(\hbar, j)$.}
\label{tab:Gphi4asymp}
\end{subtable}
\caption{Effective action in $\varphi^4$-theory.}
\end{table}
For instance,
\begin{align*} \left. G^{\varphi^4} \right|_{\varphi_c=0} &= W^{\varphi^4}(\hbar, 0) \\  \left. \partial_{\varphi_c}^2 G^{\varphi^4} \right|_{\varphi_c=0} &= \left. - \frac{\partial j}{\partial \varphi_c} \right|_{\varphi_c=0} = - \frac{1}{\left. \partial_j^2 W^{\varphi^4} \right|_{j=0}} \\          \left. \partial_{\varphi_c}^4 G^{\varphi^4} \right|_{\varphi_c=0} &= \frac{\left. \partial_j^4 W^{\varphi^4} \right|_{j=0}}{ \left( \left. \partial_j^2 W^{\varphi^4} \right|_{j=0} \right)^4}                     \end{align*}
The calculation of the asymptotic expansions can be performed by applying the $\asyOp$-derivative on these expressions and using the Leibniz and chain rules to write them in terms of the asymptotics of $W^{\varphi^4}$.
Some coefficients of $G^{\varphi^4}(\hbar, j)$ are listed in Table \ref{tab:Gphi4} with the respective asymptotics in Table \ref{tab:Gphi4asymp}.

\begin{table}
\begin{subtable}[c]{\textwidth}
\centering
\tiny
\def\arraystretch{1.5}
\begin{tabular}{|c||c|c|c|c|c|c|}
\hline
&$\hbar_{\text{ren}}^{0}$&$\hbar_{\text{ren}}^{1}$&$\hbar_{\text{ren}}^{2}$&$\hbar_{\text{ren}}^{3}$&$\hbar_{\text{ren}}^{4}$&$\hbar_{\text{ren}}^{5}$\\
\hline\hline
$\hbar(\hbar_{\text{ren}})$&$0$&$1$&$- \frac{5}{2}$&$ \frac{25}{6}$&$- \frac{15}{2}$&$ \frac{25}{3}$\\
\hline
$z_{\varphi^2}(\hbar_{\text{ren}})$&$1$&$ \frac{1}{2}$&$- \frac{7}{12}$&$ \frac{1}{8}$&$- \frac{9}{16}$&$- \frac{157}{96}$\\
\hline
$z_{\varphi^4}(\hbar_{\text{ren}})$&$1$&$- \frac{3}{2}$&$ \frac{3}{4}$&$- \frac{11}{8}$&$- \frac{45}{16}$&$- \frac{499}{32}$\\
\hline
\end{tabular}
\subcaption{Table of the first coefficients of the renormalization quantities in $\varphi^4$-theory.}
\label{tab:phi4ren}
\end{subtable}
\begin{subtable}[c]{\textwidth}
\centering
\tiny
\def\arraystretch{1.5}
\begin{tabular}{|c||c||c|c|c|c|c|c|}
\hline
&prefactor&$\hbar_{\text{ren}}^{0}$&$\hbar_{\text{ren}}^{1}$&$\hbar_{\text{ren}}^{2}$&$\hbar_{\text{ren}}^{3}$&$\hbar_{\text{ren}}^{4}$&$\hbar_{\text{ren}}^{5}$\\
\hline\hline
$\left(\asyOpV{\frac32}{0}{\hbar_{\text{ren}}} \hbar \right)(\hbar_{\text{ren}})$&$e^{- \frac{15}{4}} \frac{\hbar^{-2}}{\sqrt{2} \pi}$&$-36$&$ \frac{387}{2}$&$- \frac{13785}{32}$&$ \frac{276705}{256}$&$- \frac{8524035}{8192}$&$ \frac{486577005}{65536}$\\
\hline
$\left(\asyOpV{\frac32}{0}{\hbar_{\text{ren}}} z_{\varphi^2} \right)(\hbar_{\text{ren}})$&$e^{- \frac{15}{4}} \frac{\hbar^{-2}}{\sqrt{2} \pi}$&$-18$&$ \frac{219}{4}$&$ \frac{567}{64}$&$ \frac{49113}{512}$&$ \frac{8281053}{16384}$&$ \frac{397802997}{131072}$\\
\hline
$\left(\asyOpV{\frac32}{0}{\hbar_{\text{ren}}} z_{\varphi^4} \right)(\hbar_{\text{ren}})$&$e^{- \frac{15}{4}} \frac{\hbar^{-3}}{\sqrt{2} \pi}$&$-36$&$ \frac{243}{2}$&$- \frac{729}{32}$&$ \frac{51057}{256}$&$ \frac{7736445}{8192}$&$ \frac{377172477}{65536}$\\
\hline
\end{tabular}
\subcaption{Table of the first coefficients of the asymptotics of the renormalization quantities in $\varphi^4$-theory.}
\label{tab:phi4renasymp}
\end{subtable}
\caption{Renormalization constants in $\varphi^4$-theory.}
\end{table}

Using the procedure established in Section \ref{sec:hopfalgebra}, the renormalization constants can be calculated by defining the invariant charge as 
\begin{align*} \alpha(\hbar):= \left( \frac{ \left( \left. \partial_{\varphi_c}^4 G^{\varphi^4} \right|_{\varphi_c=0}(\hbar) \right)^{\frac12} }{ \left. \partial_{\varphi_c}^2 G^{\varphi^4} \right|_{\varphi_c=0} (\hbar) } \right)^2. \end{align*}
Having defined the invariant charge, the calculation of the renormalization constants is completely equivalent to the calculation for $\varphi^3$-theory. 
The results are given in Table \ref{tab:phi4ren} and \ref{tab:phi4renasymp}.

Argyres et al.\ \cite{argyres2001zero} remarked that $1-z_{\varphi_c^4}(\hbar_\text{ren})$ does not count the number of skeleton diagrams in $\varphi^4$-theory as might be expected from analogy to $\varphi^3$-theory. The fact that this cannot by the case can be seen from the second term of $z_{\varphi_c^4}(\hbar_\text{ren})$ which is positive (see Table \ref{tab:phi4ren}), destroying a counting function interpretation of $1-z_{\varphi_c^4}(\hbar_\text{ren})$. In \cite{borinsky2016generating} it was proven by the author that additionally to skeleton diagrams, also chains of one loop diagrams, $\begin{tikzpicture}[x=2ex,y=2ex,baseline={([yshift=-.5ex]current bounding box.center)}] \coordinate (v0); \coordinate [right=.5 of v0] (vm1); \coordinate [right=.5 of vm1] (v1); \node [right=.5 of v1] (v2) {$\ldots$}; \coordinate [right=.5 of v1] (v1m); \coordinate [right=.01 of v2] (v2m); \coordinate [right=.5 of v2m] (v3); \coordinate [right=.5 of v3] (vm2); \coordinate [right=.5 of vm2] (v4); \coordinate [above left=.5 of v0] (i0); \coordinate [below left=.5 of v0] (i1); \coordinate [above right=.5 of v4] (o0); \coordinate [below right=.5 of v4] (o1); \draw (vm1) circle(.5); \draw (vm2) circle(.5); \draw (i0) -- (v0); \draw (i1) -- (v0); \draw (o0) -- (v4); \draw (o1) -- (v4); \draw ([shift=(90:.5)]v1m) arc (90:270:.5); \draw ([shift=(-90:.5)]v2m) arc (-90:90:.5); \filldraw (v0) circle(1pt); \filldraw (v1) circle(1pt); \filldraw (v3) circle(1pt); \filldraw (v4) circle(1pt); \end{tikzpicture}$ contribute to the generating function $z_{\varphi_c^4}(\hbar_\text{ren})$. The chains of one loop bubbles contribute with alternating sign.

The generating function of skeleton diagrams in $\varphi^4$ theory is given by \cite{borinsky2016generating},
\begin{align} 1 - z_{\varphi_c^4}(\hbar_\text{ren}) + 3 \sum \limits_{n\geq 2} (-1)^n \left( \frac{\hbar_\text{ren} }{2} \right)^n, \end{align}
where we needed to include a factor of $4!$ to convert from \cite{borinsky2016generating} to the present notation of leg-fixed diagrams.
The first coefficients are,
\begin{gather*} 0, \frac{3}{2},0,1,3, \frac{31}{2}, \frac{529}{6}, \frac{2277}{4}, \frac{16281}{4}, \frac{254633}{8}, \frac{2157349}{8}, \frac{39327755}{16}, \frac{383531565}{16}, \ldots \end{gather*}
The asymptotic expansion of this sequence agrees with the one of $( 1-z_{\varphi_c^4}(\hbar_\text{ren}))$,
\begin{align*} [\hbar_\text{ren}^n]( 1-z_{\varphi_c^4}(\hbar_\text{ren})) \underset{n\rightarrow \infty}{\sim} \frac{e^{-\frac{15}{4}}}{\sqrt{2}\pi} \left(\frac{2}{3}\right)^{n+3} \Gamma(n+3) \left( 36 - \frac{3}{2}\frac{243}{2} \frac{1}{n+2} \right. \\ \left. + \left(\frac{3}{2}\right)^2 \frac{729}{32}\frac{1}{(n+1)(n+2)} - \left(\frac{3}{2}\right)^3 \frac{51057}{256}\frac{1}{n(n+1)(n+2)} +\ldots \right). \end{align*}
More coefficients are given in Table \ref{tab:phi4renasymp}.
\subsection{QED-type theories}
We will discuss more general theories with two types of `particles', which are of QED-type in the sense that we can interpret one particle as boson and the other as fermion with a fermion-fermion-boson vertex. 

Consider the partition function
\begin{align*} Z(\hbar, j, \eta) &:= \int_\R \frac{dx}{\sqrt{2 \pi \hbar}} \int_\C\frac{dz d\bar z}{\pi \hbar} e^{\frac{1}{\hbar} \left( - \frac{x^2}{2} - |z|^2 + x |z|^2 + j x + \eta \bar z + \bar \eta z \right)}. \end{align*}
The Gaussian integration over $z$ and $\bar z$ can be performed immediately,
\begin{align} \begin{split} \label{eqn:qedprototype} Z(\hbar, j, \eta) &= \int_\R \frac{dx}{\sqrt{2 \pi \hbar}} \frac{dz d\bar z}{\pi \hbar} e^{\frac{1}{\hbar} \left( - \frac{x^2}{2} - (1-x)\left|z- \frac{\eta}{1-x}\right|^2 + j x + \frac{|\eta|^2}{1-x} \right)} \\ &= \int_\R \frac{dx}{\sqrt{2 \pi \hbar}} \frac{1}{1-x} e^{\frac{1}{\hbar} \left( - \frac{x^2}{2} + j x + \frac{|\eta|^2}{1-x} \right)} \\ &= \int_\R \frac{dx}{\sqrt{2 \pi \hbar}} e^{\frac{1}{\hbar} \left( - \frac{x^2}{2} + j x + \frac{|\eta|^2}{1-x} + \hbar \log \frac{1}{1-x} \right)} \end{split} \end{align}
Note that the transformation above has not been justified rigorously in the scope of formal integrals. A more powerful definition of `formal integrals' can be made which includes the above manipulations \cite{borinsky2017graph}. Here, it is sufficient to consider the last line in eq.\ \eqref{eqn:qedprototype} as input for our mathematical machinery and the previous as a physical motivation.
The combinatorial interpretation of this expression is the following: $\frac{|\eta|^2}{1-x}$ generates a fermion propagator line and $\hbar \log \frac{1}{1-x}$ generates a fermion loop, both with an arbitrary number of boson lines attached. 
The interpretation of the $j x$ and $-\frac{x^2}{2}$ terms are standard.

We will consider the following variations of this partition function:
\begin{labeling}{(Quenched QED)}
\item[(QED)]
In quantum electrodynamics (QED) all fermion loops have an even number of fermion edges, as Furry's theorem guarantees that diagrams with odd fermion loops vanish. The modification,
\begin{align*} \hbar \log \frac{1}{1-x} \rightarrow \hbar \frac12 \left( \log \frac{1}{1-x} + \log \frac{1}{1+x} \right) = \frac12\hbar \log \frac{1}{1-x^2}, \end{align*}
results in the required partition function \cite{cvitanovic1978number, itzykson2006quantum}.
\item[(Quenched QED)]
In the quenched approximation of QED, fermion loops are neglected altogether. This corresponds to the modification $\hbar \log \frac{1}{1-x} \rightarrow 0$.
\item[(Yukawa)]
We will also consider the integral without modification. Also odd fermion loops are allowed in this case. This can be seen as the zero-dimensional version of Yukawa theory.
\end{labeling}
\subsubsection{QED}
\label{sec:QED}

In QED the partition function in eq.\ \eqref{eqn:qedprototype} must be modified to
\begin{align*} Z^\text{QED}(\hbar, j, \eta):= \int_\R \frac{dx}{\sqrt{2 \pi \hbar}} e^{\frac{1}{\hbar} \left( - \frac{x^2}{2} + j x + \frac{|\eta|^2}{1-x} + \frac12 \hbar \log \frac{1}{1-x^2} \right)}. \end{align*}
As in $\varphi^4$-theory, we hide the dependence on the sources inside a composition:
\begin{align*} Z^\text{QED}(\hbar, j, \eta)&:= \left( 1 + \frac{j^2}{2\hbar} + \frac{|\eta|^2}{\hbar} \right)  Z_0^\text{QED} \left( \frac{\hbar \left( 1 + \frac{2 j |\eta|^2}{\hbar^2} \right)}{\left(1 - \frac{2 |\eta|^2}{\hbar}\right)\left( 1-\frac{j^2}{\hbar} \right)} \right) \\ &+ \bigO(j^4) + \bigO(j^2 |\eta|^2) + \bigO(|\eta|^4), \intertext{where} Z^\text{QED}_0(\hbar) &:= Z^\text{QED}(\hbar, 0, 0) = \int_\R \frac{dx}{\sqrt{2 \pi \hbar}} e^{\frac{1}{\hbar} \left( - \frac{x^2}{2} + \frac12 \hbar \log \frac{1}{1-x^2} \right)}. \end{align*}
Recall that this expression is meant to be expanded under the integral sign. Because $e^{\frac12 \log \frac{1}{1-x^2}} = \frac{1}{\sqrt{1-x^2}}$, we conclude, using the rules of Gaussian integration that
\begin{align*} Z^\text{QED}_0(\hbar)= \sum_{n=0}^\infty \hbar^{n} (2n-1)!! [x^{2n}] \frac{1}{\sqrt{1-x^2}}. \end{align*}
In Example \ref{expl:sinegordoncurve} it was shown using Proposition \ref{prop:formalchangeofvar} that this may be written as,
\begin{align*} Z^\text{QED}_0(\hbar)= \Fop\left[ -\frac{\sin^2(x)}{2}\right](\hbar). \end{align*}
The partition function of zero-dimensional QED without sources is therefore equal to the partition function of the zero-dimensional sine-Gordon model. 

Using Corollary \ref{crll:comb_int_asymp}, it is straightforward to calculate the all-order asymptotics. 
The saddle points of $-\frac{\sin^2(x)}{2}$ all lie on the real axis. The dominant saddles are at $\tau_\pm = \pm \frac{\pi}{2}$. We find that $A = -\frac{\sin^2(\tau_\pm)}{2} = \frac12$ and $\Sact(\tau_\pm) - \Sact(\tau_\pm+x) = -\frac{\sin^2(x)}{2}$. Therefore, $Z^\text{QED}_0 \in \fring{\hbar}{\frac12}{0}$ and
\begin{align*} \asyOpV{\frac12}{0}{\hbar} Z^\text{QED}_0(\hbar)= \asyOpV{\frac12}{0}{\hbar} \Fop\left[ -\frac{\sin^2(x)}{2}\right](\hbar) = \frac{2}{2\pi} \Fop\left[ -\frac{\sin^2(x)}{2}\right](-\hbar). \end{align*}
The calculation of the asymptotics of $Z^\text{QED}(\hbar, j, \eta)$ as well as setting up the free energy $W^\text{QED}(\hbar,j,\eta)$ and calculating its asymptotics are analogous to the preceding examples. The respective first coefficients are listed in Tables \ref{tab:Zqed} and \ref{tab:Wqed}.

\begin{table}
\begin{subtable}[c]{\textwidth}
\centering
\tiny
\def\arraystretch{1.5}
\begin{tabular}{|c||c||c|c|c|c|c|c|}
\hline
&prefactor&$\hbar^{0}$&$\hbar^{1}$&$\hbar^{2}$&$\hbar^{3}$&$\hbar^{4}$&$\hbar^{5}$\\
\hline\hline
$\partial_j^{0} (\partial_{\eta} \partial_{ \bar \eta} )^{0} Z^{\text{QED}} \big|_{\substack{j=0\\\eta=0}}$&$\hbar^{0}$&$1$&$ \frac{1}{2}$&$ \frac{9}{8}$&$ \frac{75}{16}$&$ \frac{3675}{128}$&$ \frac{59535}{256}$\\
\hline
$\partial_j^{2} (\partial_{\eta} \partial_{ \bar \eta} )^{0} Z^{\text{QED}} \big|_{\substack{j=0\\\eta=0}}$&$\hbar^{-1}$&$1$&$ \frac{3}{2}$&$ \frac{45}{8}$&$ \frac{525}{16}$&$ \frac{33075}{128}$&$ \frac{654885}{256}$\\
\hline
$\partial_j^{0} (\partial_{\eta} \partial_{ \bar \eta} )^{1} Z^{\text{QED}} \big|_{\substack{j=0\\\eta=0}}$&$\hbar^{-1}$&$1$&$ \frac{3}{2}$&$ \frac{45}{8}$&$ \frac{525}{16}$&$ \frac{33075}{128}$&$ \frac{654885}{256}$\\
\hline
$\partial_j^{1} (\partial_{\eta} \partial_{ \bar \eta} )^{1} Z^{\text{QED}} \big|_{\substack{j=0\\\eta=0}}$&$\hbar^{-1}$&$1$&$ \frac{9}{2}$&$ \frac{225}{8}$&$ \frac{3675}{16}$&$ \frac{297675}{128}$&$ \frac{7203735}{256}$\\
\hline
\end{tabular}
\subcaption{The first coefficients of the trivariate generating function $Z^{\text{QED}}(\hbar, j, \eta)$.}
\end{subtable}
\begin{subtable}[c]{\textwidth}
\centering
\tiny
\def\arraystretch{1.5}
\begin{tabular}{|c||c||c|c|c|c|c|c|}
\hline
&prefactor&$\hbar^{0}$&$\hbar^{1}$&$\hbar^{2}$&$\hbar^{3}$&$\hbar^{4}$&$\hbar^{5}$\\
\hline\hline
$\asyOpV{\frac12}{0}{\hbar} \partial_j^{0} (\partial_{\eta} \partial_{ \bar \eta} )^{0} Z^{\text{QED}} \big|_{\substack{j=0\\\eta=0}}$&$\frac{\hbar^{0}}{\pi}$&$1$&$- \frac{1}{2}$&$ \frac{9}{8}$&$- \frac{75}{16}$&$ \frac{3675}{128}$&$- \frac{59535}{256}$\\
\hline
$\asyOpV{\frac12}{0}{\hbar} \partial_j^{2} (\partial_{\eta} \partial_{ \bar \eta} )^{0} Z^{\text{QED}} \big|_{\substack{j=0\\\eta=0}}$&$\frac{\hbar^{-2}}{\pi}$&$1$&$ \frac{1}{2}$&$- \frac{3}{8}$&$ \frac{15}{16}$&$- \frac{525}{128}$&$ \frac{6615}{256}$\\
\hline
$\asyOpV{\frac12}{0}{\hbar} \partial_j^{0} (\partial_{\eta} \partial_{ \bar \eta} )^{1} Z^{\text{QED}} \big|_{\substack{j=0\\\eta=0}}$&$\frac{\hbar^{-2}}{\pi}$&$1$&$ \frac{1}{2}$&$- \frac{3}{8}$&$ \frac{15}{16}$&$- \frac{525}{128}$&$ \frac{6615}{256}$\\
\hline
$\asyOpV{\frac12}{0}{\hbar} \partial_j^{0} (\partial_{\eta} \partial_{ \bar \eta} )^{2} Z^{\text{QED}} \big|_{\substack{j=0\\\eta=0}}$&$\frac{\hbar^{-3}}{\pi}$&$1$&$- \frac{1}{2}$&$ \frac{1}{8}$&$- \frac{3}{16}$&$ \frac{75}{128}$&$- \frac{735}{256}$\\
\hline
\end{tabular}
\subcaption{The first coefficients of the trivariate generating function $\asyOpV{\frac12}{0}{\hbar}Z^{\text{QED}}(\hbar, j, \eta)$.}
\end{subtable}
\caption{Partition function in QED.}
\label{tab:Zqed}
\end{table}

\begin{table}
\begin{subtable}[c]{\textwidth}
\centering
\tiny
\def\arraystretch{1.5}
\begin{tabular}{|c||c|c|c|c|c|c|}
\hline
&$\hbar^{0}$&$\hbar^{1}$&$\hbar^{2}$&$\hbar^{3}$&$\hbar^{4}$&$\hbar^{5}$\\
\hline\hline
$\partial_j^{0} (\partial_{\eta} \partial_{ \bar \eta} )^{0} W^{\text{QED}} \big|_{\substack{j=0\\\eta=0}}$&$0$&$0$&$ \frac{1}{2}$&$1$&$ \frac{25}{6}$&$26$\\
\hline
$\partial_j^{2} (\partial_{\eta} \partial_{ \bar \eta} )^{0} W^{\text{QED}} \big|_{\substack{j=0\\\eta=0}}$&$1$&$1$&$4$&$25$&$208$&$2146$\\
\hline
$\partial_j^{0} (\partial_{\eta} \partial_{ \bar \eta} )^{1} W^{\text{QED}} \big|_{\substack{j=0\\\eta=0}}$&$1$&$1$&$4$&$25$&$208$&$2146$\\
\hline
$\partial_j^{1} (\partial_{\eta} \partial_{ \bar \eta} )^{1} W^{\text{QED}} \big|_{\substack{j=0\\\eta=0}}$&$1$&$4$&$25$&$208$&$2146$&$26368$\\
\hline
\end{tabular}
\subcaption{The first coefficients of the trivariate generating function $W^{\text{QED}}(\hbar, j, \eta)$.}
\end{subtable}
\begin{subtable}[c]{\textwidth}
\centering
\tiny
\def\arraystretch{1.5}
\begin{tabular}{|c||c||c|c|c|c|c|c|}
\hline
&prefactor&$\hbar^{0}$&$\hbar^{1}$&$\hbar^{2}$&$\hbar^{3}$&$\hbar^{4}$&$\hbar^{5}$\\
\hline\hline
$\asyOpV{\frac12}{0}{\hbar} \partial_j^{0} (\partial_{\eta} \partial_{ \bar \eta} )^{0} W^{\text{QED}} \big|_{\substack{j=0\\\eta=0}}$&$\frac{\hbar^{1}}{\pi}$&$1$&$-1$&$ \frac{1}{2}$&$- \frac{17}{2}$&$ \frac{67}{8}$&$- \frac{3467}{8}$\\
\hline
$\asyOpV{\frac12}{0}{\hbar} \partial_j^{2} (\partial_{\eta} \partial_{ \bar \eta} )^{0} W^{\text{QED}} \big|_{\substack{j=0\\\eta=0}}$&$\frac{\hbar^{-1}}{\pi}$&$1$&$-1$&$- \frac{3}{2}$&$- \frac{13}{2}$&$- \frac{341}{8}$&$- \frac{2931}{8}$\\
\hline
$\asyOpV{\frac12}{0}{\hbar} \partial_j^{0} (\partial_{\eta} \partial_{ \bar \eta} )^{1} W^{\text{QED}} \big|_{\substack{j=0\\\eta=0}}$&$\frac{\hbar^{-1}}{\pi}$&$1$&$-1$&$- \frac{3}{2}$&$- \frac{13}{2}$&$- \frac{341}{8}$&$- \frac{2931}{8}$\\
\hline
$\asyOpV{\frac12}{0}{\hbar} \partial_j^{1} (\partial_{\eta} \partial_{ \bar \eta} )^{1} W^{\text{QED}} \big|_{\substack{j=0\\\eta=0}}$&$\frac{\hbar^{-2}}{\pi}$&$1$&$-1$&$- \frac{3}{2}$&$- \frac{13}{2}$&$- \frac{341}{8}$&$- \frac{2931}{8}$\\
\hline
\end{tabular}
\subcaption{The first coefficients of the trivariate generating function $\asyOpV{\frac12}{0}{\hbar}W^{\text{QED}}(\hbar, j, \eta)$.}
\end{subtable}
\caption{Free energy in QED.}
\label{tab:Wqed}
\end{table}

The effective action is given by the two variable Legendre transformation of $W^{\text{QED}}$:
\begin{align*} G^{\text{QED}}(\hbar, \phi_c, \psi_c) &= W^{\text{QED}}(\hbar, j, \eta) - j \phi_c - \bar \eta \psi_c - \eta \bar \psi_c, \end{align*}
where $\phi_c = \partial_j W^{\text{QED}}$ and $\psi_c = \partial_{\bar\eta} W^{\text{QED}}$. Because there are no tadpole diagrams in QED, it follows that,
\begin{gather*} G^\text{QED} \big|_{\substack{\varphi_c=0\\\psi_c=0}} = W^\text{QED}\big|_{\substack{j=0\\\eta=0}} \\ \partial_{\psi_c} \partial_{\bar \psi_c} G^\text{QED} \big|_{\substack{\varphi_c=0\\\psi_c=0}} = - \frac{1}{\partial_\eta \partial_{\bar \eta} W^\text{QED}\big|_{\substack{j=0\\\eta=0}} } \\ \partial_{\varphi_c}^2 G^\text{QED} \big|_{\substack{\varphi_c=0\\\psi_c=0}} = - \frac{1}{\partial_j^2 W^\text{QED}\big|_{\substack{j=0\\\eta=0}} } \\ \partial_{\varphi_c}\partial_{\psi_c} \partial_{\bar \psi_c} G^\text{QED} \big|_{\substack{\varphi_c=0\\\psi_c=0}} = \frac{ \partial_j \partial_\eta \partial_{\bar \eta} W^\text{QED}\big|_{\substack{j=0\\\eta=0}} }{ \partial_j^2 W^\text{QED}\big|_{\substack{j=0\\\eta=0}} \left( \partial_\eta \partial_{\bar \eta} W^\text{QED}\big|_{\substack{j=0\\\eta=0}} \right)^2 }. \end{gather*}
The calculation of asymptotics is similar to the example of $\varphi^4$-theory. Coefficients for the effective action are listed in Table \ref{tab:GQED}.

\begin{table}
\begin{subtable}[c]{\textwidth}
\centering
\tiny
\def\arraystretch{1.5}
\begin{tabular}{|c||c|c|c|c|c|c|}
\hline
&$\hbar^{0}$&$\hbar^{1}$&$\hbar^{2}$&$\hbar^{3}$&$\hbar^{4}$&$\hbar^{5}$\\
\hline\hline
$\partial_{\phi_c}^{0} (\partial_{\psi_c} \partial_{ \bar \psi_c} )^{0} G^{\text{QED}} \big|_{\substack{\phi_c=0\\\psi_c=0}}$&$0$&$0$&$ \frac{1}{2}$&$1$&$ \frac{25}{6}$&$26$\\
\hline
$\partial_{\phi_c}^{2} (\partial_{\psi_c} \partial_{ \bar \psi_c} )^{0} G^{\text{QED}} \big|_{\substack{\phi_c=0\\\psi_c=0}}$&$-1$&$1$&$3$&$18$&$153$&$1638$\\
\hline
$\partial_{\phi_c}^{0} (\partial_{\psi_c} \partial_{ \bar \psi_c} )^{1} G^{\text{QED}} \big|_{\substack{\phi_c=0\\\psi_c=0}}$&$-1$&$1$&$3$&$18$&$153$&$1638$\\
\hline
$\partial_{\phi_c}^{1} (\partial_{\psi_c} \partial_{ \bar \psi_c} )^{1} G^{\text{QED}} \big|_{\substack{\phi_c=0\\\psi_c=0}}$&$1$&$1$&$7$&$72$&$891$&$12672$\\
\hline
\end{tabular}
\subcaption{The first coefficients of the trivariate generating function $G^{\text{QED}}(\hbar, \phi_c, \psi_c)$.}
\end{subtable}
\begin{subtable}[c]{\textwidth}
\centering
\tiny
\def\arraystretch{1.5}
\begin{tabular}{|c||c||c|c|c|c|c|c|}
\hline
&prefactor&$\hbar^{0}$&$\hbar^{1}$&$\hbar^{2}$&$\hbar^{3}$&$\hbar^{4}$&$\hbar^{5}$\\
\hline\hline
$\asyOpV{\frac12}{0}{\hbar} \partial_{\phi_c}^{0} (\partial_{\psi_c} \partial_{ \bar \psi_c} )^{0} G^{\text{QED}} \big|_{\substack{\phi_c=0\\\psi_c=0}}$&$\frac{\hbar^{1}}{\pi}$&$1$&$-1$&$ \frac{1}{2}$&$- \frac{17}{2}$&$ \frac{67}{8}$&$- \frac{3467}{8}$\\
\hline
$\asyOpV{\frac12}{0}{\hbar} \partial_{\phi_c}^{2} (\partial_{\psi_c} \partial_{ \bar \psi_c} )^{0} G^{\text{QED}} \big|_{\substack{\phi_c=0\\\psi_c=0}}$&$\frac{\hbar^{-1}}{\pi}$&$1$&$-3$&$- \frac{9}{2}$&$- \frac{57}{2}$&$- \frac{2025}{8}$&$- \frac{22437}{8}$\\
\hline
$\asyOpV{\frac12}{0}{\hbar} \partial_{\phi_c}^{0} (\partial_{\psi_c} \partial_{ \bar \psi_c} )^{1} G^{\text{QED}} \big|_{\substack{\phi_c=0\\\psi_c=0}}$&$\frac{\hbar^{-1}}{\pi}$&$1$&$-3$&$- \frac{9}{2}$&$- \frac{57}{2}$&$- \frac{2025}{8}$&$- \frac{22437}{8}$\\
\hline
$\asyOpV{\frac12}{0}{\hbar} \partial_{\phi_c}^{1} (\partial_{\psi_c} \partial_{ \bar \psi_c} )^{1} G^{\text{QED}} \big|_{\substack{\phi_c=0\\\psi_c=0}}$&$\frac{\hbar^{-2}}{\pi}$&$1$&$-7$&$- \frac{3}{2}$&$- \frac{75}{2}$&$- \frac{3309}{8}$&$- \frac{41373}{8}$\\
\hline
\end{tabular}
\subcaption{The first coefficients of the trivariate generating function $\asyOpV{\frac12}{0}{\hbar}G^{\text{QED}}(\hbar, \phi_c, \psi_c)$.}
\end{subtable}
\caption{Effective action in QED.}
\label{tab:GQED}
\end{table}

To calculate the renormalization constants we define the invariant charge as,
\begin{align*} \alpha(\hbar):= \left( \frac{ \partial_{\varphi_c}\partial_{\psi_c} \partial_{\bar \psi_c} G^\text{QED} \big|_{\substack{\varphi_c=0\\\psi_c=0}} } { \left( - \partial_{\varphi_c}^2 G^\text{QED} \big|_{\substack{\varphi_c=0\\\psi_c=0}} \right)^\frac12 \left( - \partial_{\psi_c} \partial_{\bar \psi_c} G^\text{QED}\big|_{\substack{\varphi_c=0\\\psi_c=0}} \right) } \right)^2. \end{align*}
The first coefficients of the renormalization constants and their asymptotics are listed in Table \ref{tab:QEDren}.

\begin{table}
\begin{subtable}[c]{\textwidth}
\centering
\tiny
\def\arraystretch{1.5}
\begin{tabular}{|c||c|c|c|c|c|c|}
\hline
&$\hbar_{\text{ren}}^{0}$&$\hbar_{\text{ren}}^{1}$&$\hbar_{\text{ren}}^{2}$&$\hbar_{\text{ren}}^{3}$&$\hbar_{\text{ren}}^{4}$&$\hbar_{\text{ren}}^{5}$\\
\hline\hline
$\hbar(\hbar_{\text{ren}})$&$0$&$1$&$-5$&$14$&$-58$&$20$\\
\hline
$z_{\phi_c^2}(\hbar_{\text{ren}})$&$1$&$1$&$-1$&$-1$&$-13$&$-93$\\
\hline
$z_{|\psi_c|^2}(\hbar_{\text{ren}})$&$1$&$1$&$-1$&$-1$&$-13$&$-93$\\
\hline
$z_{\phi_c |\psi_c|^2}(\hbar_{\text{ren}})$&$1$&$-1$&$-1$&$-13$&$-93$&$-1245$\\
\hline
\end{tabular}
\subcaption{Table of the first coefficients of the renormalization quantities in QED.}
\end{subtable}
\begin{subtable}[c]{\textwidth}
\centering
\tiny
\def\arraystretch{1.5}
\begin{tabular}{|c||c||c|c|c|c|c|c|}
\hline
&prefactor&$\hbar_{\text{ren}}^{0}$&$\hbar_{\text{ren}}^{1}$&$\hbar_{\text{ren}}^{2}$&$\hbar_{\text{ren}}^{3}$&$\hbar_{\text{ren}}^{4}$&$\hbar_{\text{ren}}^{5}$\\
\hline\hline
$\left(\asyOpV{\frac12}{0}{\hbar_{\text{ren}}} \hbar \right)(\hbar_{\text{ren}})$&$e^{- \frac{5}{2}} \frac{\hbar^{-1}}{\pi}$&$-2$&$24$&$- \frac{379}{4}$&$ \frac{6271}{12}$&$ \frac{38441}{64}$&$ \frac{17647589}{480}$\\
\hline
$\left(\asyOpV{\frac12}{0}{\hbar_{\text{ren}}} z_{\phi_c^2} \right)(\hbar_{\text{ren}})$&$e^{- \frac{5}{2}} \frac{\hbar^{-1}}{\pi}$&$-1$&$ \frac{13}{2}$&$ \frac{67}{8}$&$ \frac{5177}{48}$&$ \frac{513703}{384}$&$ \frac{83864101}{3840}$\\
\hline
$\left(\asyOpV{\frac12}{0}{\hbar_{\text{ren}}} z_{|\psi_c|^2} \right)(\hbar_{\text{ren}})$&$e^{- \frac{5}{2}} \frac{\hbar^{-1}}{\pi}$&$-1$&$ \frac{13}{2}$&$ \frac{67}{8}$&$ \frac{5177}{48}$&$ \frac{513703}{384}$&$ \frac{83864101}{3840}$\\
\hline
$\left(\asyOpV{\frac12}{0}{\hbar_{\text{ren}}} z_{\phi_c |\psi_c|^2} \right)(\hbar_{\text{ren}})$&$e^{- \frac{5}{2}} \frac{\hbar^{-2}}{\pi}$&$-1$&$ \frac{13}{2}$&$ \frac{67}{8}$&$ \frac{5177}{48}$&$ \frac{513703}{384}$&$ \frac{83864101}{3840}$\\
\hline
\end{tabular}
\subcaption{Table of the first coefficients of the asymptotics of the renormalization quantities in QED.}
\end{subtable}
\caption{Renormalization constants in QED.}
\label{tab:QEDren}
\end{table}

As in the example of $\varphi^3$-theory, the $z$-factor for the vertex, $z_{\varphi_c \psi_c \bar \psi_c}$ can be used to enumerate the number of skeleton diagrams. Asymptotically, this number is given by,
\begin{gather*} [\hbar_\text{ren}^n]( 1-z_{\varphi_c \psi_c \bar \psi_c}(\hbar_\text{ren})) \underset{n\rightarrow \infty}{\sim} \frac{e^{-\frac{5}{2}}}{\pi} \left(\frac{1}{2}\right)^{-n-2} \Gamma(n+2) \left( 1 - \frac{1}{2}\frac{13}{2} \frac{1}{n+1} \right. \\ \left. - \left(\frac{1}{2}\right)^2 \frac{67}{8}\frac{1}{n(n+1)} - \left(\frac{1}{2}\right)^3 \frac{5177}{48}\frac{1}{(n-1)n(n+1)} +\ldots \right), \end{gather*}
which can be read off Table \ref{tab:QEDren}. The first two coefficients of this expansion were also given in \cite{cvitanovic1978number} in a different notation.
\subsubsection{Quenched QED}
For the quenched approximation, we need to remove the $\log$-term in the partition function given in eq.\ \eqref{eqn:qedprototype}:
\begin{align*} Z^\text{QQED}(\hbar, j, \eta) := \int_\R \frac{dx}{\sqrt{2 \pi \hbar}} e^{\frac{1}{\hbar} \left( - \frac{x^2}{2} + j x + \frac{|\eta|^2}{1-x} \right)} \end{align*}
The partition function cannot be reduced to a generating function over diagrams without sources as the only diagram without sources is the empty diagram. 

To obtain the first order in $|\eta|^2$, the partition function can be rewritten as,
\begin{gather*} Z^\text{QQED}(\hbar, j, \eta) = \\ e^{\frac{j^2}{2\hbar}} \left( 1+ \frac{|\eta|^2}{\hbar(1-j)} \int_\R \frac{dx}{\sqrt{2 \pi \left( \frac{\hbar}{(1-j)^2} \right)}} \frac{1}{1-x} e^{-\frac{x^2}{2\left( \frac{\hbar}{(1-j)^2} \right)} } + \bigO(|\eta|^4) \right). \end{gather*}
The formal integral in this expression can be easily expanded:
\begin{align*} \int_\R \frac{dx}{\sqrt{2 \pi \hbar}} \frac{1}{1-x} e^{-\frac{x^2}{2\hbar} } &= \sum_{n=0}^\infty \hbar^n(2n-1)!! =: \chi(\hbar) \end{align*}
This is in fact the expression, we encountered in Example \ref{expl:counterexpl}, whose asymptotics cannot be calculated by Corollary \ref{crll:comb_int_asymp} or Theorem \ref{thm:comb_int_asymp}. But extracting the asymptotics `by hand' is trivial. Because $(2n-1)!! = \frac{2^{n+\frac12}}{\sqrt{2 \pi}}\Gamma\left(n+\frac12\right)$, we can write,
\begin{align*} \asyOpV{\frac12}{}{\hbar} \chi(\hbar) = \frac{1}{\sqrt{2 \pi \hbar}}, \end{align*}
in the language of the ring of factorially divergent power series. It follows that,
\begin{align*} Z^\text{QQED}(\hbar, j, \eta) &= e^{\frac{j^2}{2\hbar}} \left( 1+ \frac{|\eta|^2}{\hbar(1-j)} \chi\left(\frac{\hbar}{(1-j)^2}\right) + \bigO(|\eta|^4) \right) \\ \asyOpV{\frac12}{}{\hbar} Z^\text{QQED}(\hbar, j, \eta) &= \frac{|\eta|^2 e^{\frac{j^2}{2\hbar}}}{\hbar (1-j)} \left( \asyOpV{\frac12}{}{\hbar} \chi\left( \frac{\hbar}{(1-j)^2} \right) \right) (\hbar) + \bigO(|\eta|^4) \intertext{and by the chain rule for $\asyOp$,} \asyOpV{\frac12}{}{\hbar} Z^\text{QQED}(\hbar, j, \eta) &= \frac{|\eta|^2 e^{\frac{j^2}{2\hbar}}}{\hbar (1-j)} \left[ e^{\frac12 \left( \frac{1}{\hbar} - \frac{1}{\widetilde \hbar} \right)} \asyOpV{\frac12}{}{\widetilde \hbar} \chi\left( \widetilde \hbar\right) \right]_{\widetilde \hbar = \frac{\hbar}{(1-j)^2}} + \bigO(|\eta|^4) \\ &= \frac{|\eta|^2 e^{\frac{j^2}{2\hbar}}}{\hbar (1-j)} \left[ e^{\frac12 \left( \frac{1}{\hbar} - \frac{1}{\widetilde \hbar} \right)} \frac{1}{\sqrt{2 \pi \widetilde \hbar}} \right]_{\widetilde \hbar = \frac{\hbar}{(1-j)^2}} + \bigO(|\eta|^4) \\ &= \frac{|\eta|^2 e^{\frac{j}{\hbar}}}{\sqrt{2\pi} \hbar^{\frac32}} + \bigO(|\eta|^4) \end{align*}
Obtaining the free energy, which is essentially equivalent to the partition function, is straightforward,
\begin{align*} W^\text{QQED}(\hbar, j, \eta) &= \hbar \log Z^\text{QQED}(\hbar, j, \eta) \\ &= \frac{j^2}{2}+ \frac{|\eta|^2}{1-j} \chi\left(\frac{\hbar}{(1-j)^2}\right) + \bigO(|\eta|^4) \\ \asyOpV{\frac12}{}{\hbar} W^\text{QQED}(\hbar, j, \eta) &= \frac{|\eta|^2e^{\frac{j-\frac{j^2}{2}}{\hbar}}}{\sqrt{2\pi \hbar}} + \bigO(|\eta|^4). \end{align*}
\begin{table}
\begin{subtable}[c]{\textwidth}
\centering
\tiny
\def\arraystretch{1.5}
\begin{tabular}{|c||c|c|c|c|c|c|}
\hline
&$\hbar^{0}$&$\hbar^{1}$&$\hbar^{2}$&$\hbar^{3}$&$\hbar^{4}$&$\hbar^{5}$\\
\hline\hline
$\partial_j^{0} (\partial_{\eta} \partial_{ \bar \eta} )^{1} W^{\text{QQED}} \big|_{\substack{j=0\\\eta=0}}$&$1$&$1$&$3$&$15$&$105$&$945$\\
\hline
$\partial_j^{1} (\partial_{\eta} \partial_{ \bar \eta} )^{1} W^{\text{QQED}} \big|_{\substack{j=0\\\eta=0}}$&$1$&$3$&$15$&$105$&$945$&$10395$\\
\hline
\end{tabular}
\subcaption{The first coefficients of the trivariate generating function $W^{\text{QQED}}(\hbar, j, \eta)$.}
\end{subtable}
\begin{subtable}[c]{\textwidth}
\centering
\tiny
\def\arraystretch{1.5}
\begin{tabular}{|c||c||c|c|c|c|c|c|}
\hline
&prefactor&$\hbar^{0}$&$\hbar^{1}$&$\hbar^{2}$&$\hbar^{3}$&$\hbar^{4}$&$\hbar^{5}$\\
\hline\hline
$\asyOpV{\frac12}{0}{\hbar} \partial_j^{0} (\partial_{\eta} \partial_{ \bar \eta} )^{1} W^{\text{QQED}} \big|_{\substack{j=0\\\eta=0}}$&$\frac{\hbar^{0}}{\sqrt{2 \pi \hbar}}$&$1$&$0$&$0$&$0$&$0$&$0$\\
\hline
$\asyOpV{\frac12}{0}{\hbar} \partial_j^{1} (\partial_{\eta} \partial_{ \bar \eta} )^{1} W^{\text{QQED}} \big|_{\substack{j=0\\\eta=0}}$&$\frac{\hbar^{-1}}{\sqrt{2 \pi \hbar}}$&$1$&$0$&$0$&$0$&$0$&$0$\\
\hline
\end{tabular}
\subcaption{The first coefficients of the trivariate generating function $\asyOpV{\frac12}{0}{\hbar}W^{\text{QQED}}(\hbar, j, \eta)$.}
\end{subtable}
\caption{Free energy in quenched QED.}
\label{tab:WQQED}
\end{table}%
The effective action obtained by the Legendre transformation of $W^{\text{QQED}}$ can also be expressed explicitly:
\begin{align*} G^\text{QQED}(\hbar, \varphi_c, \psi_c) &= -\frac{\varphi_c^2}{2}+ |\psi_c|^2\frac{(\varphi_c-1)}{\chi\left(\frac{\hbar}{(1-\varphi_c)^2}\right)} + \bigO(|\psi_c|^4) \\ \asyOpV{\frac12}{}{\hbar} G^\text{QQED}(\hbar, \varphi_c, \psi_c) &= |\psi_c|^2\frac{e^{\frac{\varphi_c - \frac{\varphi_c^2}{2}}{\hbar}}}{\sqrt{2\pi \hbar}} \frac{(1-\varphi_c)^2}{\chi\left(\frac{\hbar}{(1-\varphi_c)^2}\right)^2} + \bigO(|\psi_c|^4). \end{align*}
The first coefficients of the free energy and effective action are listed in Tables \ref{tab:WQQED} and \ref{tab:GQQED} together with the respective asymptotics.

\begin{table}
\begin{subtable}[c]{\textwidth}
\centering
\tiny
\def\arraystretch{1.5}
\begin{tabular}{|c||c|c|c|c|c|c|}
\hline
&$\hbar^{0}$&$\hbar^{1}$&$\hbar^{2}$&$\hbar^{3}$&$\hbar^{4}$&$\hbar^{5}$\\
\hline\hline
$\partial_{\phi_c}^{0} (\partial_{\psi_c} \partial_{ \bar \psi_c} )^{1} G^{\text{QQED}} \big|_{\substack{\phi_c=0\\\psi_c=0}}$&$-1$&$1$&$2$&$10$&$74$&$706$\\
\hline
$\partial_{\phi_c}^{1} (\partial_{\psi_c} \partial_{ \bar \psi_c} )^{1} G^{\text{QQED}} \big|_{\substack{\phi_c=0\\\psi_c=0}}$&$1$&$1$&$6$&$50$&$518$&$6354$\\
\hline
\end{tabular}
\subcaption{The first coefficients of the trivariate generating function $\Gamma^{\text{QQED}}(\hbar, \phi_c, \psi_c)$.}
\end{subtable}
\begin{subtable}[c]{\textwidth}
\centering
\tiny
\def\arraystretch{1.5}
\begin{tabular}{|c||c||c|c|c|c|c|c|}
\hline
&prefactor&$\hbar^{0}$&$\hbar^{1}$&$\hbar^{2}$&$\hbar^{3}$&$\hbar^{4}$&$\hbar^{5}$\\
\hline\hline
$\asyOpV{\frac12}{0}{\hbar} \partial_{\phi_c}^{0} (\partial_{\psi_c} \partial_{ \bar \psi_c} )^{1} G^{\text{QQED}} \big|_{\substack{\phi_c=0\\\psi_c=0}}$&$\frac{\hbar^{0}}{\sqrt{2\pi\hbar}}$&$1$&$-2$&$-3$&$-16$&$-124$&$-1224$\\
\hline
$\asyOpV{\frac12}{0}{\hbar} \partial_{\phi_c}^{1} (\partial_{\psi_c} \partial_{ \bar \psi_c} )^{1} G^{\text{QQED}} \big|_{\substack{\phi_c=0\\\psi_c=0}}$&$\frac{\hbar^{-1}}{\sqrt{2\pi\hbar}}$&$1$&$-4$&$-3$&$-22$&$-188$&$-1968$\\
\hline
\end{tabular}
\subcaption{The first coefficients of the trivariate generating function $\asyOpV{\frac12}{0}{\hbar}\Gamma^{\text{QQED}}(\hbar, \phi_c, \psi_c)$.}
\end{subtable}
\caption{Effective action in quenched QED.}
\label{tab:GQQED}
\end{table}

The invariant charge is defined as
\begin{align*} \alpha(\hbar):= \left( \frac{ \partial_{\varphi_c}\partial_{\psi_c} \partial_{\bar \psi_c} G^\text{QED} \big|_{\substack{\varphi_c=0\\\psi_c=0}} } { \left( - \partial_{\psi_c} \partial_{\bar \psi_c} G^\text{QED}\big|_{\substack{\varphi_c=0\\\psi_c=0}} \right) } \right)^2, \end{align*}
and the calculation of the renormalization quantities works as before. Some coefficients are listed in Table \ref{tab:QQEDren}.
\begin{table}
\begin{subtable}[c]{\textwidth}
\centering
\tiny
\def\arraystretch{1.5}
\begin{tabular}{|c||c|c|c|c|c|c|}
\hline
&$\hbar_{\text{ren}}^{0}$&$\hbar_{\text{ren}}^{1}$&$\hbar_{\text{ren}}^{2}$&$\hbar_{\text{ren}}^{3}$&$\hbar_{\text{ren}}^{4}$&$\hbar_{\text{ren}}^{5}$\\
\hline\hline
$\hbar(\hbar_{\text{ren}})$&$0$&$1$&$-4$&$8$&$-28$&$-48$\\
\hline
$z_{|\psi_c|^2}(\hbar_{\text{ren}})$&$1$&$1$&$-1$&$-1$&$-7$&$-63$\\
\hline
$z_{\phi_c |\psi_c|^2}(\hbar_{\text{ren}})$&$1$&$-1$&$-1$&$-7$&$-63$&$-729$\\
\hline
\end{tabular}
\subcaption{Table of the first coefficients of the renormalization quantities in quenched QED.}
\end{subtable}
\begin{subtable}[c]{\textwidth}
\centering
\tiny
\def\arraystretch{1.5}
\begin{tabular}{|c||c||c|c|c|c|c|c|}
\hline
&prefactor&$\hbar_{\text{ren}}^{0}$&$\hbar_{\text{ren}}^{1}$&$\hbar_{\text{ren}}^{2}$&$\hbar_{\text{ren}}^{3}$&$\hbar_{\text{ren}}^{4}$&$\hbar_{\text{ren}}^{5}$\\
\hline\hline
$\left(\asyOpV{\frac12}{0}{\hbar_{\text{ren}}} \hbar \right)(\hbar_{\text{ren}})$&$e^{-2} \frac{\hbar^{0}}{\sqrt{2\pi\hbar}}$&$-2$&$20$&$-62$&$ \frac{928}{3}$&$ \frac{2540}{3}$&$ \frac{330296}{15}$\\
\hline
$\left(\asyOpV{\frac12}{0}{\hbar_{\text{ren}}} z_{|\psi_c|^2} \right)(\hbar_{\text{ren}})$&$e^{-2} \frac{\hbar^{0}}{\sqrt{2\pi\hbar}}$&$-1$&$6$&$4$&$ \frac{218}{3}$&$890$&$ \frac{196838}{15}$\\
\hline
$\left(\asyOpV{\frac12}{0}{\hbar_{\text{ren}}} z_{\phi_c |\psi_c|^2} \right)(\hbar_{\text{ren}})$&$e^{-2} \frac{\hbar^{-1}}{\sqrt{2\pi\hbar}}$&$-1$&$6$&$4$&$ \frac{218}{3}$&$890$&$ \frac{196838}{15}$\\
\hline
\end{tabular}
\subcaption{Table of the first coefficients of the asymptotics of the renormalization quantities in quenched QED.}
\end{subtable}
\caption{Renormalization constants in quenched QED.}
\label{tab:QQEDren}
\end{table}
The sequence generated by $1-z_{\varphi_c \psi_c \bar \psi_c}(\hbar_\text{ren})$, which enumerates the number of skeleton quenched QED vertex diagrams, was also given in \cite{broadhurst1999four}. It is entry \texttt{A049464} in the OEIS \cite{oeis}.
The asymptotics, read off from Table \ref{tab:QQEDren}, of this sequence are,
\begin{gather*} [\hbar_\text{ren}^n]( 1-z_{\varphi_c \psi_c \bar \psi_c}(\hbar_\text{ren})) \underset{n\rightarrow \infty}{\sim} e^{-2} (2n+1)!! \left( 1 - \frac{6}{2n+1} \right. \\ \left. - \frac{4}{(2n-1)(2n+1)} - \frac{218}{3}\frac{1}{(2n-3)(2n-1)(2n+1)} +\ldots \right), \end{gather*}
where we used $(2n-1)!! =\frac{2^{n+\frac12}}{\sqrt{2\pi}}\Gamma(n+\frac12)$. The first five coefficients of this expansion have been conjectured by Broadhurst \cite{davidexpansion} based on numerical calculations.
\subsubsection{Yukawa theory}

The partition function of Yukawa-theory in zero-dimensions is given by, 
\begin{align*} Z^{\text{Yuk}}(\hbar, j, \eta) &:= \int \frac{dx}{\sqrt{2 \pi \hbar}} e^{\frac{1}{\hbar} \left( - \frac{x^2}{2} + j x + \frac{|\eta|^2}{1-x} + \hbar \log \frac{1}{1-x} \right)} \end{align*}
Similarly, to the case of quenched QED, we can rewrite this, with $\chi(\hbar)= \sum_{n=0}^\infty (2n-1)!! \hbar^n$, as,
\begin{align*} Z^{\text{Yuk}}(\hbar, j, \eta) &= \frac{e^{\frac{j^2}{2\hbar}}}{1-j-\frac{|\eta|^2}{\hbar}} \chi\left(\frac{\hbar}{\left(1-j-\frac{|\eta|^2}{\hbar}\right)^2} \right) + \bigO(|\eta|^4), \end{align*}
where we expanded up to first order in $|\eta|^2$.

\begin{table}
\begin{subtable}[c]{\textwidth}
\centering
\tiny
\def\arraystretch{1.5}
\begin{tabular}{|c||c||c|c|c|c|c|c|}
\hline
&prefactor&$\hbar^{0}$&$\hbar^{1}$&$\hbar^{2}$&$\hbar^{3}$&$\hbar^{4}$&$\hbar^{5}$\\
\hline\hline
$\partial_j^{0} (\partial_{\eta} \partial_{ \bar \eta} )^{0} Z^{\text{Yuk}} \big|_{\substack{j=0\\\eta=0}}$&$\hbar^{0}$&$1$&$1$&$3$&$15$&$105$&$945$\\
\hline
$\partial_j^{1} (\partial_{\eta} \partial_{ \bar \eta} )^{0} Z^{\text{Yuk}} \big|_{\substack{j=0\\\eta=0}}$&$\hbar^{0}$&$1$&$3$&$15$&$105$&$945$&$10395$\\
\hline
$\partial_j^{2} (\partial_{\eta} \partial_{ \bar \eta} )^{0} Z^{\text{Yuk}} \big|_{\substack{j=0\\\eta=0}}$&$\hbar^{-1}$&$1$&$3$&$15$&$105$&$945$&$10395$\\
\hline
$\partial_j^{0} (\partial_{\eta} \partial_{ \bar \eta} )^{1} Z^{\text{Yuk}} \big|_{\substack{j=0\\\eta=0}}$&$\hbar^{-1}$&$1$&$3$&$15$&$105$&$945$&$10395$\\
\hline
$\partial_j^{1} (\partial_{\eta} \partial_{ \bar \eta} )^{1} Z^{\text{Yuk}} \big|_{\substack{j=0\\\eta=0}}$&$\hbar^{-1}$&$2$&$12$&$90$&$840$&$9450$&$124740$\\
\hline
\end{tabular}
\subcaption{The first coefficients of the trivariate generating function $Z^{\text{Yuk}}(\hbar, j, \eta)$.}
\end{subtable}
\begin{subtable}[c]{\textwidth}
\centering
\tiny
\def\arraystretch{1.5}
\begin{tabular}{|c||c||c|c|c|c|c|c|}
\hline
&prefactor&$\hbar^{0}$&$\hbar^{1}$&$\hbar^{2}$&$\hbar^{3}$&$\hbar^{4}$&$\hbar^{5}$\\
\hline\hline
$\asyOpV{\frac12}{0}{\hbar} \partial_j^{0} (\partial_{\eta} \partial_{ \bar \eta} )^{0} Z^{\text{Yuk}} \big|_{\substack{j=0\\\eta=0}}$&$\frac{\hbar^{0}}{\sqrt{2\pi\hbar}}$&$1$&$0$&$0$&$0$&$0$&$0$\\
\hline
$\asyOpV{\frac12}{0}{\hbar} \partial_j^{1} (\partial_{\eta} \partial_{ \bar \eta} )^{0} Z^{\text{Yuk}} \big|_{\substack{j=0\\\eta=0}}$&$\frac{\hbar^{-1}}{\sqrt{2\pi\hbar}}$&$1$&$0$&$0$&$0$&$0$&$0$\\
\hline
$\asyOpV{\frac12}{0}{\hbar} \partial_j^{2} (\partial_{\eta} \partial_{ \bar \eta} )^{0} Z^{\text{Yuk}} \big|_{\substack{j=0\\\eta=0}}$&$\frac{\hbar^{-2}}{\sqrt{2\pi\hbar}}$&$1$&$0$&$0$&$0$&$0$&$0$\\
\hline
$\asyOpV{\frac12}{0}{\hbar} \partial_j^{0} (\partial_{\eta} \partial_{ \bar \eta} )^{1} Z^{\text{Yuk}} \big|_{\substack{j=0\\\eta=0}}$&$\frac{\hbar^{-2}}{\sqrt{2\pi\hbar}}$&$1$&$0$&$0$&$0$&$0$&$0$\\
\hline
$\asyOpV{\frac12}{0}{\hbar} \partial_j^{1} (\partial_{\eta} \partial_{ \bar \eta} )^{1} Z^{\text{Yuk}} \big|_{\substack{j=0\\\eta=0}}$&$\frac{\hbar^{-3}}{\sqrt{2\pi\hbar}}$&$1$&$-1$&$0$&$0$&$0$&$0$\\
\hline
\end{tabular}
\subcaption{The first coefficients of the trivariate generating function $\asyOpV{\frac12}{0}{\hbar}Z^{\text{Yuk}}(\hbar, j, \eta)$.}
\end{subtable}
\caption{Partition function in Yukawa-theory.}
\label{tab:ZYuk}
\end{table}

It follows from $\asyOpV{\frac12}{0}{\hbar} \chi(\hbar) = \frac{1}{\sqrt{2\pi \hbar}}$ and the chain rule that,
\begin{align*} \asyOpV{\frac12}{0}{\hbar} Z^{\text{Yuk}}(\hbar, j, \eta) &= \frac{1}{\sqrt{2\pi \hbar}} e^{\frac{j}{\hbar}} \left( 1 + |\eta|^2 \frac{1-j}{\hbar^2} \right) + \bigO(|\eta|^4) \end{align*}
As in the case of quenched QED, the asymptotic expansions for each order in $j$ and $|\eta|$ up to $|\eta|^2$ of the disconnected diagrams are finite and therefore exact. Some coefficients are given in Table \ref{tab:ZYuk}.
\begin{table}
\begin{subtable}[c]{\textwidth}
\centering
\tiny
\def\arraystretch{1.5}
\begin{tabular}{|c||c|c|c|c|c|c|}
\hline
&$\hbar^{0}$&$\hbar^{1}$&$\hbar^{2}$&$\hbar^{3}$&$\hbar^{4}$&$\hbar^{5}$\\
\hline\hline
$\partial_j^{0} (\partial_{\eta} \partial_{ \bar \eta} )^{0} W^{\text{Yuk}} \big|_{\substack{j=0\\\eta=0}}$&$0$&$0$&$1$&$ \frac{5}{2}$&$ \frac{37}{3}$&$ \frac{353}{4}$\\
\hline
$\partial_j^{1} (\partial_{\eta} \partial_{ \bar \eta} )^{0} W^{\text{Yuk}} \big|_{\substack{j=0\\\eta=0}}$&$0$&$1$&$2$&$10$&$74$&$706$\\
\hline
$\partial_j^{2} (\partial_{\eta} \partial_{ \bar \eta} )^{0} W^{\text{Yuk}} \big|_{\substack{j=0\\\eta=0}}$&$1$&$1$&$6$&$50$&$518$&$6354$\\
\hline
$\partial_j^{0} (\partial_{\eta} \partial_{ \bar \eta} )^{1} W^{\text{Yuk}} \big|_{\substack{j=0\\\eta=0}}$&$1$&$2$&$10$&$74$&$706$&$8162$\\
\hline
$\partial_j^{1} (\partial_{\eta} \partial_{ \bar \eta} )^{1} W^{\text{Yuk}} \big|_{\substack{j=0\\\eta=0}}$&$1$&$6$&$50$&$518$&$6354$&$89782$\\
\hline
\end{tabular}
\subcaption{The first coefficients of the trivariate generating function $W^{\text{Yuk}}(\hbar, j, \eta)$.}
\end{subtable}
\begin{subtable}[c]{\textwidth}
\centering
\tiny
\def\arraystretch{1.5}
\begin{tabular}{|c||c||c|c|c|c|c|c|}
\hline
&prefactor&$\hbar^{0}$&$\hbar^{1}$&$\hbar^{2}$&$\hbar^{3}$&$\hbar^{4}$&$\hbar^{5}$\\
\hline\hline
$\asyOpV{\frac12}{0}{\hbar} \partial_j^{0} (\partial_{\eta} \partial_{ \bar \eta} )^{0} W^{\text{Yuk}} \big|_{\substack{j=0\\\eta=0}}$&$\frac{\hbar^{1}}{\sqrt{2\pi\hbar}}$&$1$&$-1$&$-2$&$-10$&$-74$&$-706$\\
\hline
$\asyOpV{\frac12}{0}{\hbar} \partial_j^{1} (\partial_{\eta} \partial_{ \bar \eta} )^{0} W^{\text{Yuk}} \big|_{\substack{j=0\\\eta=0}}$&$\frac{\hbar^{0}}{\sqrt{2\pi\hbar}}$&$1$&$-2$&$-3$&$-16$&$-124$&$-1224$\\
\hline
$\asyOpV{\frac12}{0}{\hbar} \partial_j^{2} (\partial_{\eta} \partial_{ \bar \eta} )^{0} W^{\text{Yuk}} \big|_{\substack{j=0\\\eta=0}}$&$\frac{\hbar^{-1}}{\sqrt{2\pi\hbar}}$&$1$&$-4$&$-3$&$-22$&$-188$&$-1968$\\
\hline
$\asyOpV{\frac12}{0}{\hbar} \partial_j^{0} (\partial_{\eta} \partial_{ \bar \eta} )^{1} W^{\text{Yuk}} \big|_{\substack{j=0\\\eta=0}}$&$\frac{\hbar^{-1}}{\sqrt{2\pi\hbar}}$&$1$&$-2$&$-3$&$-16$&$-124$&$-1224$\\
\hline
$\asyOpV{\frac12}{0}{\hbar} \partial_j^{1} (\partial_{\eta} \partial_{ \bar \eta} )^{1} W^{\text{Yuk}} \big|_{\substack{j=0\\\eta=0}}$&$\frac{\hbar^{-2}}{\sqrt{2\pi\hbar}}$&$1$&$-4$&$-3$&$-22$&$-188$&$-1968$\\
\hline
\end{tabular}
\subcaption{The first coefficients of the trivariate generating function $\asyOpV{\frac12}{0}{\hbar}W^{\text{Yuk}}(\hbar, j, \eta)$.}
\end{subtable}
\caption{Free energy in Yukawa-theory.}
\label{tab:WYuk}
\end{table}%
The free energy is defined as usual,
\begin{gather*} W^{\text{Yuk}}(\hbar, j, \eta) = \hbar \log Z^{\text{Yuk}}(\hbar, j, \eta) = \\ \frac{j^2}{2} + \hbar \log\frac{1}{1-j-\frac{|\eta|^2}{\hbar}} + \hbar \log \chi \left(\frac{\hbar}{\left( 1 - j - \frac{|\eta|^2}{\hbar} \right)^2 }\right) + \bigO(|\eta|^4), \end{gather*}
Its asymptotics are given by,
\begin{align*} \asyOpV{\frac12}{0}{\hbar} W^{\text{Yuk}}(\hbar, j, \eta) &=       \frac{\hbar}{\sqrt{2\pi \hbar}} e^{\frac{j - \frac{j^2}{2} }{\hbar}} \frac{1-j- \frac{|\eta|^2}{\hbar} \left(1-\frac{(1-j)^2}{\hbar} \right)}{\chi\left( \frac{\hbar}{(1-j - \frac{|\eta|^2}{\hbar})^2} \right)} + \bigO(|\eta|^4) \end{align*}
Some coefficients are given in Table \ref{tab:WYuk}.
The 1PI effective action is given by the Legendre transformation of $W^{\text{Yuk}}(\hbar, j, \eta)$. 
\begin{align*} G^{\text{Yuk}}(\hbar, \varphi_c, \psi_c) &= W^{\text{Yuk}}(\hbar, j, \eta) - j \varphi_c - \bar \eta \psi_c - \eta \bar \psi_c, \end{align*}
where $j, \eta, \varphi_c$ and $\psi_c$ are related by the equations,
$\varphi_c = \partial_j W^{\text{Yuk}}$ and $\psi_c = \partial_{\bar\eta} W^{\text{Yuk}}$.
Performing this Legendre transform is non-trivial in contrast to the preceding three examples, because we can have tadpole diagrams as in the case of $\varphi^3$-theory.

As for $\varphi^3$-theory, we define 
\begin{align*} \gamma^{\text{Yuk}}_0(\hbar) &:= \frac{G^{\text{Yuk}}\big|_{\substack{\varphi_c=0\\ \psi_c=0}}}{\hbar} = \frac{W^{\text{Yuk}}\big|_{\substack{j=j_0\\ \eta = 0}}}{\hbar} \\ &= \frac{j_0(\hbar)^2}{2 \hbar} + \log\frac{1}{1-j_0(\hbar)} + \log \chi \left(\frac{\hbar}{\left( 1 - j_0(\hbar) \right)^2 }\right), \end{align*}
where $j_0(\hbar)$ is the power series solution of 
$0 = \partial_j W^{\text{Yuk}}\big|_{j=j_0(\hbar)}$.
This gives 
\begin{gather*} G^{\text{Yuk}}(\hbar, \varphi_c, \psi_c) = \\ - \frac{\varphi_c^2}{2} +\hbar \log \frac{1}{1-\varphi_c} + \hbar \gamma^{\text{Yuk}}_0\left( \frac{\hbar}{(1-\varphi_c)^2} \right) + \hbar \frac{|\psi_c|^2(1-\varphi_c)}{j_0\left(\frac{\hbar}{(1-\varphi_c)^2} \right)} + \bigO(|\psi_c|^4). \end{gather*}
This equation also has a simple combinatorial interpretation: Every fermion line of a vacuum diagram can be dressed with an arbitrary number of external photons associated to a $\frac{1}{1-\varphi_c}$ factor. Every additional loop gives two additional fermion propagators. The first two terms compensate for the fact that there are no vacuum diagrams with zero or one loop.

The asymptotics result from an application of the $\asyOp$-derivative. Some coefficients are listed in Table \ref{tab:GYuk}. These sequences were also studied in \cite{kuchinskii1998combinatorial}. They obtained the constant, $e^{-1}$, and the linear coefficients $-\frac92$ and $-\frac{5}{2}$ using a combination of numerical and analytic techniques. 
\begin{table}
\begin{subtable}[c]{\textwidth}
\centering
\tiny
\def\arraystretch{1.5}
\begin{tabular}{|c||c|c|c|c|c|c|}
\hline
&$\hbar^{0}$&$\hbar^{1}$&$\hbar^{2}$&$\hbar^{3}$&$\hbar^{4}$&$\hbar^{5}$\\
\hline\hline
$\partial_{\phi_c}^{0} (\partial_{\psi_c} \partial_{ \bar \psi_c} )^{0} G^{\text{Yuk}} \big|_{\substack{\phi_c=0\\\psi_c=0}}$&$0$&$0$&$ \frac{1}{2}$&$1$&$ \frac{9}{2}$&$31$\\
\hline
$\partial_{\phi_c}^{1} (\partial_{\psi_c} \partial_{ \bar \psi_c} )^{0} G^{\text{Yuk}} \big|_{\substack{\phi_c=0\\\psi_c=0}}$&$0$&$1$&$1$&$4$&$27$&$248$\\
\hline
$\partial_{\phi_c}^{2} (\partial_{\psi_c} \partial_{ \bar \psi_c} )^{0} G^{\text{Yuk}} \big|_{\substack{\phi_c=0\\\psi_c=0}}$&$-1$&$1$&$3$&$20$&$189$&$2232$\\
\hline
$\partial_{\phi_c}^{0} (\partial_{\psi_c} \partial_{ \bar \psi_c} )^{1} G^{\text{Yuk}} \big|_{\substack{\phi_c=0\\\psi_c=0}}$&$-1$&$1$&$3$&$20$&$189$&$2232$\\
\hline
$\partial_{\phi_c}^{1} (\partial_{\psi_c} \partial_{ \bar \psi_c} )^{1} G^{\text{Yuk}} \big|_{\substack{\phi_c=0\\\psi_c=0}}$&$1$&$1$&$9$&$100$&$1323$&$20088$\\
\hline
\end{tabular}
\subcaption{The first coefficients of the trivariate generating function $G^{\text{Yuk}}(\hbar, \varphi_c, \psi_c)$.}
\end{subtable}
\begin{subtable}[c]{\textwidth}
\centering
\tiny
\def\arraystretch{1.5}
\begin{tabular}{|c||c||c|c|c|c|c|c|}
\hline
&prefactor&$\hbar^{0}$&$\hbar^{1}$&$\hbar^{2}$&$\hbar^{3}$&$\hbar^{4}$&$\hbar^{5}$\\
\hline\hline
$\asyOpV{\frac12}{0}{\hbar} \partial_{\phi_c}^{0} (\partial_{\psi_c} \partial_{ \bar \psi_c} )^{0} G^{\text{Yuk}} \big|_{\substack{\phi_c=0\\\psi_c=0}}$&$e^{-1} \frac{\hbar^{1}}{\sqrt{2\pi\hbar}}$&$1$&$- \frac{3}{2}$&$- \frac{31}{8}$&$- \frac{393}{16}$&$- \frac{28757}{128}$&$- \frac{3313201}{1280}$\\
\hline
$\asyOpV{\frac12}{0}{\hbar} \partial_{\phi_c}^{1} (\partial_{\psi_c} \partial_{ \bar \psi_c} )^{0} G^{\text{Yuk}} \big|_{\substack{\phi_c=0\\\psi_c=0}}$&$e^{-1} \frac{\hbar^{0}}{\sqrt{2\pi\hbar}}$&$1$&$- \frac{5}{2}$&$- \frac{43}{8}$&$- \frac{579}{16}$&$- \frac{44477}{128}$&$- \frac{5326191}{1280}$\\
\hline
$\asyOpV{\frac12}{0}{\hbar} \partial_{\phi_c}^{2} (\partial_{\psi_c} \partial_{ \bar \psi_c} )^{0} G^{\text{Yuk}} \big|_{\substack{\phi_c=0\\\psi_c=0}}$&$e^{-1} \frac{\hbar^{-1}}{\sqrt{2\pi\hbar}}$&$1$&$- \frac{9}{2}$&$- \frac{43}{8}$&$- \frac{751}{16}$&$- \frac{63005}{128}$&$- \frac{7994811}{1280}$\\
\hline
$\asyOpV{\frac12}{0}{\hbar} \partial_{\phi_c}^{0} (\partial_{\psi_c} \partial_{ \bar \psi_c} )^{1} G^{\text{Yuk}} \big|_{\substack{\phi_c=0\\\psi_c=0}}$&$e^{-1} \frac{\hbar^{-1}}{\sqrt{2\pi\hbar}}$&$1$&$- \frac{9}{2}$&$- \frac{43}{8}$&$- \frac{751}{16}$&$- \frac{63005}{128}$&$- \frac{7994811}{1280}$\\
\hline
$\asyOpV{\frac12}{0}{\hbar} \partial_{\phi_c}^{1} (\partial_{\psi_c} \partial_{ \bar \psi_c} )^{1} G^{\text{Yuk}} \big|_{\substack{\phi_c=0\\\psi_c=0}}$&$e^{-1} \frac{\hbar^{-2}}{\sqrt{2\pi\hbar}}$&$1$&$- \frac{17}{2}$&$ \frac{29}{8}$&$- \frac{751}{16}$&$- \frac{75021}{128}$&$- \frac{10515011}{1280}$\\
\hline
\end{tabular}
\subcaption{The first coefficients of the trivariate generating function $\asyOpV{\frac12}{0}{\hbar}G^{\text{Yuk}}(\hbar, \varphi_c, \psi_c)$.}
\end{subtable}
\caption{Effective action in Yukawa-theory.}
\label{tab:GYuk}
\end{table}

The calculation of the renormalization constants proceeds as in the other cases with the invariant charge defined as for QED. The first coefficients are listed in Table \ref{tab:Yukren}.

\begin{table}
\begin{subtable}[c]{\textwidth}
\centering
\tiny
\def\arraystretch{1.5}
\begin{tabular}{|c||c|c|c|c|c|c|}
\hline
&$\hbar_{\text{ren}}^{0}$&$\hbar_{\text{ren}}^{1}$&$\hbar_{\text{ren}}^{2}$&$\hbar_{\text{ren}}^{3}$&$\hbar_{\text{ren}}^{4}$&$\hbar_{\text{ren}}^{5}$\\
\hline\hline
$\hbar(\hbar_{\text{ren}})$&$0$&$1$&$-5$&$10$&$-36$&$-164$\\
\hline
$z_{\phi_c^2}(\hbar_{\text{ren}})$&$1$&$1$&$-1$&$-3$&$-13$&$-147$\\
\hline
$z_{|\psi_c|^2}(\hbar_{\text{ren}})$&$1$&$1$&$-1$&$-3$&$-13$&$-147$\\
\hline
$z_{\phi_c |\psi_c|^2}(\hbar_{\text{ren}})$&$1$&$-1$&$-3$&$-13$&$-147$&$-1965$\\
\hline
\end{tabular}
\subcaption{Table of the first coefficients of the renormalization quantities in Yukawa-theory.}
\end{subtable}
\begin{subtable}[c]{\textwidth}
\centering
\tiny
\def\arraystretch{1.5}
\begin{tabular}{|c||c||c|c|c|c|c|c|}
\hline
&prefactor&$\hbar_{\text{ren}}^{0}$&$\hbar_{\text{ren}}^{1}$&$\hbar_{\text{ren}}^{2}$&$\hbar_{\text{ren}}^{3}$&$\hbar_{\text{ren}}^{4}$&$\hbar_{\text{ren}}^{5}$\\
\hline\hline
$\left(\asyOpV{\frac12}{0}{\hbar_{\text{ren}}} \hbar \right)(\hbar_{\text{ren}})$&$e^{- \frac{7}{2}} \frac{\hbar^{-1}}{\sqrt{2 \pi\hbar}}$&$-2$&$26$&$- \frac{377}{4}$&$ \frac{963}{2}$&$ \frac{140401}{64}$&$ \frac{16250613}{320}$\\
\hline
$\left(\asyOpV{\frac12}{0}{\hbar_{\text{ren}}} z_{\phi_c^2} \right)(\hbar_{\text{ren}})$&$e^{- \frac{7}{2}} \frac{\hbar^{-1}}{\sqrt{2 \pi\hbar}}$&$-1$&$ \frac{15}{2}$&$ \frac{97}{8}$&$ \frac{1935}{16}$&$ \frac{249093}{128}$&$ \frac{42509261}{1280}$\\
\hline
$\left(\asyOpV{\frac12}{0}{\hbar_{\text{ren}}} z_{|\psi_c|^2} \right)(\hbar_{\text{ren}})$&$e^{- \frac{7}{2}} \frac{\hbar^{-1}}{\sqrt{2 \pi\hbar}}$&$-1$&$ \frac{15}{2}$&$ \frac{97}{8}$&$ \frac{1935}{16}$&$ \frac{249093}{128}$&$ \frac{42509261}{1280}$\\
\hline
$\left(\asyOpV{\frac12}{0}{\hbar_{\text{ren}}} z_{\phi_c |\psi_c|^2} \right)(\hbar_{\text{ren}})$&$e^{- \frac{7}{2}} \frac{\hbar^{-2}}{\sqrt{2 \pi\hbar}}$&$-1$&$ \frac{15}{2}$&$ \frac{97}{8}$&$ \frac{1935}{16}$&$ \frac{249093}{128}$&$ \frac{42509261}{1280}$\\
\hline
\end{tabular}
\subcaption{Table of the first coefficients of the asymptotics of the renormalization quantities in Yukawa-theory.}
\end{subtable}
\caption{Renormalization constants in Yukawa-theory.}
\label{tab:Yukren}
\end{table}

In \cite{molinari2005hedin,molinari2006enumeration} various low-order coefficients, which were obtained in this section, were enumerated using Hedin's equations \cite{hedin1965new}. The numerical results for the asymptotics given in \cite{molinari2006enumeration} agree with the analytic results obtained here. The $\Gamma(x)$ expansion of  \cite{molinari2006enumeration} corresponds to the generating function $\partial_{\varphi_c}\partial_{\psi_c} \partial_{\bar \psi_c} G^\text{Yuk} \big|_{\substack{\varphi_c=0\\\psi_c=0}}(\hbar)$ and the $\Gamma(u)$ expansion to the generating function $2-z_{\varphi_c \psi_c \bar \psi_c}(\hbar)$. The later is the generating function of all skeleton diagrams in Yukawa theory. Written traditionally the asymptotics are, 
\begin{gather*} [\hbar_\text{ren}^n]( 1-z_{\varphi_c \psi_c \bar \psi_c}(\hbar_\text{ren})) \underset{n\rightarrow \infty}{\sim} e^{-\frac{7}{2}} (2n+3)!! \left( 1 - \frac{15}{2} \frac{1}{2n+3} \right. \\ \left. -\frac{97}{8}\frac{1}{(2n+1)(2n+3)} - \frac{1935}{16}\frac{1}{(2n-1)(2n+1)(2n+3)} +\ldots \right). \end{gather*}

\section*{Acknowledgements}
Many thanks to Dirk Kreimer for steady encouragement and counseling. I wish to express my gratitude to David Broadhurst. He sparked my interest in asymptotic expansions in zero-dimensional field theory and I benefited greatly from our discussions.
\bibliographystyle{plainnat}
\bibliography{literature}
\appendix

\section{Singularity analysis of $x(y)$}
\label{sec:singanalysis}

\begin{proof}[Proof of Theorem \ref{thm:comb_int_asymp}]
Starting with Proposition \ref{prop:formalchangeofvar}
\begin{align} \label{eqn:appendixrefF} \mathcal{F}[\Sact(x)](\hbar)&= \sum_{n=0}^\infty \hbar^{n} (2n+1)!! [y^{2n+1}] x(y), \end{align}
we wish to compute the singular expansion of $x(y)$ at the removable singularity $(x,y)= (\tau_i, \rho_i)$ defined as the solution of $\frac{y^2}{2} = -\Sact(x)$ with positive linear coefficient, to obtain the asymptotics of $\mathcal{F}[\Sact(x)](\hbar)$.

Solving for $y$ and shifting the defining equation of the hyperelliptic curve to the point of the singularity gives,
\begin{align*} y &= \sqrt{-2\Sact(x) } \\ \rho_i - y &= \sqrt{-2\Sact(\tau_i) } - \sqrt{-2\Sact(x)} \\ 1 - \frac{y}{\rho_i} &= 1- \sqrt{\frac{\Sact(x)}{\Sact(\tau_i)}} \intertext{where $\rho_i = \sqrt{-2\Sact(\tau_i)}$. The right hand side is expected to be of the form $\approx \frac{\Sact''(\tau_i)}{2} (\tau_i-x)^2$ for $x\rightarrow \tau_i$. Setting $u_i=\sqrt{1-\frac{y}{\rho_i}}$ and $v_i=\tau_i-x$ gives} u_i &= \sqrt{1- \sqrt{\frac{\Sact(\tau_i-v_i)}{\Sact(\tau_i)}}}, \end{align*}
where the branch of the square root which agrees with the locally positive expansion around the origin must be taken.
We would like to solve this equation for $v_i$ to obtain the Puiseux expansion at the singular point:
\begin{align*} v_i(u_i) = \sum_{k=1}^\infty d_{i,k} u_i^k \end{align*}
The coefficients $d_{i,k}$ can be expressed using the Lagrange inversion formula,
\begin{align} \label{eqn:lagrange_singularity} d_{i,k} = [u_i^n] v_i(u_i) = \frac{1}{n} [v_i^{n-1}] \left( \frac{v_i}{\sqrt{1-\sqrt{\frac{\Sact(\tau_i-v_i)}{\Sact(\tau_i)}}}} \right)^n. \end{align}
The asymptotics of $[y^n] x(y)$ are given by the singular expansions around dominant singularities,
\begin{align*} [y^n] x(y) &\underset{n\rightarrow \infty}{\sim} [y^n] \sum_{i \in I} \sum_{k=0} d_{i,k} \left(1 - \frac{y}{\rho_i}\right)^\frac{k}{2}. \intertext{Expanding using the generalized binomial theorem and noting that only odd powers of $k$ contribute asymptotically gives,} [y^n] x(y) &\underset{n\rightarrow \infty}{\sim} \sum_{i\in I} (-1)^n \rho_i^{-n} \sum_{k=0} { k+\frac12 \choose n} d_{i,2k}. \end{align*}
Substituted into eq.\ \eqref{eqn:appendixrefF}, this results in,
\begin{align*} [\hbar^n] \mathcal{F}[S](\hbar) &= - (2n+1)!! \sum_{i \in I} \rho_i^{-2n-1} \sum_{k=0}^{R-1} { k+\frac12 \choose 2n+1 } d_{i,2k+1} \\ &+ \bigO\left(\sum_{i \in I}\left(\frac{2}{\rho_i^2}\right)^{n+\frac12}\Gamma(n-R)\right) \end{align*}
where the asymptotic behavior of the binomial ${ k+\frac12 \choose 2n+1 } \underset{n\rightarrow \infty}{\sim} \frac{C_k}{(2n+1)^{k+\frac32}}$ and the double factorial $(2n+1)!! = 2^{n+\frac32} \frac{\Gamma(n+\frac32)}{\sqrt{2 \pi}}$ were used to derive the form of the rest term. Substituting eq.\ \eqref{eqn:lagrange_singularity} results in 
\begin{align} \begin{split} \label{eqn:Fopsingular} [\hbar^n] \mathcal{F}[S](\hbar) &= \sum_{k=0}^{R-1} \frac{(2n-1)!!}{2} { k-\frac12 \choose 2n } \sum_{i \in I} \rho_i^{-2n-1} [v_i^{2k}] \phi_i(v_i)^{2k+1} \\ &+ \bigO\left(\sum_{i \in I}(-\Sact(\tau_i))^{n}\Gamma(n-R)\right) \end{split} \end{align}
where $\phi_i(v_i):= \frac{-v_i}{\sqrt{1-\sqrt{\frac{\Sact(\tau_i-v_i)}{\Sact(\tau_i)}}}}$.
It is easily checked by the reflection and duplication formulas for the $\Gamma$-function that
\begin{align} \label{eqn:gamma_binom_ident} \frac{(2n-1)!!}{2} { k - \frac12 \choose 2n } &= \frac{(-1)^k 2^{n-k} \Gamma(k+\frac12)}{(2\pi)^{\frac32}} \frac{\Gamma(n-\frac{k}{2} +\frac14) \Gamma(n-\frac{k}{2}+\frac34) }{ \Gamma(n+1) }. \end{align}
The following lemma by Paris \cite[Lemma 1]{paris1992smoothing},
\begin{align*} \frac{\Gamma(n+a)\Gamma(n+b)}{n!} &= \sum_{m=0}^{R-1} (-1)^m \frac{\risefac{(1-a)}{m} \risefac{(1-b)}{m}}{m!} \Gamma(n+a+b-1-m) && \\ &+ \bigO(\Gamma(n+a+b-1-R)) &&\forall R \geq 0, \end{align*}
can be used to expand the product of $\Gamma$ functions. The expression $\risefac{a}{n}=\frac{\Gamma(n+a)}{\Gamma(a)}$ denotes the rising factorial. 
Applying this to eq. \eqref{eqn:gamma_binom_ident} gives,
\begin{gather*} \frac{(2n-1)!!}{2} { k - \frac12 \choose 2n } = \\ \frac{(-1)^k 2^{n-k}}{(2\pi)^\frac32} \sum_{m=0}^{R-k-1} (-1)^m 2^{-2m} \frac{\Gamma(k+\frac12+2m)} { m!} \Gamma(n-k-m) + \bigO(\Gamma(n-R)) \\ = \frac{2^{n}}{2\pi} \sum_{m=0}^{R-k-1} (-1)^{m+k} 2^{-3m-2k-\frac12} (2(m+k)+1)!! \\ \times { 2m + k - \frac12 \choose m } \Gamma(n-k-m) + \bigO(\Gamma(n-R)) \end{gather*}
This can be substituted into eq.\ \eqref{eqn:Fopsingular}:
\begin{gather*} [\hbar^n] \mathcal{F}[S](\hbar) = \sum_{k=0}^{R-1} \frac{2^{n}}{2\pi} \sum_{m=0}^{R-k-1} (-1)^{m+k} 2^{-3m-2k-\frac12} (2(m+k)+1)!! \\ \times { 2m + k - \frac12 \choose m }\Gamma(n-k-m) \sum_{i \in I} \rho_i^{-2n-1} [v_i^{2k}] \phi_i(v_i)^{2k+1} \\ + \bigO\left( \sum_{i \in I}(-\Sact(\tau_i))^{n} \Gamma(n-R)\right) \\ = \sum_{m=0}^{R-1} \frac{2^{n}}{2\pi} (-1)^{m} 2^{-2m-\frac12} (2m+1)!! \Gamma(n-m) \\ \times \sum_{i \in I} \rho_i^{-2n-1} \sum_{k=0}^{m} 2^{-k} { m + k - \frac12 \choose k } [v_i^{2(m-k)}] \phi_i(v_i)^{2(m-k)+1} \\ + \bigO\left(\sum_{i \in I}(-\Sact(\tau_i))^{n}\Gamma(n-R)\right) \end{gather*}
The inner sum evaluates to,
\begin{align*} &\sum_{k=0}^{m} 2^{-k} { m + k - \frac12 \choose k } [v_i^{2(m-k)}] \phi_i(v_i)^{2(m-k)+1} \\ &= [v_i^{2m}] \phi_i(v_i)^{2m+1} \sum_{k=0}^{m} { m + k - \frac12 \choose k } 2^{-k}\left(\frac{v_i}{\phi_i(v_i)} \right)^{2k} \\ &= [v_i^{2m}] \frac{\phi_i(v_i)^{2m+1}}{\left(1-\frac12 \left(\frac{v_i}{\phi_i(v_i)} \right)^2\right)^{m+\frac12}} = [v_i^{2m}] \left( \frac{\phi_i(v_i)}{\sqrt{1-\frac12 \left(\frac{v_i}{\phi_i(v_i)} \right)^2}}\right)^{2m+1} \\ & = 2^{m+\frac12} [v_i^{2m}] \left( \frac{-v_i}{\sqrt{ 1 - \frac{\Sact(\tau_i-v_i)}{\Sact(\tau_i)} } } \right)^{2m+1} \\ & = (-2\Sact(\tau_i))^{m+\frac12} [v_i^{2m}] \left( \frac{-v_i}{\sqrt{ \Sact(\tau_i-v_i) -\Sact(\tau_i) } } \right)^{2m+1} \end{align*}
Therefore,
\begin{align*} [\hbar^n] \mathcal{F}[S](\hbar) &= \frac{1}{2\pi} \sum_{m=0}^{R-1} (-1)^{m} (2m+1)!! \Gamma(n-m) \\ &\times \sum_{i \in I} (-\Sact(\tau_i))^{n-m} [v_i^{2m}] \left( \frac{v_i}{\sqrt{ 2\Sact(\tau_i+v_i) -2\Sact(\tau_i) } } \right)^{2m+1} \\ &+ \bigO\left(\sum_{i \in I} (-\Sact(\tau_i))^{n}\Gamma(n-R)\right) \\ &= \frac{1}{2\pi i} \sum_{m=0}^{R-1} (2m+1)!! \Gamma(n-m) \\ &\times \sum_{i \in I} (-\Sact(\tau_i))^{n-m} [v_i^{2m}] \left( \frac{v_i}{\sqrt{ - 2 \left( \Sact(\tau_i+v_i) -\Sact(\tau_i) \right) } } \right)^{2m+1} \\ &+ \bigO\left(\sum_{i \in I} (-\Sact(\tau_i))^{n}\Gamma(n-R)\right) \end{align*}
which proves the theorem after using the Lagrange inversion formula and Proposition \ref{prop:formalchangeofvar}.
\end{proof}

\end{document}